\documentclass[11pt,english]{article}
\usepackage[T1]{fontenc}
\usepackage[utf8]{inputenc}
\usepackage{lmodern}
\usepackage{amssymb,amsmath,amsthm,amsfonts,bbm,dsfont}
\usepackage[pdftex]{graphicx}
\usepackage{float,placeins,flafter}
\usepackage{enumerate}
\usepackage{natbib}
\usepackage{multirow}
\usepackage{verbatim}
\usepackage{xpatch}

\makeatletter
\xpatchcmd{\@thm}{\thm@headpunct{.}}{\thm@headpunct{}}{}{}
\makeatother
\numberwithin{equation}{section}
\newtheorem{theorem}{Theorem}

\newtheorem{assumption}{Assumption}
\newtheorem*{assumption_nonumber}{Assumption}

\newtheorem{definition}{Definition}

\newtheorem{lemma}{Lemma}
\theoremstyle{definition}

\newcommand{\convind}{\overset{d}{\longrightarrow}}
\newcommand{\norm}[1]{\left \|#1\right \|}

\newcommand{\eps}{\varepsilon}

\begin{document}

\title{Nonparametric Difference-in-Differences in Repeated Cross-Sections
with Continuous Treatments\thanks{%
First draft: July 8, 2013. We have benefited from helpful comments from Han
Hong (co-editor), an anonymous associate editor, two anonymous referees,
Arne Uhlendorf and seminar participants at Boston College, Bristol
University, Chicago, Crest, GRIPS, TSE and UCL. The paper was partly written
when the first author was visiting Boston College and PSE, which the author
thanks for their hospitality. The usual disclaimer applies. While not at the
core of this paper, some elements are taken from the earlier draft
\textquotedblleft On the role of time in nonseparable panel data
models\textquotedblright\ by Hoderlein and Sasaki, which is now retired.}}
\author{
\begin{tabular}{ccc}
Xavier D'Haultf\oe uille\thanks{%
Xavier D'Haultf\oe uille: CREST, 5 avenue Henry le Chatelier 91120 Palaiseau
France, email: \texttt{xavier.dhaultfoeuille@ensae.fr}.} & Stefan Hoderlein%
\thanks{%
Stefan Hoderlein: Emory University, Department of Economics, 1602 Fishburne
Dr., Atlanta, GA 30322, USA, email: \texttt{stefan.hoderlein@emory.edu}.} &
Yuya Sasaki\thanks{%
Yuya Sasaki: Vanderbilt University, Department of Economics, 2301 Vanderbilt
Place Nashville, TN 37235-1819, USA. email: \texttt{%
yuya.sasaki@vanderbilt.edu.}} \\
CREST & Emory University & Vanderbilt University%
\end{tabular}
}
\date{First version: July 8, 2013\\
This version: February 4, 2022}
\maketitle

\begin{abstract}
\setlength{\baselineskip}{0.45cm}

This paper studies the identification of causal effects of a continuous
treatment using a new difference-in-difference strategy. Our approach allows
for endogeneity of the treatment, and employs repeated cross-sections. It
requires an exogenous change over time which affects the treatment in a
heterogeneous way, stationarity of the distribution of unobservables and a
rank invariance condition on the time trend. On the other hand, we do not
impose any functional form restrictions or an additive time trend, and we
are invariant to the scaling of the dependent variable. Under our
conditions, the time trend can be identified using a control group, as in
the binary difference-in-differences literature. In our scenario, however,
this control group is defined by the data. We then identify average and
quantile treatment effect parameters. We develop corresponding nonparametric
estimators and study their asymptotic properties. Finally, we apply our
results to the effect of disposable income on consumption. \bigskip \newline
\textbf{Keywords:} identification, repeated cross-sections, nonlinear
models, continuous treatment, random coefficients, endogeneity,
difference-in-differences. \bigskip \newline
${}$
\end{abstract}

\section{Introduction}

Differences-in-Differences (DID) is arguably one of the most popular methods
for policy evaluation. In its standard version, it allows to identify the
causal effect of a binary treatment on a given outcome, even when units are
not allocated randomly to the treatment. The idea is to compare the
evolution of the average outcome of the treatment group, which receives a
treatment after a certain date, with that of the control group, which
remains untreated. The DID strategy builds on the so-called common trend
assumption, viz. the assumption that the changes in $Y(0)$, the potential
outcome absent the treatment, are identical between the treatment and the
control group. One way to see this condition is that treatment changes
should be exogenous in that they are not related to changes in $Y(0)$. Under
common trends, the average time trend on $Y(0)$ can be identified using the
control group, as this group does not experience the effect of the
treatment. Once this time trend is accounted for, we can identify the
average treatment effect by a simple before-after comparison on the
treatment group.

\medskip A crucial limitation of the standard DID framework is that the
treatment is required to be binary. Yet, in many cases, units experience
various treatment intensities, and not just a zero-one treatment. Examples
include, among many others, unemployment benefits, specific public
expenditures (e.g., hospitals expenditure per capita, teacher's wages),
changes in prices, in income etc. Usual solutions in such cases are to
either consider a linear model or to discretize the treatment. Both
solutions are problematic. In the first case, the model cannot account for,
e.g., unobservable terms affecting both treatment intensity and treatment
effects, effectively assuming that the treatment has the same effect for
every level of treatment intensity. In the second case, discretization
introduces arbitrariness and leads to an information loss. E.g., after
discretization, even vastly different income changes would be assumed to
have the same effect, because all that matters is the fact that they change.
This makes the DID framework not useful to study, e.g., causal marginal
propensities to consume out of income, because individuals at low levels of
income are more likely to face liquidity constraints than individuals at
high levels, and the size of the income change arguably matters
\citep[see,
e.g.,][]{hsieh2003consumers}.

\medskip In this paper, we propose a solution that circumvents these issues.
Specifically, we show identification of several treatment effect parameters,
allowing for nonlinear and heterogeneous effects of the treatment without
imposing functional form restrictions (e.g., linearity), nor any
discretization. We also allow for heterogeneous time trends on potential
outcomes. This is important, since assuming the same time trends for all
units may be overly restrictive. Firms or individuals with different
productivities may be affected differently by macroeconomic shocks, for
instance.

\medskip The idea behind our identification strategy retains the spirit of
the DID approach. We use the fact that the distribution of the treatment, in
our case a continuous random variable, changes over time for some exogenous
reason, yet some units remain at the same level of treatment. In our
approach, the latter units form the control group, which allows us to
identify the (heterogeneous) time trend on potential outcomes. Once this
time trend has been removed, the distribution of the appropriately modified
potential outcome does not vary over time any longer. Then, any difference
over time in the distribution of the modified, observed outcome should be
solely due to the treatment. With this insight, we can identify causal
effects of the treatment.

\medskip To make this strategy operational, we rely on three main
assumptions. First, we assume that units sharing the same rank in the
distribution of the treatment at two different periods have the same
distribution of the unobservables governing potential outcomes. This
assumption is related to the aforementioned exogeneity of the change in the
distribution of the treatment over time. It first ensures that groups,
similar to the control and treatment group with binary treatments, may be
defined through the ranking of the treatment. With this construction of
groups, the condition becomes almost the same as Assumptions 1 and 3 in \cite%
{Athey06}, on which our paper builds to identify (heterogeneous) time
trends. Second, we assume that a unit with stable unobservables in two
different periods will also have the same ranking in the distributions of
potential outcome in these two periods. Again, this assumption is identical
to Assumption 2 in \cite{Athey06}. Third, we suppose that the change in the
distribution of the treatment is heterogeneous, in the sense that the
cumulative distribution functions (cdfs) of the treatment variable between
the two periods cross. This crossing point defines the control group and
allows us to identify the heterogeneous time trends.

\medskip Despite some similarities with the nonlinear
difference-in-differences setting of \cite{Athey06}, our continuous
treatment set-up exhibits several important distinct features. First, \cite%
{Athey06} focus on the binary treatment case, while we consider a continuous
treatment. Second, our control group is determined by the data rather than
fixed ex ante. While our paper also shares some similarities with the paper
by \cite{Chaisemartin14}, the framework and identification strategy is
nonetheless very different: In particular, \cite{Chaisemartin14} focus on
binary (or ordered) treatments. With a continuous treatment, their strategy
would require to have a control group for which the whole distribution of
the treatment variable would remain unchanged over time, an assumption
unlikely to be satisfied in practice. In contrast, we only require the
distribution of the treatment to change in such a way that there exists a
crossing point. A change in the mean and the variance of a normal
distribution, for instance, satisfies this requirement.

\medskip We consider several extensions to our main setup. First, we show
how covariates can be included into our analysis. Second, we establish that
our model extends in a straightforward way to multidimensional continuous
treatments. Third, we show that while a number of parameters cannot be point
identified, basically because time only provides us with limited exogenous
variations, they can be partially identified using weak local curvature
conditions. Finally, we prove that under functional form restrictions that
still allow for ample heterogeneity, we can point identify all marginal
effects.

\medskip Based on this extensive identification analysis, we also develop
nonparametric sample counterparts estimators. While our estimators of the
average and quantile effects involve several nonparametric steps, each one
of these steps is straightforward to perform. We show the asymptotic
normality of the estimators. When the \textquotedblleft control
group\textquotedblright\ corresponds to a single point of support of the
continuous treatment, the estimators are not root-n consistent, but converge
at standard univariate nonparametric rates.

\medskip Finally, we apply our methodology to analyze the marginal
propensity to consume in the US. We exploit for that purpose a change in the
schedule of the Earned Income Tax Credit (EITC) between 1987 and 1989. We
argue that this change generates exactly the crossing condition we require.
Applying our method, we obtain an estimated time trend that displays
heterogeneity, underlying the need to go beyond mere additive time trends.
Moreover, our estimates of the marginal effects suggest that low-income
individuals increase substantially their consumption (by around 50\%), while
medium-income individuals would not significantly adjust their consumption.
This is in line with many findings in the literature, see e.g., \cite%
{johnson2006household} and \cite{kaplan2014model}.

\medskip The paper is organized as follows. In Section 2, we introduce the
model formally, discuss the parameters of interest and provide our main
identification results. The extensions considered above are discussed in
Section 3. Section 4 is devoted to estimation. Section 5 presents the
application, and Section 6 concludes. All proofs are gathered in the
appendix.

\section{Model and Main Identification Results}

\subsection{Assumptions}

We consider a potential outcome framework with a continuous treatment. The
potential outcome at period $t$, corresponding to a treatment $x\in \mathcal{%
X}\subset \mathbb{R}$, is denoted by $Y_{t}(x)$, with $Y_t(x)\in\mathbb{R}$.
We observe, at each period $t\in \{1,...,T\}$, the actual treatment $X_{t}$
and the corresponding outcome, $Y_{t}\equiv Y_{t}(X_{t})$. We are
particularly interested in the average and quantile treatment on the treated
effects:
\begin{eqnarray*}
\Delta ^{ATT}(x,x^{\prime }) &\equiv &E\left[ Y_{T}(x^{\prime
})-Y_{T}(x)|X_{T}=x\right] , \\
\Delta ^{QTT}(p,x,x^{\prime }) &\equiv &F_{Y_{T}(x^{\prime
})|X_{T}}^{-1}(p|x)-F_{Y_{T}(x)|X_{T}}^{-1}(p|x),
\end{eqnarray*}
for any $x$ and $x^{\prime }$ in the support $\text{Supp}(X_T)$ of $X_{T}$.
Here, $F_{A|B}(a|b)$ denotes the conditional cdf of a random variable $A$ at
$a$, given that a random vector $B$ takes the value $b$, and $%
F_{A|B}^{-1}(\tau |b)$ denotes its inverse, the $\tau $-conditional quantile
function. We henceforth focus on the effects at period $T,$ because they are
the most natural to compute in general, but we can identify similar effects
at any date.

\medskip The main issue in identifying the parameters above is endogeneity
of the actual treatment, i.e., $X_{t}$ may depend on $(Y_{t}(x))_{x\in
\mathcal{X}}$. In such a case, naive estimators do not coincide with the
average and quantile treatment effects defined above. For instance, $%
E(Y_{T}|X_{T}=x^{\prime })-E(Y_{T}|X_{T}=x)\neq \Delta ^{ATT}(x,x^{\prime })$%
. Our idea for identifying these causal parameters, then, is to use
exogenous changes in $X_{t}$ (due to, e.g., a policy change), and apply a
difference-in-difference type strategy. To make this idea operational, we
restrict the way time affects both observed and unobserved variables by
imposing three main restrictions. The first restriction is a stationarity
condition on the observed and unobserved determinants of the outcome. The
second restriction limits the way time affects the outcome itself. The third
restriction affects the way the distribution of $X_{t}$ changes over time.
We discuss them in turn using the notation $V_{t}=F_{X_{t}}(X_{t})$ to
denote the rank of an individual in the distribution of the treatment.

\begin{assumption}
\label{hyp:stationarity} (stationarity of unobservables) For all $t \in
\{1,...,T\}$, $Y_t(x)=g_t(U_t(x))$ where for all $(x,v)\in \mathcal{X}\times
[0,1]$ the distribution of the unobserved variable $U_t(x)|V_t=v$ does not
depend on $t$.
\end{assumption}

We can interpret this assumption as follows. First, it defines implicitly
groups, similar to control and treatment groups with binary treatments,
through the rank variable $V_t$. Then, we assume that within each group,
unobserved terms related to potential outcomes have a time-invariant
distribution. This latter condition is similar to Assumptions 3.1 and 3.3 in
\cite{Athey06}, where the authors also assume that within both the control
and the treatment group the distribution of the unobserved term related to $%
Y_{t}(0)$ is constant over time.

\medskip Importantly, Assumption \ref{hyp:stationarity} does not restrict
the cross-sectional dependence between $U_{t}(x)$ and $V_{t}$, which is at
the core of the endogeneity problem we face in this scenario. On the other
hand, it rules out changes in the type of endogeneity, as the distribution
of $(U_{t}(x),V_{t})$ is supposed to be time invariant. In our application
below, $X_{t}$ corresponds to disposable income. The tax rate affects
disposable income, but a change in the tax rate is unlikely to change the
ranking of individuals in the income distribution, holding other
characteristics constant (e.g., the number of household members). In other
applications, this condition may be more restrictive. We discuss this point
further when we draw a parallel with instrumental variable models in Section %
\ref{subs:compar_IV} below.

\medskip The following assumption specifies the second requirement mentioned
above:

\begin{assumption}
(rank invariance on the time trend) For all $(x,t)\in \mathcal{X}\times
\{1,...,T\}$, $U_t(x)\in \mathbb{R}$ and $g_t$ is strictly increasing.
Without loss of generality, we let $g_T(y)=y$ for all $y\in\text{Supp}(Y_T)$%
. \label{hyp:direct_effect}
\end{assumption}

Assumption \ref{hyp:direct_effect} is again materially identical to
Assumption 3.2 in \cite{Athey06}. Combined with Assumption \ref%
{hyp:stationarity}, it states that an individual that has the same
unobservable in two different periods (i.e., $U_{t}(x)\equiv U_{t^{\prime
}}(x)$) will also have the same ranking in the distributions of potential
outcomes $Y_{t}(x)$ and $Y_{t^{\prime }}(x)$ in the same two periods.
Assumption \ref{hyp:direct_effect} generalizes the standard translation
model $g_{t}(u)=\delta _{t}+u$ to allow for heterogeneous time trends. This
can be important in some applications. For instance, macroeconomic shocks
may have different effects on high- and low-wage earners. Note that given
the strict monotonicity condition, we can always make the normalization $%
g_{T}(y)=y$, by just redefining $U_{t}(x)$ as $g_{T}(U_{t}(x)) $ and $%
g_{t}(y)$ as $g_{t}\circ g_{T}^{-1}(y)$.

\begin{assumption}
(crossing points) For all $t\in \{1,...,T-1\}$, there exists $x_{t}^{\ast
}\in \mathbb{R}$ such that $F_{X_t}(x_{t}^{\ast })=F_{X_T}(x_{t}^{\ast })\in
(0,1) $. \label{hyp:crossing}
\end{assumption}

Contrary to Assumptions \ref{hyp:stationarity}-\ref{hyp:direct_effect},
Assumption \ref{hyp:crossing} only involves observables, and is therefore
directly testable in the data. It means, roughly speaking, that the
exogenous change (induced by, e.g., a policy change) affects individuals'
treatment in a heterogeneous way. Requiring time to have a heterogeneous
effect on the treatment is also required in the usual
difference-in-difference strategy, and a similar condition is required in
fuzzy settings considered by \cite{Chaisemartin14}.

\medskip Note that $x_{t}^{\ast }$ can be identified and estimated using the
data. We consider such an estimator in Section \ref{sec:estimation} below.
However, sometimes the value of the crossing point may also be inferred from
the design of the policy change. Comparing the theoretical crossing point
with the crossing point obtained from the data then constitutes a check for
the hypothesis that the policy change has not changed the distribution of
the unobservables. We refer to the application below for more details about
this.

\medskip Two additional remarks on Assumption \ref{hyp:crossing} are in
order. First, this assumption holds if $F_{X_{t}}$ remains constant with $t$%
. In this case, however, we identify only the trivial parameters $\Delta
^{ATT}(x,x)=\Delta ^{QTT}(p,x,x)=0$. Second, we assume for simplicity
crossings between the cdf of $X_{T}$ and all other cdfs, but actually, $T-1$
crossings are sufficient, provided that we can \textquotedblleft
relate\textquotedblright\ them to each other, for instance if the cdf of $%
X_{t}$ crosses that of $X_{t+1}$ for $1\leq t<T$. Also, with only one
crossing between $F_{X_{s}}$ and $F_{X_{t}}$, we still identify some
treatment effects at periods $s$ or $t$, following the same logic as in
Section \ref{sub:ident} below. So even if Assumption \ref{hyp:crossing} does
not make this apparent, adding periods help because it increases the odds of
having at least one crossing point, which is sufficient for identifying some
causal parameters.

\medskip The last assumption we impose is a regularity condition:

\begin{assumption}
(regularity conditions) For all $t\in\{1,...,T\}$, $E(|Y_t|)<\infty$ and $%
F_{X_t}$ is continuous on $\text{Supp}(X_t)$, which is an interval included
in $\mathcal{X}$. For all $x^{\prime }\in \text{Supp}(X_t)$ and $u\in \text{%
Supp}(U_t(x^{\prime }))$, there exist versions of $E[Y_t|X_t]$ and $%
P^{U_t(x^{\prime }) |V_t}$ such that $x\mapsto E[Y_t|X_t=x]$ and $v \mapsto
F_{U_t(x^{\prime })|V_t}(u|v)$ are continuous. \label{hyp:model}
\end{assumption}

The continuity conditions are mild, yet important to define properly
conditional expectations or cdfs (e.g., $E[Y_{t}|X_{t}=x_{t}^{\ast }]$ or $%
F_{Y_{t}|X_{t}}(y|x_{t}^{\ast })$).

\subsection{Examples}

To better understand the types of data generating processes which our
assumptions permit, we consider two examples of workhorse models.

\subsubsection{Simple Linear Systems}

\label{subs:linear} Let us suppose that
\begin{align}
Y_{t}(x)& =\alpha _{t}+x\beta +U_{t} ,  \label{eq:simple_linear} \\
X_{t}& =\gamma _{t}+\delta _{t}\eta _{t},  \label{eq:FS_linear}
\end{align}%
where $(\alpha _{t},\beta ,\gamma _{t},\delta _{t})$ are constants and the
marginal distribution of $(U_{t},\eta _{t})$ is assumed constant over time.
Suppose also that $\text{Supp}(\eta _{t})=\mathbb{R}$ and $\delta _{t}\neq
\delta _{T}$ for all $t\neq T$. Then, Assumptions \ref{hyp:stationarity}-\ref%
{hyp:crossing} hold with $U_{t}(x)=x\beta +U_{t}$, $g_{t}(u)=\alpha _{t}+u$
and $x_{t}^{\ast }=(\gamma _{t}-\gamma _{T})/(\delta _{T}-\delta _{t})$.
Assumption \ref{hyp:model} holds under mild restrictions on the distribution
of $(U_{t},\eta _{t})$. Note that the model allows for any dependence
between $U_{t}$ and $\eta _{t}$. Thus, $X_{t}$ is endogenous in general in
the outcome equation, and we cannot recover $\beta $ directly by the OLS.
Note, moreover, that we cannot use time as an instrumental variable in the
outcome equation either, because it has a direct effect on $Y_{t}$, so none
of the standard tools work.

\medskip As mentioned above, if the policy change has a pure location effect
on $X_t$, so that $\delta_t=\delta_T$ for all $t$, then Assumption \ref%
{hyp:crossing} is not satisfied. We require to have individuals unaffected
by the change, and this holds with a change in scale in \eqref{eq:FS_linear}.

\medskip Note that we did not impose any condition on the dependence between
$(U_{s},\eta _{s})$ and $(U_{t},\eta _{t})$. Hence, the model allows for any
form of serial dependence of the unobservables. On the other hand, models
with a lagged dependent variable are typically ruled out by our assumptions.
To see this, suppose that we replace \eqref{eq:simple_linear} by
\begin{equation}
Y_{t}(x)=\alpha _{t}+x\beta +\rho Y_{t-1}+U_{t}.  \label{eq:AR_model} \\
\end{equation}
Then,
\begin{equation*}
Y_{t}(x)=\widetilde{\alpha }_{t}+\beta x+\widetilde{U}_{t}
\end{equation*}%
with $\widetilde{\alpha }_{t}=\alpha +\sum_{k=1}^{\infty }\rho ^{k}\left[
\alpha _{t-k}+\beta \gamma _{t-k}\right] $ and $\widetilde{U}%
_{t}=U_{t}+\sum_{k=1}^{\infty }\rho ^{k}[\beta \delta _{t-k}\eta
_{t-k}+U_{t-k}]$. Therefore, the distribution of $\widetilde{U}_{t}$ depends
on $t$ in general, unless $\rho =0$. Then,
\begin{equation*}
F_{Y_{t}(x)}^{-1}\circ F_{Y_{t^{\prime }}(x)}(y)=\widetilde{\alpha }%
_{t}+\beta x+F_{\widetilde{U}_{t}}^{-1}\circ F_{\widetilde{U}_{t^{\prime
}}}(y-\alpha _{t^{\prime }}^{\prime }-\beta x),
\end{equation*}%
which depends on $x$ in general when $\beta\neq 0$ and $\rho\neq 0$. On the
other hand, Assumptions \ref{hyp:stationarity}-\ref{hyp:direct_effect} imply
that for any $(t,t^{\prime })$, $F_{Y_{t}(x)}^{-1}\circ F_{Y_{t^{\prime
}}(x)}(y)$ does not depend on $x$. In other words, Assumptions \ref%
{hyp:stationarity}-\ref{hyp:direct_effect} are violated in general when $%
\rho \neq 0$.

\medskip This feature is not specific to our assumptions. A similar issue
arises in the standard difference-in-differences setup. To see this,
consider model \eqref{eq:AR_model} again, but now with a binary treatment,
for which $X_{t}=0$ for all $t\leq 1,$ and $X_{2}=G$, the dummy of being in
the treatment group as opposed to the control group in the second period.
Suppose, moreover, that $E(U_{t}|G)$ does not depend on $t$. In the case of $%
\rho =0$, the common trend condition is satisfied with $%
E(Y_{2}(0)|G)-E(Y_{1}(0)|G)=\alpha _{2}-\alpha _{1}$. But if $\rho \neq 0$,
then the common trends assumptions fails to hold in general, since
\begin{equation*}
E(Y_{2}(0)|G)-E(Y_{1}(0)|G=0)=\alpha _{2}-\alpha _{1}+\sum_{k=1}^{\infty
}\rho ^{k}\alpha _{1-k}+\frac{\rho }{1-\rho }E(U_{1}|G),
\end{equation*}%
which depends on $G$.

\subsubsection{Quantile Regression Type Models}

The previous model does not allow for heterogeneous treatment effects or
heterogeneous time trends on potential outcomes. We may, however, analyze
models with heterogeneous features as they are compatible with our
assumptions. The following model exemplifies this:
\begin{align*}
Y_{t}& =f_{t}\left[ \alpha (U_{t})+X_{t}\beta (U_{t})\right] \\
X_{t}& =h_{t}(\eta _{t}),
\end{align*}%
where we assume that the marginal distribution of $(U_{t},\eta _{t})$ is
constant over time, $\text{Supp}(\eta _{t})=\mathbb{R}$, and for all $t\neq
T $, there exists $e_{t}$ such that $h_{t}(e_{t})=h_{T}(e_{t})$. We also
assume that $f_{t}$ and $e\mapsto \alpha (e)+x\beta (e)$ are strictly
increasing. In this scenario, Assumptions \ref{hyp:stationarity}-\ref%
{hyp:crossing} are satisfied, with $U_{t}(x)=f_{T}(x\beta (U_{t}))$ and $%
g_{t}(y)=f_{t}\circ f_{T}^{-1}(y)$. Contrary to the previous example, this
model allows for both heterogeneous treatment effects, through the random
coefficient $\beta (U_{t})$, and an heterogeneous time trend, through the
function $f_{t}$. In the special case where $f_{t}(y)=y+\gamma_t$, the model
is a linear correlated random coefficients model. Note that even with such a
restriction on $f_{t}$, the treatment effect function $e\mapsto \beta (e)$
cannot be identified through standard quantile regression of $Y_{t}$ on $%
X_{t}$, because of the dependence between $X_{t}$ and $U_{t}$.

\subsection{Main Identification Results}

\label{sub:ident}

Our identification strategy works in two steps: In the first step, we
identify the effect of time on the outcome, i.e., the function $g_{t}$. This
implies that we identify $\widetilde{Y}_{t}=g_{t}^{-1}(Y_{t})$, whose
distribution does not depend on time anymore (conditional on $V_{t})$. Then,
in a second step, we can use time as an instrument to recover specific
causal effects. For ease of exposition, we first outline our method in the
case of $T=2$.

\subsubsection{Step 1: Identification of the Time Trend}

\label{ssub:time_trend}

To recover $g_{1}$, we rely on observations at the crossing point, i.e.
observations for whom $X_{1}=x_{1}^{\ast }$. Under Assumptions \ref%
{hyp:stationarity}-\ref{hyp:model}, the following is true:%
\begin{align}
P\left( Y_{2}\leq y|X_{2}=x_{1}^{\ast }\right) &\overset{A.2}{=}P\left(
U_{2}(x_{1}^{\ast })\leq y|V_{2}=F_{X_{2}}(x_{1}^{\ast })\right)  \notag \\
&\overset{A.1}{=} P\left( U_{1}(x_{1}^{\ast })\leq
y|V_{1}=F_{X_{2}}(x_{1}^{\ast })\right)  \notag \\
&\overset{A.3}{=} P\left( U_{1}(x_{1}^{\ast })\leq
y|V_{1}=F_{X_{1}}(x_{1}^{\ast })\right)  \notag \\
&\overset{A.2}{=} P\left( g_{1}(U_{1}(x_{1}^{\ast }))\leq
g_{1}(y)|X_{1}=x_{1}^{\ast }\right)  \notag \\
&= P\left( Y_{1}\leq g_{1}(y)|X_{1}=x_{1}^{\ast }\right) ,
\label{eq:for_ident_m}
\end{align}%
where we indicate the respective assumptions employed by superscripts upon
equalities. As a result, $g_{1}$ is identified by
\begin{equation}
g_{1}(y)=F_{Y_{1}|X_{1}}^{-1}\left[ F_{Y_{2}|X_{2}}(y|x_{1}^{\ast
})|x_{1}^{\ast }\right] .  \label{eq:ident_m}
\end{equation}%
Hence, under our assumptions, the time trend $g_{1}$ can be identified using
observations for which $X_{1}=x_{1}^{\ast }$ and $X_{2}=x_{1}^{\ast }$.
These two sets of observations, though distinct as we use repeated cross
sections, have the same distribution of unobservables and the same value of
the treatment. Therefore, differences between the distributions of outcomes
can only stem from the effect of time itself. This idea is very similar to
that used in difference-in-differences, where the control group permits the
identification of the (common) time trend. For this reason, in what follows
we classify all observations satisfying $X_{1}=x_{1}^{\ast }$ to form the
\textquotedblleft control group\textquotedblright .

\medskip Note that our model allows for heterogeneous time trends. As \cite%
{Athey06}, we therefore recover a whole function $g_{1}$ rather than a
single coefficient for the time trend, as in the standard
difference-in-differences model. Also as \cite{Athey06}, we identify $g_{1}$
by a quantile-quantile transform. When it comes to the identification of the
time trend, the main difference between our approach and that of \cite%
{Athey06} lies in how the control group is defined. While it is defined ex
ante in \cite{Athey06}, it is data-driven and defined through the crossing
points here.

\medskip Beyond the identification of $g_{1}$, \eqref{eq:ident_m} reveals
that the model is testable, if there are several crossing points between $%
F_{X_{1}}$ and $F_{X_{2}}$, say $x_{1}^{\ast }$ and $x_{1}^{\ast \ast }$. In
such a case, our model implies indeed that for all $y$,
\begin{equation*}
F_{Y_{1}|X_{1}}^{-1}\left[ F_{Y_{2}|X_{2}}(y|x_{1}^{\ast })|x_{1}^{\ast }%
\right] =F_{Y_{1}|X_{1}}^{-1}\left[ F_{Y_{2}|X_{2}}(y|x_{1}^{\ast \ast
})|x_{1}^{\ast \ast }\right] ,
\end{equation*}%
which is a testable restriction. Related to this, if the true set of
crossing points is an interval $I$, say, we have $F_{X_1|X_1\in I} =
F_{X_2|X_2\in I}$. Then, integrating \eqref{eq:for_ident_m} over $x_1^*\in I$%
, we obtain
\begin{equation*}
P\left( Y_{2}\leq y|X_{2} \in I \right) = P\left( Y_1 \leq g_1(y)|X_1 \in I
\right).
\end{equation*}
Therefore, $g_1(y)=F^{-1}_{Y_1|X_1\in I}\left[F_{Y_2|X_2\in I}(y)\right]$,
which implies that $g_1$ could be in principle estimated at a parametric
rather than nonparametric rate in this case.

\subsubsection{Step 2: Identification of ATT and QTT}

Next, we consider the identification of the treatment effects $%
\Delta^{ATT}(x,x^{\prime })$ and $\Delta ^{QTT}(p,x,x^{\prime })$. Start out
by considering the transformed potential and observed outcomes $\widetilde{Y}%
_t(x)=g_t^{-1}(Y_t(x))$ and $\widetilde{Y}_t=g_t^{-1}(Y_t)$ for $t=1,2$. By
virtue of Assumption \ref{hyp:stationarity}, $F_{\widetilde{Y}_1(x)}=F_{%
\widetilde{Y}_2(x)}$. Time can thus be seen as an instrument for the
treatment: while it affects the treatment (or its distribution, to be
precise), it has no direct effect on potential outcomes. The same idea is
used in a different DID framework by \cite{Chaisemartin14}.

\medskip To proceed with the identification of our model, let $%
q_t=F_{X_t}^{-1}\circ F_{X_T}$. Thus, $q_t(x)$ denotes the value of $X_t$
(say, income in period $t$) for an individual at the same rank as another
individual whose period $T$ income is $X_T=x$. Then,
\begin{eqnarray*}
E\left[ \widetilde{Y}_1|X_1=q_1(x)\right] &=&E\left[
U_1(q_1(x))|V_1=F_{X_2}(x)\right] \\
&\overset{A.1}{=}&E\left[U_2(q_1(x))|V_2=F_{X_2}(x)\right] \\
&=&E\left[U_2(q_1(x))|X_2=x\right] .
\end{eqnarray*}%
By the normalization $g_2(y)=y$, the latter is the mean counterfactual
outcome at period 2 for individuals with $X_2=x$ if $X_2$ was moved
exogenously to $q_1(x)$. We can therefore identify $\Delta ^{ATT}(x,q_1(x))$%
, the average effect of moving $X_2$ from their initial value $x$ to $q_1(x)$%
, by
\begin{eqnarray*}
\Delta ^{ATT}(x,q_1(x)) &\overset{A.2}{=}&E\left[ U_2(q_1(x))-U_2(x)|X_2=x%
\right] \\
&=&E\left[ \widetilde{Y}_1|X_1=q_1(x)\right] -E\left[ \widetilde{Y}_2|X_2=x%
\right].
\end{eqnarray*}
This means that we can obtain $\Delta^{ATT}(x,x^{\prime })$ for any pair $%
(x,x^{\prime })$ such that $x^{\prime }=q_1(x)$.

\medskip Similarly, we have, for any $p\in (0,1)$,
\begin{eqnarray*}
F_{\widetilde{Y}_1|X_1}^{-1}(p|q_1(x))
&=&F_{U_1(q_1(x))|V_1}^{-1}(p|F_{X_2}(x)) \\
&\overset{A. \ref{hyp:stationarity}}{=}%
&F_{U_2(q_1(x))|V_2}^{-1}(p|F_{X_2}(x)) \\
&\overset{A.\ref{hyp:direct_effect}}{=}& F_{Y_2(q_1(x))|X_2}^{-1}(p|x).
\end{eqnarray*}
This implies that
\begin{equation*}
\Delta^{QTT}(p,x,q_1(x)) = F_{\widetilde{Y}_1|X_1}^{-1}(p|q_1(x)) - F_{%
\widetilde{Y}_2|X_2}^{-1}(p|x).
\end{equation*}

Theorem \ref{thm:point_ident} summarizes our findings so far, and
generalizes it to any value of $T$.

\medskip

\begin{theorem}
Under Assumptions \ref{hyp:stationarity}-\ref{hyp:model}, we identify, for
all $x\in \text{Supp}(X_T)$, $p \in (0,1)$ and $t\in \{1,...,T-1\}$, the
functions $g_t$ and the average and quantile treatment effects $\Delta
^{ATT}(x,q_t(x))$ and $\Delta ^{QTT}(p,x,q_t(x))$. \label{thm:point_ident}
\end{theorem}

\medskip Note that if $x\mapsto Y_T(x)$ is differentiable, we have, by the
mean value theorem,
\begin{equation*}
Y_T(q_t(x)) - Y_T(x) = Y^{\prime }_T(\widetilde{X}) (q_t(x) - x),
\end{equation*}
for some random term $\widetilde{X}\in [x, q_t(x)]$. As a result, by Theorem %
\ref{thm:point_ident}, we identify
\begin{equation}
\Delta^{AME}_{\text{app}}(x) \equiv E[Y_T^{\prime }(\widetilde{X}) | X_T=x]
= \frac{\Delta^{ATT}(x,q_t(x))}{q_t(x) - x}.  \label{eq:ATT_AME}
\end{equation}
In other words, $\Delta^{ATT}(x,q_t(x))/(q_t(x) - x)$ may be interpreted as
an average marginal effect for units at $X_T=x$. Contrary to usually,
however, the derivative of $Y_T(.)$ is not evaluated at the current
treatment value $x$, but at another point $\widetilde{X}\in (x, q_t(x))$. If
$q_t(x)$ is close to $x$ or $Y_T$ is close to being linear, we can
nevertheless expect $Y_T^{\prime }(\widetilde{X})$ to be close to the usual
term $Y^{\prime }_T(x)$. As shown in Appendix \ref{sub:identification_ME},
we actually exactly identify the usual average marginal effect $%
\Delta^{AME}(x)\equiv E[Y_T^{\prime }(x) | X_T=x]$ at some particular values
of $x$.

\medskip Equation \eqref{eq:ATT_AME} also implies that we can identify
average marginal effects on larger subpopulation. Specifically, let $%
I_c=\{x\in \text{Supp}(X_T):|q_t(x)-x|>c\}$ for some $c>0$. Then, we
identify
\begin{equation*}
E[\Delta^{AME}_{\text{app}}(X_T) | X_T\in I_c] = E\left[\frac{%
\Delta^{ATT}(X_T,q_t(X_T))}{q_t(X_T) - X_T}\bigg| X_T\in I_c\right].
\end{equation*}
The advantage of considering this object is statistical accuracy, as we
average over the subpopulation such that $X_T\in I_c$.

\medskip With $T>2$, more periods produce more variations and thus allow one
to identify more treatment effects. Also, while Theorem \ref{thm:point_ident}
establishes the identification of treatment effects at period $T$, the same
reasoning yields the identification of treatment effects at any other
periods. To see this, note that
\begin{align*}
E\left[ Y_{t}(q_{t}^{-1}(x))-Y_{t}(x)|X_{t}=x\right] & =E\left[
g_{t}(U_{t}(q_{t}^{-1}(x))|V_{t}=F_{X_{t}}(x)\right] -E\left[ Y_{t}|X_{t}=x%
\right] \\
& =E\left[ g_{t}(U_{T}(q_{t}^{-1}(x))|V_{T}=F_{X_{t}}(x)\right] -E\left[
Y_{t}|X_{t}=x\right] \\
& =E\left[ g_{t}(Y_{T})|X_{T}=q_{t}^{-1}(x)\right] -E\left[ Y_{t}|X_{t}=x%
\right] .
\end{align*}%
The right-hand side is identified, since $g_{t}$ is identified, as outlined
above. Hence, we identify all period $t$-parameters of the form $E\left[
Y_{t}(q_{t}^{-1}(x))-Y_{t}(x)|X_{t}=x\right] $ and $%
F_{Y_{t}(q_{t}^{-1}(x))|X_{t}=x}^{-1}-F_{Y_{t}(x)|X_{t}=x}^{-1}$.

\medskip If $F_{X_{t}}$ does not vary over time, then $q_{t}(x)=x$ and
Theorem \ref{thm:point_ident} boils down to the identification of the
trivial parameters $\Delta ^{ATT}(\xi ,\xi )=0$ an $\Delta ^{QTT}(\xi ,\xi
)=0$. As mentioned above, the distribution of $X_{t}$ needs to vary for our
method to have non-trivial identification power. Finally, we cannot point
identify from Theorem \ref{thm:point_ident} the parameters $\Delta
^{ATT}(x,x^{\prime })$ and $\Delta ^{QTT}(x,x^{\prime })$ if $x^{\prime
}\neq q_{t}(x)$ for some $t\in \{1,...,T-1\}$. We show however in Subsection %
\ref{sub:partial} that we can at least set identify these parameters under
plausible curvature restrictions, and in Subsection \ref{sub:point_ident}
that we can point identify them under stronger conditions.

\subsection{Relationship to other Approaches}

\label{sub:discussion}

\subsubsection{Comparison with Panel Data Models}

\label{subsub:comparison_with_panel_data_models}

While we rely on time variation to identify causal effects, as in panel
data, our assumptions contrast with those typically used in panel data.
First, our stationarity condition is different from the condition
\begin{equation}
U_{s}(x)|X_{1},...,X_{T}\sim U_{t}(x)|X_{1},...,X_{T},
\label{eq:condit_panel}
\end{equation}%
commonly assumed in panel data
\citep[see,
e.g.,][]{Manski87,Honore92,Hoderlein12,Graham12,Chernozhukov13,chernozhukov2015nonparametric}%
. To understand the differences between the two, consider two polar cases.
In the first, endogeneity stems from contemporaneous simultaneity between $%
U_{t}(x)$ and $V_{t}$, as is often the case with variables that are jointly
determined, while $(U_{t}(x),V_{t})_{t=1...T}$ are i.i.d. across time. If
so, Assumption \ref{hyp:stationarity} is satisfied. On the other hand, %
\eqref{eq:condit_panel} does not hold, unless $U_{t}(x)$ is independent of $%
V_{t}$, because the distribution of $U_{s}(x)$ conditional on $%
(X_{1},...,X_{T})$ is a function of $X_{s}$ only, i.e., $%
f_{U_{s}(x)|X_{1},...,X_{T}}(a|x_{1},...,x_{T})=f_{U_{s}(x)|X_{s}}(a|x_{s})$%
, while the conditional distribution $U_{t}(x)$ is a function of $X_{t}$
only, and they do not coincide in general if $x_{s}\neq x_{t}$. Assuming $%
(U_{s}(x),V_{s})$ independent of $(U_{t}(x),V_{t})$ is of course often
unrealistic, but the same conclusion would hold with, say, a vector
autoregressive structure.

\medskip In the second case, $U_{t}(x)=(A(x),U_{t})$ where $A(x)$ is an
individual effect potentially correlated with $X_{1},...,X_{T}$ and $%
(U_{t})_{t=1}^{T}$ are i.i.d. idiosyncratic shocks that are independent of $%
(A(x),X_{1},...,X_{T})$. In this case, the condition \eqref{eq:condit_panel}
is always satisfied. On the other hand, Assumption \ref{hyp:stationarity}
holds only under a special correlation structure between $A(x)$ and $%
(X_{1},...,X_{T})$: $A(x)|V_{t}=v\sim A(x)|V_{s}=v,$ which for instance
imposes Cov$(A(x),V_{t})=$ Cov$(A(x),V_{s}),\;s\neq t$. While this still
allows for arbitrary contemporaneous correlation between $A(x)$ and $V_{t}$,
it does not allow for any time-varying covariance.

\medskip Another difference with panel data models lies in the type of
variations that we require on $X_{t}$. With panels, we require the
individual value of the treatment to vary over time, the fixed effects
absorbing any variable that is constant across time. Such a requirement is
not needed here, since the distribution of $X_{t}$ can change over time even
if $X_{t}$ is constant for each individual, provided new generations are
involved at date $t$ compared to date $s$. On the other hand, compared to
panel data, we do not identify anything here, apart from the time trend $%
g_{t}$, when the treatment changes at an individual level but the
distribution of $X_{t}$ remains constant over time. This is one key aspect
that distinguishes our identification strategy from panel data based
strategies.


\subsubsection{Comparison with instrumental variable models}

\label{subs:compar_IV}

Our result is also related to the literature on identification of triangular
models with instruments and cross-sectional data
\citep[see in
particular][]{Imbens09}. Such models take the following form:
\begin{align*}
Y & = g(X, U), \\
X & = h(Z, V),
\end{align*}
where $Z$ denotes the instrument, $V\in\mathbb{R}$ $h(z,\cdot)$ is
increasing and $(U, V) \perp\!\!\!\perp Z$. We rely on a similar structure here, with $Z$ playing the role of the
instrument. Assumption \ref{hyp:stationarity} then corresponds to the
condition $(U, V) \perp\!\!\!\perp Z$. Our model still has one distinctive
feature from this model: time may have a direct effect on the outcome
variable, though this effect has to be restricted through Assumption \ref%
{hyp:direct_effect}.\footnote{\label{foot:Z_discrete} Another difference
with \cite{Imbens09} is that the instrument is discrete in our setup. As a
result, some common parameters such as the overall average marginal effects
are not identified without further restrictions.} The role of the crossing
condition, then, is to pin down this effect, so that we can modify the
outcome in such a way that time becomes a valid instrument.

\medskip This parallel also illustrates some possible limitations of our
approach. In particular, \cite{Kasy11} showed that if in reality $V$ is
multidimensional (with still $(U, V) \perp\!\!\!\perp Z$), then in general $%
U $ is not independent of $Z$ conditional on $F_{X|Z}(X|Z)$. In our context,
this means that Assumption \ref{hyp:stationarity} fails to hold if $X_t$
depends on a multiple unobserved terms. Consider for instance returns to
schooling. In the model of \cite{card2001estimating}, schooling $X_t$
depends on individual marginal cost $c_t$ and individual returns $r_t$
through the relationship $X_t=(c_t - r_t)/k$ for some constant $k>0$.
Suppose that returns $r_t$ are time invariant but exogenous variations in
tuition fees affect marginal costs multiplicatively, so that $c_t=\alpha_t
\widetilde{c}_t$ and $(\widetilde{c}_t, r_t)$ is time invariant. The results
of \cite{Kasy11}, and in particular his Section 2, then imply that
Assumption \ref{hyp:stationarity} would fail in this example.

\section{Extensions}

\subsection{Including Covariates}

We consider here the case where exogenous covariates $Z_{t}$ also affect the
outcome variable. Specifically, let $Y_{t}(x,z)$ denote the potential
outcome associated with the values $x$ and $z$ (of random variables $X_{t}$
and $Z_{t}$, respectively). We observe $Y_{t}\equiv Y_{t}(X_{t},Z_{t})$. We
still focus on the effect of $X_{t}$ hereafter. In this case, the preceding
analysis can be conducted conditionally on $Z_{t}$. We briefly discuss this
extension here, by considering only the discrete average and quantile
effects
\begin{eqnarray*}
\Delta ^{ATT}(x,x^{\prime },z) &\equiv &E\left[ Y_{T}(x^{\prime
},z)-Y_{T}(x,z)|X_{T}=x,Z_{T}=z\right] \qquad\text{and} \\
\Delta ^{QTT}(p,x,x^{\prime },z) &\equiv &F_{Y_{T}(x^{\prime
},z)|X_{T},Z_{T}}^{-1}(p|x,z)-F_{Y_{T}(x,z)|X_{T},Z_{T}}^{-1}(p|x,z).
\end{eqnarray*}%
The marginal effects can be handled similarly. We first restate our previous
conditions in this context. The rank variable is now defined conditionally
on $Z_{t}$, i.e., $V_{t}=F_{X_{t}|Z_{t}}(X_{t}|Z_{t})$.

\begin{assumption_nonumber}
\label{hyp:stationarity_cov}{\normalfont \textbf{1C}}\; $\text{Supp}%
((V_t,Z_t))$ does not depend on $t$. For all $t \in \{1,...,T\}$, $%
Y_t(x,z)=g_t(z,U_t(x,z))$ where for all $(x,v,z)\in \mathcal{X}\times \text{%
Supp}((V_t,Z_t))$, the distribution of $U_t(x,z)|V_t=v,Z_t=z$ does not
depend on $t$.
\end{assumption_nonumber}

\begin{assumption_nonumber}
\label{hyp:model_cov}{\normalfont \textbf{4C}}\; For all $%
(t,z)\in\{1,...,T\}\times \text{Supp}(Z_t)$, $E(|Y_t|)<\infty$ and $%
F_{X_t|Z_t}(.|z)$ is continuous and strictly increasing on $\text{Supp}%
(X_t|Z_t=z)$. For all $(x^{\prime },z)\in \text{Supp}((X_t,Z_t))$ and $u\in
\text{Supp}(U_t(x^{\prime },z))$, there exist versions of $E[Y_t|X_t,Z_t]$
and $P^{U_t(x^{\prime }) |V_t,Z_t}$ such that $x\mapsto E[Y_t|X_t=x,Z_t=z]$
and $v \mapsto F_{U_t(x^{\prime })|V_t,Z_t}(u|v,z)$ are continuous.
\end{assumption_nonumber}

Next, we consider two versions of Assumptions \ref{hyp:direct_effect} and %
\ref{hyp:crossing}, namely Assumptions 2C-3C and 2C'-3C' below. The
trade-off between these two versions is basically between the generality of
the model and the requirements on the data. In the first version, we allow
for a more general time trend (i.e., Assumption 2C' is a particular case of
Assumption 2C). However, the crossing condition in Assumption 3C is more
demanding than in Assumption 3C', because the former requires to observe a
crossing point for each value of $z$.

\begin{assumption_nonumber}
{\normalfont \textbf{2C}} \; For all $(z,t)\in \text{Supp}(Z_T) \times
\{1,...,T\}$, $U_t(x,z)\in \mathbb{R}$ and $g_t(z,.)$ is strictly
increasing. Without loss of generality, we let $g_T(z,y)=y$ for all $(y,z)\in%
\text{Supp}((Y_T, Z_T))$. \label{hyp:direct_effect_cov}
\end{assumption_nonumber}

\begin{assumption_nonumber}
{\normalfont \textbf{3C}} \; For all $(z,t)\in \text{Supp}(Z_T) \times
\{1,...,T-1\}$, there exists $x^*_t(z)$ such that $F_{X_T|Z_T}(x^*_t(z)|z) =
F_{X_t|Z_t}(x^*_t(z)|z) \in (0,1)$. \label{hyp:crossing_cov}
\end{assumption_nonumber}

\begin{assumption_nonumber}
{\normalfont \textbf{2C'}} \; For all $(t,x,z) \in \{1,...,T\}\times \text{%
Supp}((X_t,Z_t))$, $g_t(z,U_t(x,z))=h_t(U_t(x,z))$, with $U_t(x,z)\in
\mathbb{R}$ and $h_t(.)$ strictly increasing. Without loss of generality, we
let $h_T(y)=y$ for all $y\in\text{Supp}(Y_T)$. \label{hyp:direct_effect_cov2}
\end{assumption_nonumber}

\begin{assumption_nonumber}
{\normalfont \textbf{3C'}} \; For all $t\in \{1,...,T-1\}$, there exists $%
(x^*_t,z^*_t)$ such that $F_{X_T|Z_T}(x^*_t|z^*_t) =
F_{X_t|Z_t}(x^*_t|z^*_t) \in (0,1)$. \label{hyp:crossing_cov2}
\end{assumption_nonumber}

Both sets of the assumptions lead to the same results, which are
qualitatively very similar to those of Theorem \ref{thm:point_ident}. In
what follows, we let $%
q_{t}(x|z)=F_{X_{t}|Z_{t}}^{-1}(F_{X_{T}|Z_{T}}(x|z)|z) $. The proof of
Theorem \ref{thm:point_ident_cov}C is a straightforward extension of the
proof of Theorem \ref{thm:point_ident}, and hence is omitted.

\medskip \setcounter{theorem}{0}

\begin{theorem}
\hspace{-0.16cm}{\normalfont \textbf{C}} \, Suppose that Assumptions 1C and
4C and either Assumptions 2C-3C or Assumptions 2C'-3C' hold. Then, for
almost all $(x,z)\in \text{Supp}((X_T,Z_T))$, all $p \in (0,1)$ and all $%
t\in \{1,...,T-1\}$, the functions $g_{t}$ and the average and quantile
treatment effects $\Delta ^{ATT}(x,q_t(x|z),z)$ and $\Delta ^{QTT}(p
,x,q_t(x|z),z)$ are identified. \label{thm:point_ident_cov}
\end{theorem}

Here again, we can relate $\Delta ^{ATT}(x,q_t(x|z),z)$ with average
marginal effects. If $Y_T(.,z)$ is differentiable, by the mean value
theorem,
\begin{equation*}
Y_T(q_t(x|z),z) - Y_T(x,z) = \frac{\partial Y_T}{\partial x}(\widetilde{X}%
_z, z),
\end{equation*}
for some $\widetilde{X}_z\in (x,q_t(x|z))$. Then,
\begin{equation*}
\Delta ^{AME}_{\text{app}}(x,z)\equiv E\left[\frac{\partial Y_T}{\partial x}(%
\widetilde{X}_z, z) | X_T=x, Z_T=z\right] = \frac{\Delta ^{ATT}(x,q_t(x|z),z)%
}{q_t(x|z)-x}.
\end{equation*}
This equation implies that we can also average over $x$ and $z$ to gain
statistical power. Specifically, let $I_c=\{(x,z)\in \text{Supp}((X_T,Z_T)):
|x-q_t(x|z)|>c\}$ for some $c>0$. Under the conditions behind Theorem \ref%
{thm:point_ident_cov}, we can identify
\begin{equation*}
E\left[\Delta ^{AME}_{\text{app}}(X_T,Z_T)| (X_T,Z_T)\in I_c \right]= E\left[%
\frac{\Delta ^{ATT}(X_T,q_t(X_T|Z_T),Z_T)}{q_t(X_T|Z_T)-X_T}| (X_T,Z_T)\in
I_c\right].
\end{equation*}

\subsection{Multivariate Treatment}

\label{sub:multi}

Our framework directly extends to multivariate treatments, $%
X_{t}=(X_{1t},...,X_{kT})\in \mathbb{R}^{k}$, $k\geq 2$, by just making a
few changes. First, we now define $V_t$ to be $V_t =
(F_{X_{1t}}(X_{1t}),...,F_{X_{kt}}(X_{kT}))$. Second, we replace Assumption %
\ref{hyp:crossing} by the following condition:

\begin{assumption_nonumber}
{\normalfont \textbf{3M}}\; For all $(j,t)\in \{1,...,k\}\times \{1,...,T\}$%
, there exists $x_{jt}^{\ast }\in \mathbb{R}$ such that $F_{X_{jt}}(x_{jt}^{%
\ast })=F_{X_{jT}}(x_{jt}^{\ast })\in (0,1)$. \label{hyp:crossing_multi}
\end{assumption_nonumber}

Finally, we now define $q_{t}$ as $q_{t}(x_{1},...,x_{k})=\left(
F_{X_{1t}}^{-1}\circ F_{X_{1T}}(x_{1}),...,F_{X_{kt}}^{-1}\circ
F_{X_{kT}}(x_{k})\right)$. Then, we obtain the same point identification
result as before.

\medskip \setcounter{theorem}{0}

\begin{theorem}
\hspace{-0.16cm}{\normalfont \textbf{M}} \, Suppose Assumptions \ref%
{hyp:stationarity}, \ref{hyp:direct_effect}, 3M and \ref{hyp:model} hold.
Then, for all $(t,x)\in \{1,...,T-1\}\times \text{Supp}(X_T)$, the function $%
g_t$ and $\Delta^{ATT}(x,q_t(x))$ and $\Delta^{QTT}(p,x,q_t(x))$ are
identified. \label{thm:point_ident_multi}
\end{theorem}

\setcounter{theorem}{2}

\subsection{Partial Identification of Other Treatment Effects}

\label{sub:partial}

Theorem \ref{thm:point_ident} implies that we can point identify some but
not all average treatment effects $\Delta ^{ATT}(x,x^{\prime })$. Similarly,
we point identify the average marginal effects only at some particular
points. We show in this subsection that with three or more periods of
observation, we can get bounds for many other points under a weak local
curvature condition. Let us consider average marginal effects, for instance.
The idea is that if $x\mapsto U_{T}(x)$ is locally concave (say) and $%
q_{t}(x)<x$, then $[U_{T}(q_{t}(x))-U_{T}(x)]/[q_{t}(x)-x]$ is an upper
bound for $dU_{T}/dx(x)=dY_{T}/dx(x)$. By integration, $\Delta^{ATT}(x,q_{t}(x))/(q_{t}(x)-x)$ is therefore an
upper bound for $\Delta ^{AME}(x)$. Similarly, we obtain a lower bound for $%
\Delta ^{AME}(x)$ if $q_t(x)>x$. Figure illustrates this idea with $T=3$ and
$q_2(x)<x<q_1(x)$. Note the same idea can be used to obtain bounds $\Delta
^{ATT}(x,x^{\prime })$ for $x^{\prime }\not\in \{q_{t}(x),t=2,...,T\}$.

\begin{figure}[]
\begin{center}
\includegraphics[scale=0.8]{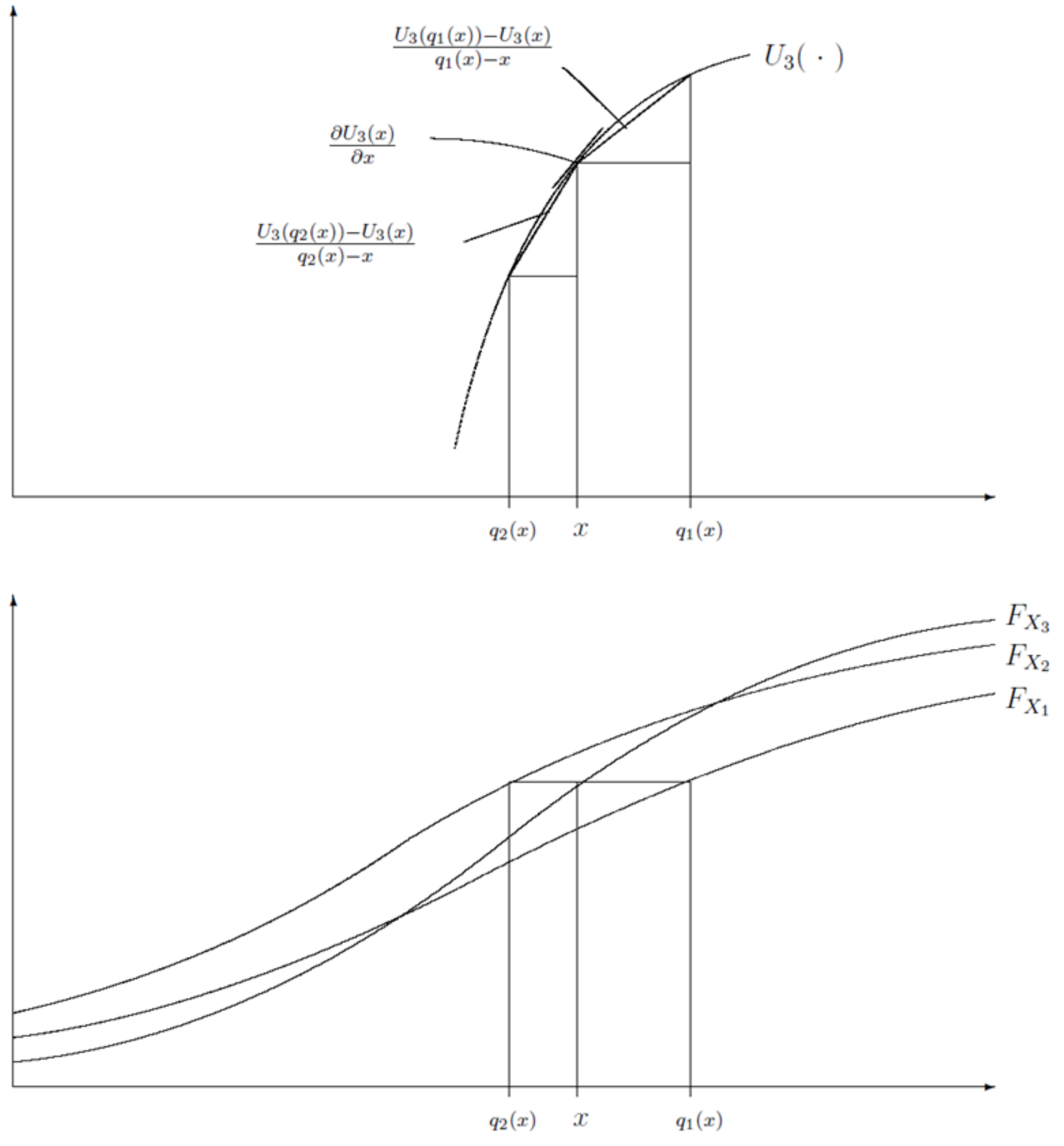}
\end{center}
\par
\vspace{-0.3cm}
\caption{Bounds under the local curvature condition}
\label{fig:bounds}
\end{figure}

The above argument works even if we do not know a priori whether $U_{T}(.)$
is concave or convex. Using the minimum and the maximum of the local
discrete treatment effect will be sufficient to obtain bounds, provided that
$U_{T}(.)$ is locally concave or locally convex around $x$. We therefore
adopt the following definition.

\begin{definition}
$x\mapsto U_T(x)$ is locally concave or convex on $[\widetilde{x},\widetilde{%
x}^{\prime }]$ if, almost surely (a.s.), it is twice differentiable and
\begin{equation*}
\frac{\partial ^2 U_T}{\partial x^2}(x) \leq 0 \; \forall x\in [\widetilde{x}%
,\widetilde{x}^{\prime }] \; \text{a.s. or } \; \frac{\partial ^2 U_T}{%
\partial x^2}(x) \geq 0 \; \forall x\in [\widetilde{x},\widetilde{x}^{\prime
}] \; \text{a.s}.
\end{equation*}
\label{hyp:local_curvature}
\end{definition}

Let us introduce, for all $(x,x^{\prime })\in \text{Supp}(X_{T})$, $(%
\underline{x}_{T}(x^{\prime }),\overline{x}_{T}(x^{\prime }))$ defined by
\begin{eqnarray*}
\underline{x}_{T}(x^{\prime }) &=&\max \{q_{t}(x),t\in
\{1,...,T-1\}:\;q_{t}(x)\neq x\text{ and }q_{t}(x)<x^{\prime }\}, \\
\overline{x}_{T}(x^{\prime }) &=&\min \{q_{t}(x),t\in
\{1,...,T-1\}:\;q_{t}(x)\neq x\text{ and }q_{t}(x)>x^{\prime }\}.
\end{eqnarray*}%
If the sets are empty, we let $\underline{x}_{T}(x^{\prime })=-\infty $ and $%
\overline{x}_{T}(x^{\prime })=+\infty $.

\medskip

\begin{theorem}
Suppose that Assumptions \ref{hyp:stationarity}-\ref{hyp:crossing}
aresatisfied. For any $x<x^{\prime }$, if $U_{T}$ is locally concave or
convex on $[\min (x,\underline{x}_{T}(x^{\prime })),\overline{x}%
_{T}(x^{\prime })]$, then
\begin{align*}
& (x^{\prime }-x)\min \left\{ \frac{\Delta ^{ATT}(x,\underline{x}%
_{T}(x^{\prime }))}{\underline{x}_{T}(x^{\prime })-x},\frac{\Delta ^{ATT}(x,%
\overline{x}_{T}(x^{\prime }))}{\overline{x}_{T}(x^{\prime })-x}\right\}
\leq \Delta ^{ATT}(x,x^{\prime }) \\
\leq & (x^{\prime }-x)\max \left\{ \frac{\Delta ^{ATT}(x,\underline{x}%
_{T}(x^{\prime }))}{\underline{x}_{T}(x^{\prime })-x},\frac{\Delta ^{ATT}(x,%
\overline{x}_{T}(x^{\prime }))}{\overline{x}_{T}(x^{\prime })-x}\right\} .
\end{align*}%
If $U_{T}$ is locally concave or convex on $[\underline{x}_{T}(x),\overline{x%
}_{T}(x)]$, then
\begin{align*}
& \min \left\{ \frac{\Delta ^{ATT}(x,\underline{x}_{T}(x))}{\underline{x}%
_{T}(x)-x},\frac{\Delta ^{ATT}(x,\overline{x}_{T}(x))}{\overline{x}_{T}(x)-x}%
\right\} \leq \Delta ^{AME}(x) \\
\leq & \max \left\{ \frac{\Delta ^{ATT}(x,\underline{x}_{T}(x))}{\underline{x%
}_{T}(x)-x},\frac{\Delta ^{ATT}(x,\overline{x}_{T}(x))}{\overline{x}_{T}(x)-x%
}\right\} .
\end{align*}

The bounds are understood to be infinite when either $\underline{x}%
_T(x^{\prime })=-\infty$ or $\overline{x}_T(x^{\prime })=+\infty$ (whether $%
x^{\prime }>x$ or $x^{\prime }=x$). \label{thm:bounds}
\end{theorem}

Both bounds are finite, provided that there exists $t,t^{\prime }$ such that
$q_{t}(x)<x<q_{t^{\prime }}(x)$, which implies that $T\geq 3$. More
generally, the bounds improve with $T$, because $(\underline{x}%
_{T}(x^{\prime }))_{T\in \mathbb{N}}$ and $(\overline{x}_{T}(x^{\prime
}))_{T\in \mathbb{N}}$ are by construction increasing and decreasing,
respectively. Also, the local curvature condition becomes less restrictive as $T$
increases, because the interval on which $U_{T}$ has to satisfy this
condition decreases. This condition is particularly credible if $%
q_{t}(x)\mapsto \Delta (x,q_{t}(x))/(q_{t}(x)-x)$ is monotonic, because such
a pattern is implied by global concavity or global convexity.

\medskip Two other remarks on Theorem \ref{thm:bounds} are in order. First,
we do not establish that the bounds are sharp, though we conjecture that
they are. Second, similar to the point identification results of Theorem \ref%
{thm:point_ident}, the partial identification results of Theorem \ref%
{thm:bounds} can be extended to the multivariate setting. Specifically, we
can use the system of inequalities
\begin{equation*}
U_{T}(q_{t}(x))-U_{t}(x)\geq \frac{\partial U_{T}(x)}{\partial x}^{\prime
}(q_{t}(x)-x),
\end{equation*}%
which hold for all $t=1...T-1$ if $U_{T}$ is locally convex (inequalities
are reverted if $U_{T}$ is locally concave). These inequalities imply some
bounds on $E[\partial U_{T}(x)/\partial x]$. A necessary condition for the
bounds to be finite on each component of $E[\partial U_{T}(x)/\partial x]$
is that $T-1\geq 2\dim (X_{t})$. This condition generalizes the above
restriction $T\geq 3$. It makes intuitive sense that more time periods are
required when the endogenous treatment is multivariate.

\medskip To illustrate Theorem \ref{thm:bounds}, we consider the following
example:
\begin{eqnarray*}
Y_{t} &=&1-\exp (-0.5(\delta _{t}+X_{t}+U_{t})) \\
X_{t} &=&\mu _{t}+\sigma _{t}\Phi ^{-1}(V_{t}),
\end{eqnarray*}%
where $V_{t}\sim U[0,1]$ and $U_{t}|V_{t}\sim \mathcal{N}(V_{t},1)$. We also
suppose that
\begin{eqnarray*}
\mu _T=2.5, &&\mu _{t}\sim \mathcal{N}(\mu _T,1)\;\text{for }t<T, \\
\sigma _T=1, &&\sigma _{t}\sim \chi ^{2}(1)\;\text{for }t<T, \\
\delta _T=0, &&\delta _{t}\sim \mathcal{N}(0,1)\;\text{for }t<T.
\end{eqnarray*}

In this example, Assumptions \ref{hyp:stationarity}, \ref{hyp:direct_effect}
(with $g_{t}(y)=1-\exp (-0.5\delta _{t})(1-y)$) and \ref{hyp:crossing} are
satisfied, the latter because $\sigma _{t}\neq \sigma _{T}$ almost surely.
The local curvature condition also holds, since $u\mapsto 1-\exp (-0.5u)$ is
concave. Figure \ref{fig:bounds_example} displays the bounds on $%
\Delta_{1}^{AME}(x)$ for $T=3,4,5$ and 6. Note that the bounds coincide for $%
T-1$ points. This simply reflects the point identification result of Theorem %
\ref{thm:point_ident_ME}. We also see that in the interval where we get finite bounds, i.e., the
interval for which $-\infty <\underline{x}_{T}(x)<\overline{x}_{T}(x)<\infty$%
, the bounds are quite informative even for $T=3$. Figure \ref%
{fig:bounds_example} also shows that as $T$ increases, both the bounds
shrink and the interval on which we get finite bounds increase. For $T=6$,
we get informative bounds for $x\in \lbrack 1,\;3.85]$, which corresponds
roughly to 85\% of the population. This means that we could also obtain
finite bounds for the average partial effect for this large fraction of the
total population.

\vspace{-0.3cm}
\begin{figure}[]
\begin{minipage}{0.45\textwidth}
\includegraphics[scale=0.35]{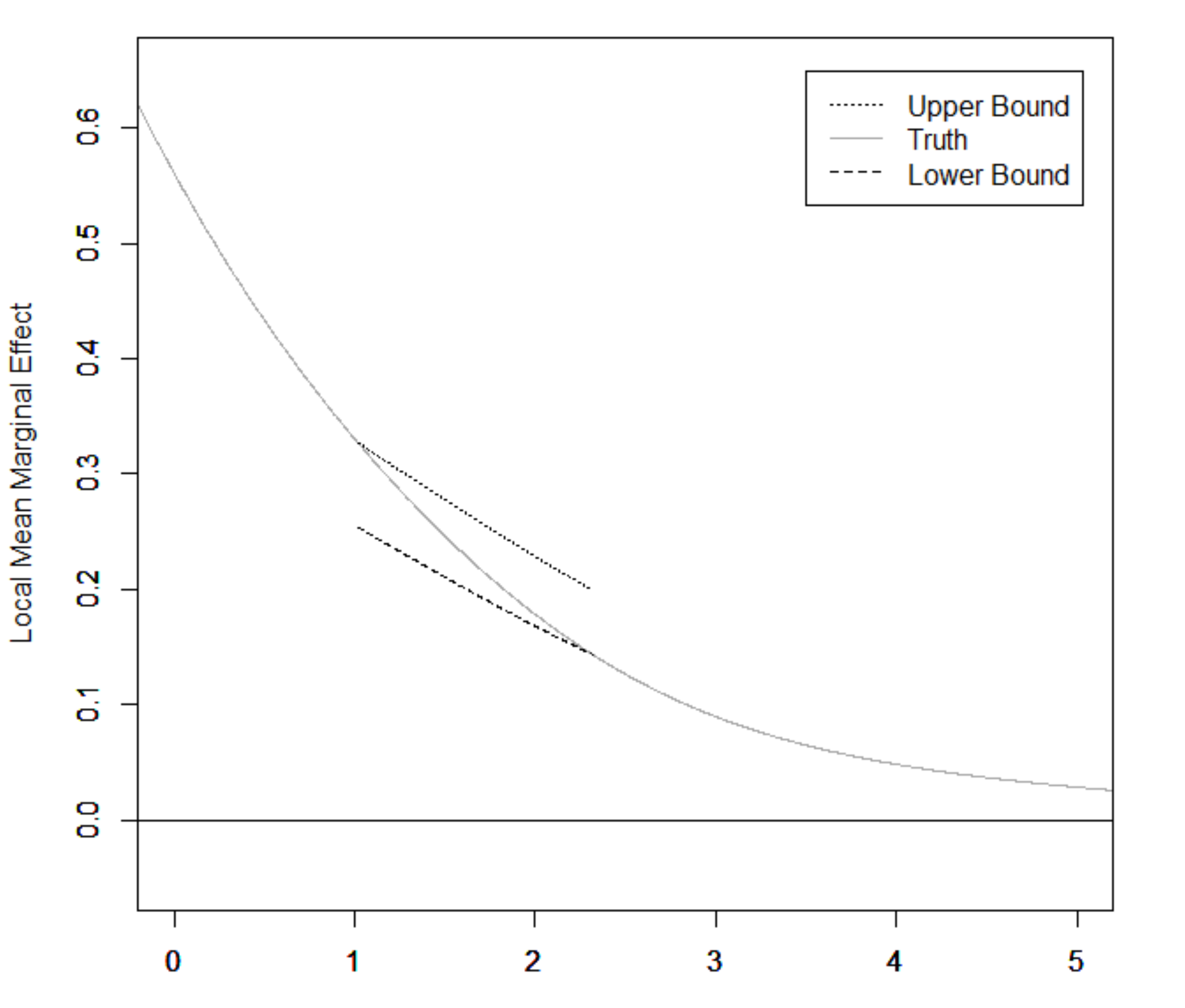}
\begin{center}
~\vspace{-0.3cm}
~\hspace{0.2cm}$T=3$
\end{center}
\includegraphics[scale=0.35]{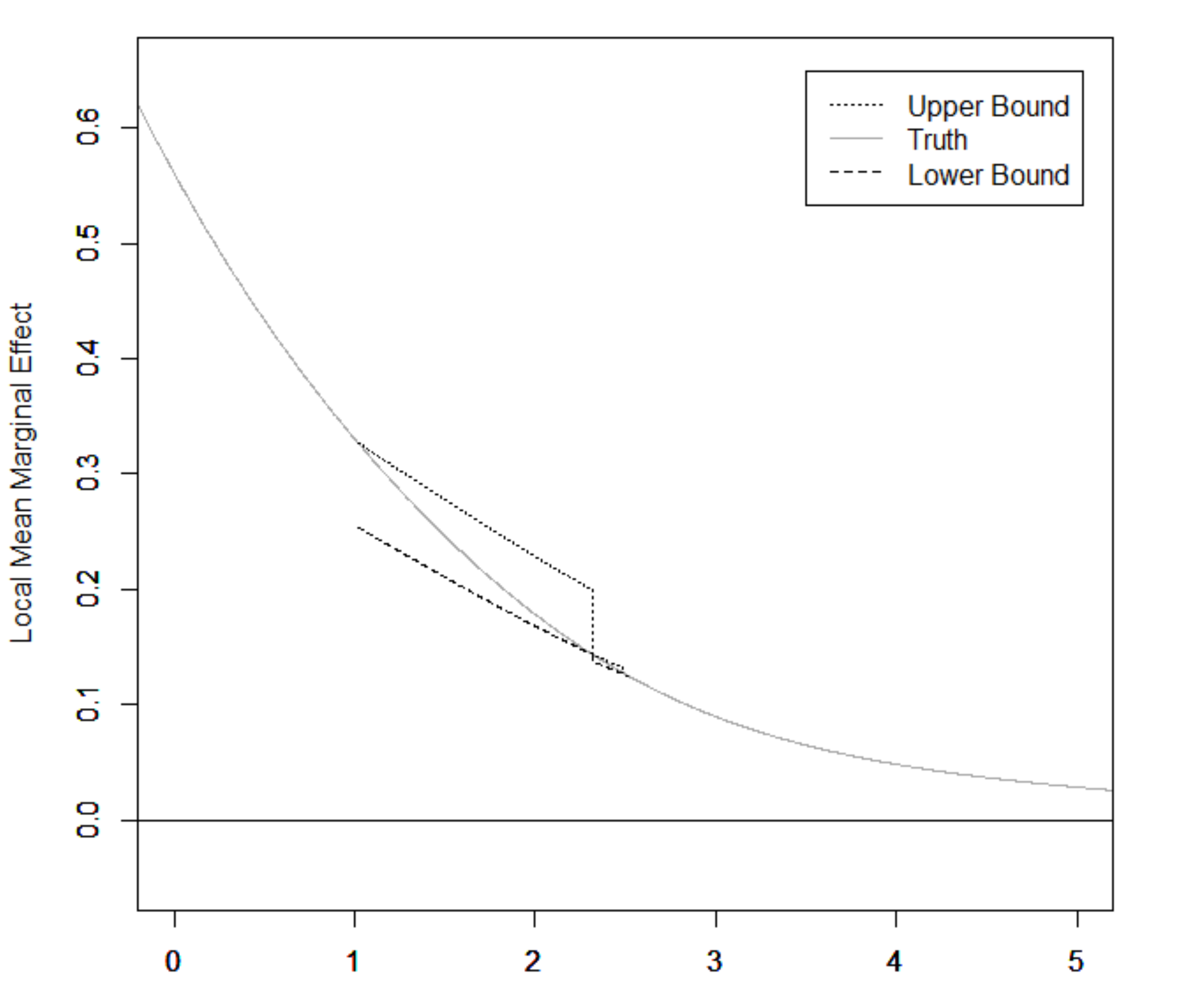}
\begin{center}
\vspace{-0.2cm}
~\hspace{0.2cm} $T=4$
\end{center}
\end{minipage}
\begin{minipage}{0.45\textwidth}
\hspace{1cm} \includegraphics[scale=0.35]{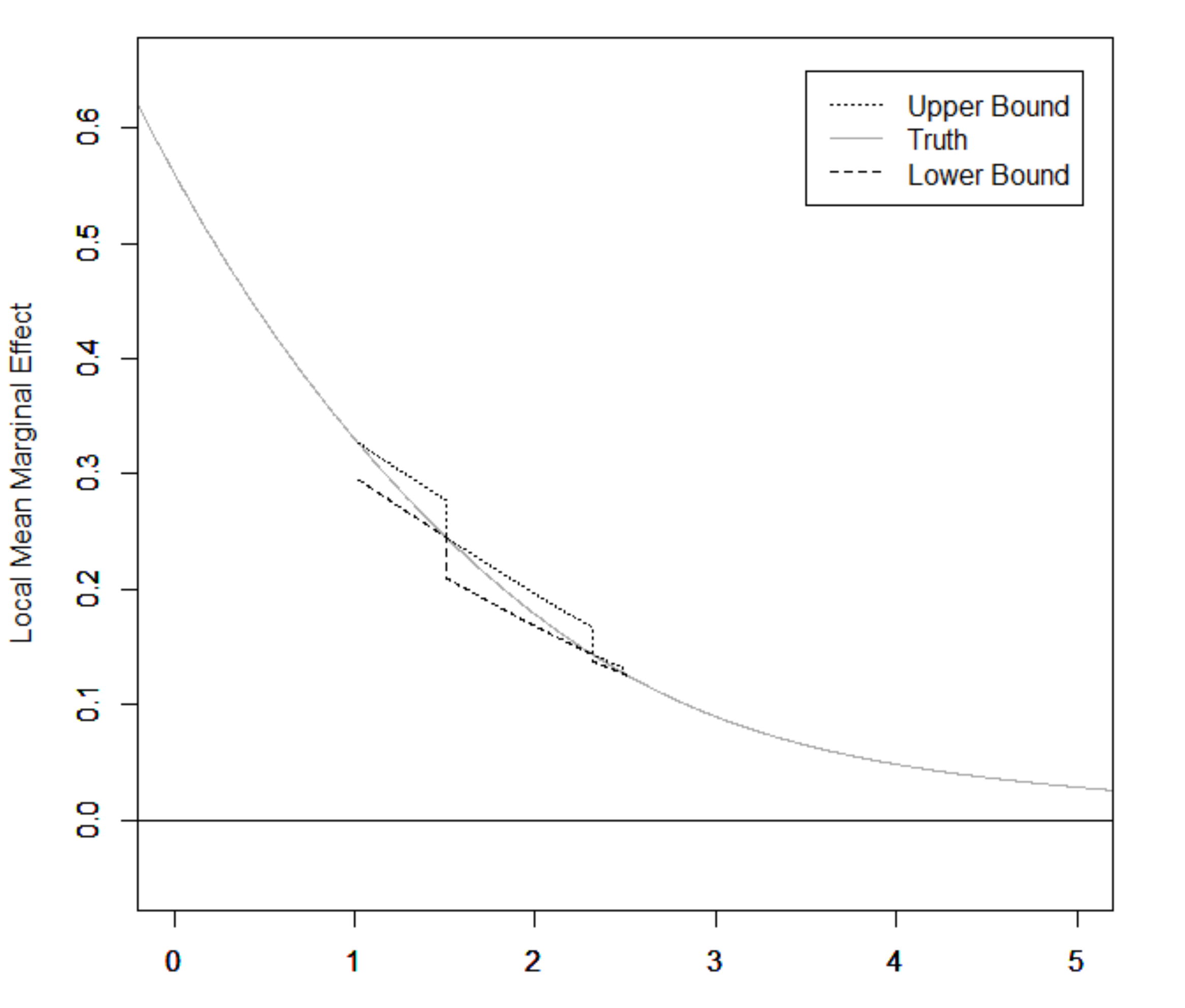}
\begin{center}
~\vspace{-0.3cm}
~\hspace{1.4cm} $T=5$
\end{center}
\hspace{1cm} \includegraphics[scale=0.35]{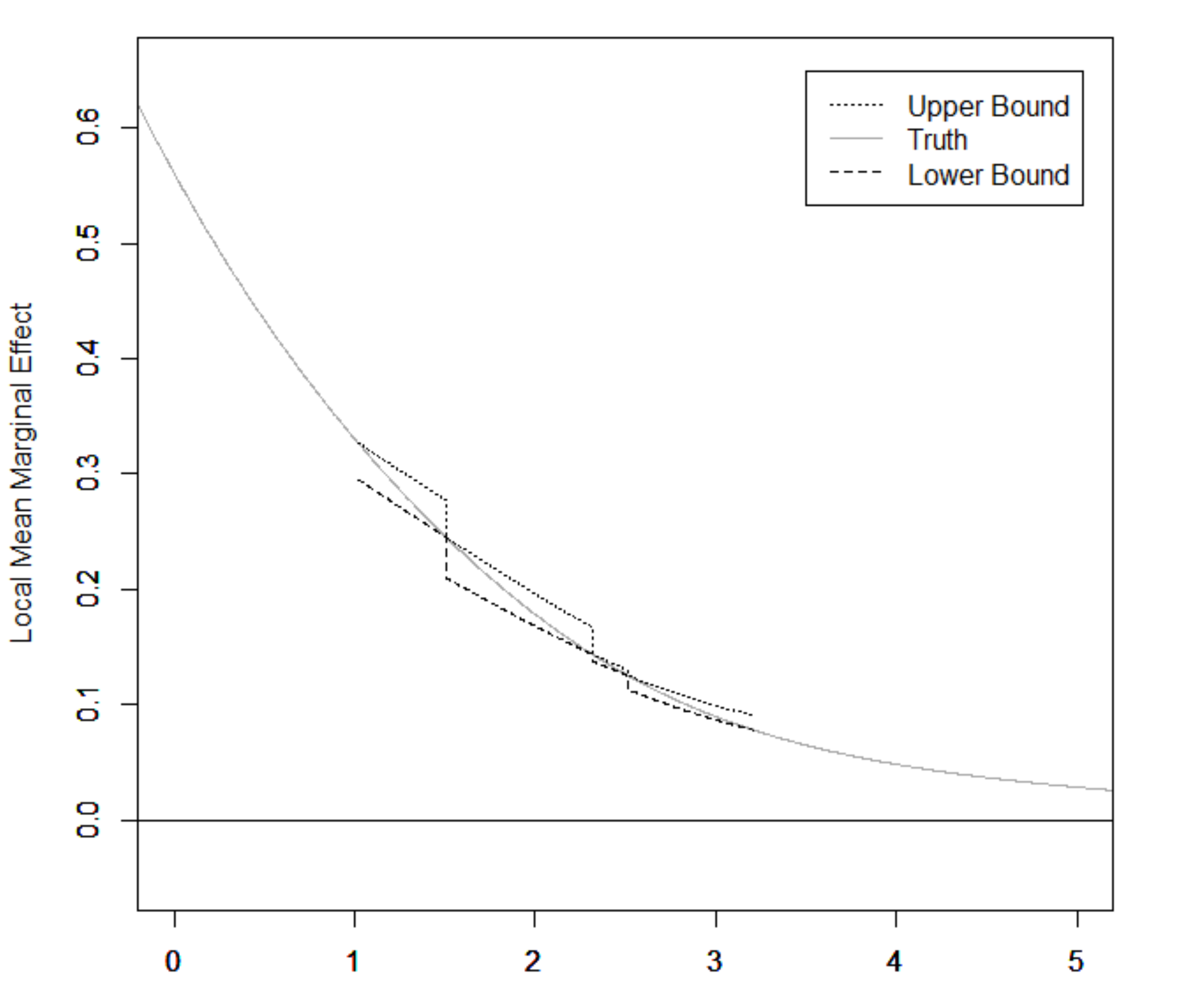}
\begin{center}
\vspace{-0.2cm}
~\hspace{1.4cm} $T=6$
\end{center}
\end{minipage}
\caption{Example of bounds on $\Delta^{AME}(x)$ for different values of $x$
and $T=3, 4, 5$ and $6$.}
\label{fig:bounds_example}
\end{figure}

\subsection{Point Identification with a Correlated Random Coefficient Model}

\label{sub:point_ident}

As we have established in Theorem \ref{thm:point_ident}, we can point
identify several treatment effect parameters under Assumptions \ref%
{hyp:stationarity}-\ref{hyp:crossing}, but these are by no means all
possible causal effects one may be interested in. Many more treatment
parameters can be set identified under often plausible curvature
restrictions, in particular average marginal effects and effects of the kind
$\Delta ^{ATT}(x,x^{\prime })$. However, these bounds may be wide in some
applications, conducting inference on the corresponding parameters may be
cumbersome or even impractical. Hence it makes sense to search for
additional assumptions that yield point identification of average structural
effects over the entire population.

\medskip We suggest here a possible route for extrapolation, based on a
random coefficient linear model of the form:
\begin{equation}
Y_t(x)=\delta _{t}+U_{0t}+x U_{1t}.  \label{eq:random_coeff_model}
\end{equation}%
Therefore, we impose a linear structure on $g_t$ ($g_t(u)=\delta_t+u$ and $%
U_t(x)$ ($U_t(x)=U_{0t}+x U_{1t}$). The model still allows for a rich,
non-scalar heterogeneity pattern through the two unobserved terms $U_{0t}$
and $U_{1t}$. Under this structure, we have, for any $(x,x^{\prime })\in
\text{Supp}(X_T)^2$, $x\neq x^{\prime }$,
\begin{equation}
\frac{\Delta ^{ATT}(x,q_t(x))}{q_t(x)-x} = E\left[ U_{1T}|X_T=x\right] =
\Delta^{AME}(x) = \frac{\Delta ^{ATT}(x,x^{\prime })}{x^{\prime }-x}.
\label{eq:ident_linearity}
\end{equation}
By Theorem \ref{thm:point_ident}, $\Delta ^{ATT}(x,q_t(x))$ is point
identified under Assumptions \ref{hyp:stationarity}-\ref{hyp:model}. This
implies that $\Delta^{AME}(x)$ and $\Delta ^{ATT}(x,x^{\prime })$ are
identified as well, whenever $q_t(x)\neq x$. As a result, the average
marginal effect over the whole population, $\Delta^{AME}=E\left[%
\Delta^{AME}(X_T)\right]$, is also point identified if $q_t(X_T)\neq X_T$
almost surely. We summarize this finding in the following theorem.

\begin{theorem}
Under Assumptions \ref{hyp:stationarity}-\ref{hyp:model} and Equation %
\eqref{eq:random_coeff_model}, for all $t<T$ and $(x,x^{\prime })\in \text{%
Supp}(X_T)^2$, $q_t(x)\neq x$, $(\delta_t)_{t<T}$, $\Delta^{ATT}(x,x^{\prime
})$ and $\Delta^{AME}(x)$ are identified. If $q_t(X_T)\neq X_T$ almost
surely, $\Delta^{AME}$ is point identified as well. \label%
{thm:full_ident_linearity}
\end{theorem}

Several remarks on this result are in order. First, we recover the same
parameter as \cite{Graham12}, who also consider a random coefficient linear
model similar to \eqref{eq:random_coeff_model}. They obtain identification
with panel data, relying on first-differencing. Compared to them, we rely on
variations in the cdf of $X_{t}$ rather than on individual variations. We
rely on a different, non-nested, restriction on the distribution of the
error term. In particular, for the same individual, $U_{1t}-U_{1s}$ could be
correlated with $X_{t}$ in our framework.

\medskip Second, Theorem \ref{thm:full_ident_linearity} readily extends to a
multivariate treatment, by just replacing the condition $q_t(x)\neq x$ by a
rank condition. Specifically, let, as in Section \ref{sub:multi}, $q_t(x)=
(q_{1t}(x_1),...,q_{kt}(x_k))^{\prime }$ and define the matrix $\mathbf{Q}%
(x) $ by
\begin{equation*}
\mathbf{Q}(x) = \left[%
\begin{array}{c}
(q_1(x)-x)^{\prime } \\
\vdots \\
(q_{T-1}(x) - x)^{\prime }%
\end{array}%
\right].
\end{equation*}
Then $\Delta^{AME}(x)$ and $\Delta^{ATT}(x,x^{\prime })$ are identified if $%
\mathbf{Q}(x)$ is full column rank. Note that the rank condition implies
that $T-1 \geq k$. It also implies that the distribution of $X_t$ differs at
each date, so that $q_s(x) \neq q_t(x)$. It makes sense that with several
endogenous variables, more time variation on $X_t$ is needed to identify
causal effects.

\medskip Third, coming back to the univariate case, Theorem \ref%
{thm:full_ident_linearity} ensures that all parameters of interest are
identified with only two time periods. This suggests that the model can be
either tested or enriched when $T>2$. To see why the linearity assumption is
testable when $T>2$, note that Equation \eqref{eq:ident_linearity} implies
\begin{equation*}
\frac{\Delta ^{ATT}(x,q_{s}(x))}{q_{s}(x)-x}=\frac{\Delta ^{ATT}(x,q_{t}(x))%
}{q_{t}(x)-x}\quad \forall s\neq t,
\end{equation*}
which can be checked in the data. With more than two time periods, we can
also identify treatment effects in the more general random coefficient
polynomial model of order $T-1$:
\begin{equation}
Y_{t}=\delta _{t}+U_{0t}+U_{1t}X_{t}+...+U_{T-1t}X_{t}^{T-1}.
\label{eq:random_polynomial}
\end{equation}%
With the same arguments as above, we recover not only average marginal
effect, but actually $E(U_{kt}|X_{t}=x)$ for all $k=1,...,T$ and all $x$
such that $(x,q_1(x),...,q_{T-1}(x))$ are all distinct. Identification of a
model similar to \eqref{eq:random_polynomial} was studied before by \cite%
{Florens08}, with cross-sectional data and under assumptions that typically
rule out discrete instruments
\citep[see also][for a
study of the identification of Model \eqref{eq:random_coeff_model} with
instruments]{Heckman98}. Here, we rely only on a finite number of time
periods, which would be equivalent to a discrete instrument, and allow for
time trend, which would correspond to a direct effect of the instrument in
\cite{Florens08}.

\medskip Alternatively, we can use additional periods to identify higher
moments of the distribution of the coefficients in the linear model %
\eqref{eq:random_coeff_model}. For instance, with $k=1$, $V(U_{01}|X_{T}=x)$%
, $V(U_{1T}|X_{T}=x)$ and Cov$(U_{01},U_{1T}|X_{T}=x)$ can be shown to be
identified with $T=3$ as soon as $x,q_1(x)$ and $q_2(x)$ are distinct.

\section{Estimation of Average and Quantile Treatment Effects}

\label{sec:estimation}

We consider in this section estimators of the parameters $\Delta^{ATT}(x,
q_t(x))$ and $\Delta^{QTT}(p, x, q_t(x))$ that are shown to be identified in
Theorem \ref{thm:point_ident}. We suppose for that purpose to observe two
independent samples corresponding to the periods $1$ and $T=2$. For
simplicity, we suppose hereafter that the two corresponding sample sizes are
identical.

\begin{assumption}
\label{a:iid} We observe the two independent samples $(Y_{i1},
X_{i1})_{i=1...n}$ and $(Y_{i2}, X_{i2})_{i=1...n}$, which are both i.i.d.
random variables drawn from the distributions $F_{Y_1,X_1}$ and $F_{Y_2,X_2}$%
, respectively.
\end{assumption}

Our estimator follows closely our identification strategy. Let us define
\begin{equation*}
\Psi_n(x) = \widehat{F}_{X_2}(x) - \widehat{F}_{X_1}(x),
\end{equation*}
where $\widehat{F}_{X_2}$ (resp. $\widehat{F}_{X_1}$ ) denotes the empirical
cdf of $X_2$ (resp. $X_1$). We first estimate $x^*_1$ by
\begin{equation}
\widehat{x}^\ast_1 = \min\left\{x \in \left[\widehat{F}_{X_1}^{-1}(%
\underline{p}), \widehat{F}_{X_1}^{-1}(\overline{p})\right]: \,
|\Psi_n(x)|\leq |\Psi_n(x^{\prime })| \; \forall x^{\prime }\in \left[%
\widehat{F}_{X_1}^{-1}(\underline{p}), \widehat{F}_{X_1}^{-1}(\overline{p})%
\right] \right\},  \label{eq:x_hat}
\end{equation}
where $\widehat{F}_{X_1}^{-1}$ denotes the empirical quantile function and $%
0<\underline{p}<\overline{p}<1$ are two given constants used to avoid
reaching the boundaries of the support of $X_1$. Note that the minimum in %
\eqref{eq:x_hat} is well defined because $\Psi_n$ is a right-continuous step
function.

\medskip Next, we estimate $q_{1}(x)=F_{X_{1}}^{-1}\circ F_{X_{2}}(x)$ by
its empirical counterpart $\widehat{q}_{1}(x)=\widehat{F}_{X_{1}}^{-1}\circ
\widehat{F}_{X_{2}}(x)$. We then estimate $g_{1}$ using an empirical
counterpart of \eqref{eq:ident_m}. For that purpose, we estimate the
conditional cdf $F_{Y_{t}|X_{t}}$, for $t\in \{1,2\}$, by
\begin{equation*}
\widehat{F}_{Y_{t}|X_{t}}(y|x)=\frac{\sum_{i=1}^{n}\mathds{1}\{Y_{it}\leq
y\}K\left( \frac{x-X_{it}}{h_{n}}\right) }{\sum_{i=1}^{n}K\left( \frac{%
x-X_{it}}{h_{n}}\right) },
\end{equation*}%
where $K$ is a kernel function and $h_{n}$ denotes the bandwidth. We then
let $\widehat{F}_{Y_{t}|X_{t}}^{-1}(.|x)$ denote the generalized inverse of $%
\widehat{F}_{Y_{t}|X_{t}}(.|x)$. We estimate $g_{1}$ by
\begin{equation*}
\widehat{g}_{1}(y)=\widehat{F}_{Y_{1}|X_{1}}^{-1}\left[ \widehat{F}%
_{Y_{2}|X_{2}}(y|\widehat{x}_{1}^{\ast })|\widehat{x}_{1}^{\ast }\right] .
\end{equation*}

Now, let us recall that $\Delta ^{ATT}(x,q_{1}(x))$ and $\Delta
^{QTT}(p,x,q_{1}(x))$ satisfy, under Assumptions \ref{hyp:stationarity}-\ref%
{hyp:crossing},
\begin{align*}
\Delta ^{ATT}(x,q_{1}(x))& =E[g_{1}(Y_{1})|X_{1}=q_{1}(x)]-E[Y_{2}|X_{2}=x],
\\
\Delta ^{QTT}(p,x,q_{1}(x))&
=F_{g_{1}(Y_{1})|X_{1}}^{-1}(p|q_{1}(x))-F_{Y_{2}|X_{2}}^{-1}(p|x).
\end{align*}%
We then estimate these two parameters by
\begin{align*}
\widehat{\Delta }^{ATT}(x,q_{1}(x))& =\frac{\sum_{i=1}^{n}\widehat{g}%
_{1}^{-1}(Y_{i1})K\left( \frac{x-X_{i1}}{h_{n}}\right) }{\sum_{i=1}^{n}K%
\left( \frac{x-X_{i1}}{h_{n}}\right) }-\frac{\sum_{i=1}^{n}Y_{i2}K\left(
\frac{x-X_{i2}}{h_{n}}\right) }{\sum_{i=1}^{n}K\left( \frac{x-X_{i2}}{h_{n}}%
\right) } \qquad\text{and} \\
\widehat{\Delta }^{QTT}(p,x,q_{1}(x))& =\widehat{F}_{\widehat{g}%
_{1}(Y_{1})|X_{1}}^{-1}(p|\widehat{q}_{1}(x))-\widehat{F}%
_{Y_{2}|X_{2}}^{-1}(p|x).
\end{align*}%
For notational simplicity, we chose here the same kernels and bandwidths for
each nonparametric terms, though we could obviously consider different ones.
We establish below that $\widehat{\Delta }^{ATT}(x,q_{1}(x))$ and $\widehat{%
\Delta }^{QTT}(p,x,q_{1}(x))$ are consistent and asymptotically normal. Our
result is based on the following conditions.

\begin{assumption}
\label{a:asymptotics_xstar} (Conditions for the root-n consistency of $%
\widehat{x}^\ast_1$ and $\widehat{q}_1(x)$) \newline
(i) There exists a unique $x^*_1$ satisfying $F_{X_1}(x^*_1)=F_{X_2}(x^*_1)
\in (0,1)$. Moreover, $F_{X_1}(x^*_1) \in (\underline{p},\overline{p})$.
\newline
(ii) For $t \in \{1,2\}$, $X_t$ admits a continuous density $f_{X_t}$
satisfying, for all $x$ in the interior of $\mathcal{X}$, $f_{X_t}(x)>0$.
Moreover, $f_{X_1}(x^*_1) \neq f_{X_2}(x^*_1)$.
\end{assumption}

\begin{assumption}
\label{a:Y} (Regularity conditions on $(X_t,Y_t)$) \newline
(i) For $t \in \{1,2\}$, $\text{Supp}(X_t,Y_t)= \mathcal{X} \times \mathcal{Y%
}$ with $\mathcal{Y}=[\underline{y},\overline{y}]$ with $-\infty<\underline{y%
}<\overline{y}<+\infty$.\newline
(ii) For $(t,x) \in \{1,2\}\times \mathcal{X}$, $F_{Y_t|X_t}(.|.)$ is
continuously differentiable and $\inf_{y \in \mathcal{Y}} f_{Y_t | X_t}(y |
x) > 0$. \newline
(iii) For all $(t,y) \in \{1,2\}\times \mathcal{Y}$, $F_{Y_t|X_t}(y|.)$ and $%
f_{X_t}$ are twice differentiable. $f_{X_t}$, $|f^{\prime }_{X_t}|$ and $%
|f^{\prime \prime }_{X_t}|$ are bounded. $\sup_{(y,x)\in \mathcal{Y} \times
\mathcal{X}}|\partial_x F_{Y_t|X_t}(y|x)|<\infty$ and $\sup_{(y,x)\in
\mathcal{Y} \times \mathcal{X}}|\partial_{xx} F_{Y_t|X_t}(y|x)|<\infty$.
\end{assumption}

\begin{assumption}
\label{a:kernel_bandwidth} (Conditions on the kernels and bandwidths)
\newline
(i) $nh^3_n/|\log(h_n)| \rightarrow +\infty$, $nh_n^5 \rightarrow 0$.
\newline
(ii) $K$ has a compact support, is differentiable with $K^{\prime }$ of
bounded variation and satisfies $K(y)\geq 0$ for all $y$. Besides, $\int
K(y)dy=1$ and $\int y K(y)dy=0$.
\end{assumption}

Assumption \ref{a:asymptotics_xstar}-(i) strengthens Assumption \ref%
{hyp:crossing} by assuming the uniqueness of the crossing point. We make
this assumption for the sake of simplicity. We could also consider the case
where the set of crossing points is an interval. As discussed in Section \ref%
{ssub:time_trend} above, we would actually expect a parametric rather than a
nonparametric rate of convergence for $\widehat{g}_1$, so Theorem \ref%
{thm:convergence_rate_att} below should still hold in this more favorable
case. Assumption \ref{a:asymptotics_xstar}-(ii) is a mild regularity
condition on $F_{X_{2}}$ and $F_{X_{1}}$. As Lemmas \ref%
{lemma:convergence_rate_xstar} and \ref{lemma:convergence_rate_q_1_x} in
Appendix A show, these two restrictions ensure that $\widehat{x}_{1}^{\ast }$
and $\widehat{q}_{1}(x)$ are root-n consistent. Assumption \ref{a:Y}
provides a set of conditions ensuring that $\widehat{g}_{1}$ is consistent
and asymptotically normal. Conditions (i) and (ii) are also made by \cite%
{Athey06}, without any $X_{t}$ in their case, in another context where
quantile-quantile transforms must be estimated. Condition (iii) is required
as well here because we deal with nonparametric estimators of conditional
cdfs rather than usual empirical cdfs, as \cite{Athey06} do. Finally,
Assumption \ref{a:kernel_bandwidth} is a standard condition on the
bandwidths and the kernels appearing in the nonparametric estimators. We
impose $nh_{n}^{5}\rightarrow 0$ in order to avoid any asymptotic bias on $%
\widehat{\Delta }^{ATT}(x,q_{1}(x))$ and $\widehat{\Delta }%
^{QTT}(p,x,q_{1}(x))$.

\begin{theorem}
\label{thm:convergence_rate_att} Suppose that Assumptions \ref%
{hyp:stationarity}-\ref{hyp:model} and \ref{a:iid}-\ref{a:kernel_bandwidth}
are satisfied. Then, for any $x\in \mathcal{X}$ such that $F_{X_1}$ is
differentiable at $q_1(x)$ with $F_{X_1}^{\prime }(q_1(x))>0$,
\begin{align*}
\sqrt{n h_n}\left(\widehat{\Delta}^{ATT}(x,q_1(x)) -
\Delta^{ATT}(x,q_1(x))\right) & \overset{d}{\longrightarrow} \mathcal{N}%
(0,V_1) \\
\sqrt{n h_n}\left(\widehat{\Delta}^{QTT}(p, x,q_1(x)) - \Delta^{QTT}(p,
x,q_1(x))\right) & \overset{d}{\longrightarrow} \mathcal{N}(0,V_2),
\end{align*}
for some $V_1,V_2$.
\end{theorem}

We do not display the asymptotic variances here, as they involve many terms
due to the multiple compositions of nonparametric estimators -- see Appendix %
\ref{sec:convergence_rate_att} for details as well as a proof. In practice,
we suggest to rely on bootstrap, as we do in the application below. We
conjecture that the bootstrap is consistent in our setting, though a formal
proof of its validity is beyond the scope of this paper. The main issue for
establishing its validity would be to prove the (conditional) weak
convergence of the process
\begin{equation*}
G_{nxt}^* = \sqrt{nh_n}\left(\widehat{F}^*_{Y_t|X_t}(.|x)-\widehat{F}%
_{Y_t|X_t}(.|x)\right), \; t \in \{1,2\},
\end{equation*}
where $\widehat{F}^*_{Y_t|X_t}$ is the bootstrap counterpart of $\widehat{F}%
_{Y_t|X_t}$. Up to our knowledge, such a result is not available in the
literature yet.

\section{Application to the Marginal Propensity to Consume}

\label{sec:application}

In this section we provide an application to a substantive economic
question: The magnitude of the marginal propensity to consume out of current
disposable income. When analyzing this question, we focus in particular on
how results obtained through our approach compare to those obtained in the
literature. In order to facilitate this comparison, we first briefly review
the literature on this question, before explaining the policy experiment we
are using, and detailing the data. We then outline how our methodology is
employed, and finally close by comparing our results with those in the
literature.

\subsection{The Economic Question}

A crucial question for the classical theory of consumption is the marginal
propensity to consume (MPC) out of income. Given its implications for the
business cycle, taxes, and government policy, the importance of the MPC can
hardly be overstated, and thus this quantity was, and still is, at the
center of a very active debate
\citep[see, e.g.,][for an
overview]{jappelli2010consumption}. An upshot of the rational expectations
revolution which, since the seminal paper of \cite{hall1978stochastic},
tried to answer questions about the effect of a marginal change in income on
consumption, is that expectations about the change matter.

\medskip In the absence of liquidity constraints (and precautionary saving
motives at very low income levels), the following is the key insight in the
literature about the effect of a marginal income change on the nondurable
consumption of a rational consumer, see, e.g., \cite{deaton1992understanding}%
: If the income change is anticipated, i.e., not related to new information,
then consumption does not respond to the income change. For an income change
that is not anticipated, if the change is viewed as transitory, then the
rational consumer is predicted to use very little of the income increase
immediately, as the transitory change in income is distributed over the
life-cycle, and its small quantity (relative to life-cycle income) does not
alter fundamentally the trade-off between consumption today and saving for
the future. Conversely, if the income change is expected to be permanent,
the individual is expected to essentially increase her consumption by the
amount of the change. This means that we only observe a substantial change
in consumption in response to an income change, if the change is surprising
and considered to be permanent.

\medskip The empirical evidence on the hypothesis of a rational consumer is
rather mixed, and has spurned an active debate. Perhaps the most problematic
evidence comes from studies involving one time transfers, see e.g., \cite%
{johnson2006household} and \citeauthor{parker2013consumer} (2013, PSJM). In
these studies, consumers are given what is clearly an expected and
transitory income shock (PSJM actually documenting aspects of the Obama era
stimulus package), yet the effect on consumption is not zero. Instead,
typical estimates for the marginal effects of an anticipated income change
range between 15\% and 25\%.

\medskip There are a number of counterarguments in defense of the rational
consumer. First, consumers could be credit constrained. PSJM find indeed
lower responses for older and high-income households, who are less likely to
be constrained. Second, consumers may exhibit a form of bounded rationality.
There are significant costs associated with computing the optimal
consumption path. If an income change is small relative to the level of
income, the benefits from adapting the optimal path in light of the changes
are small relatively to the costs associated with it, and individuals simply
avoid optimizing completely, as they would in the case of a large income
change. Evidence that individuals indeed smooth large anticipated income
changes is provided by \cite{browning2001response} and \cite%
{hsieh2003consumers}, among others. Another counterargument is that some of
the changes considered in the literature are not just small, but also
outside the \textquotedblleft usual\textquotedblright\ consumer experience.
As such, they are not representative of the typical real-world surprise
income shocks individuals deal with (a distinction that is reminiscent to
the question of whether individuals are able to assign probabilities to
these events).

\medskip In this section, we use our econometric method in conjunction with
an experiment involving the Earned Income Tax Credit (EITC) to analyze the
causal effect of increase in income on consumption for households in 1987.
We believe that this natural experiment is very insightful for the above
debate. While it provides exactly the type of variation we require for our
method, it provides (at least for a good number of households) a significant
and anticipated change in their income. Finally, the fact that our procedure
allows for nonlinearities, i.e., for the marginal effect to vary with
income, is going to be crucial to shed light on the question of the
existence of liquidity constraints.

\subsection{Policy Background: The EITC}

\label{sub:policy_background}

In the following, we provide more background on the policy experiment that
provides the exogenous variation: The Earned Income Tax Credit (EITC) is an
income support program which started in 1975 in the United States for the
purpose of mitigating poverty. The EITC provision schedule varies from year
to year, exhibiting interesting non-linearities. This feature of the program
has been used for economic analysis before, e.g., by \cite{dahl2012impact},
and a detailed documentation of the EITC can be found in \cite%
{falk2014earned}. In most of the past years, the change to the EITC schedule
has been monotone to match increasing price levels. However, the change in
the schedules between 1987 and 1989 exhibits a specific pattern which, as we
will now demonstrate, generates a crossing of the cdfs of (deflated) total
income in the respective years.

\medskip Figure \ref{fig:eitc_schedules} displays the EITC schedules in 1987
(solid line) and 1989 (dotted line) in terms of thousands of Year 2000 US
dollars for families with two or more children. Note that for individuals
with income between 9K USD and 10.75K USD, the 1987 EITC provision was
higher than the provision in 1989, whereas the reverse is true for
individuals with income above 10.75K USD. This is exactly the type of
variation which generates a crossing, if everything else is held constant.

\begin{figure}[]
\centering
\includegraphics[width=0.6\textwidth]{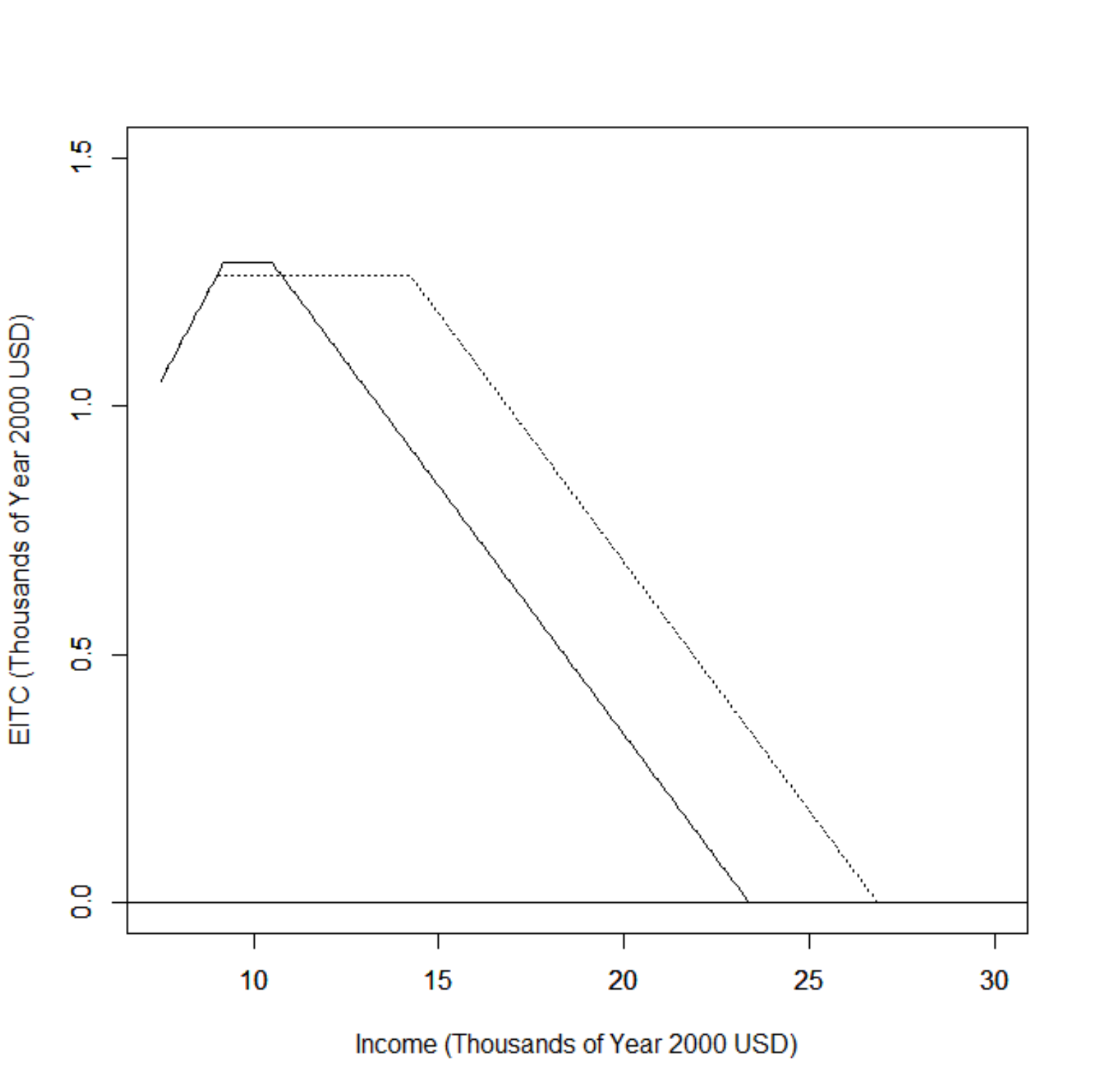}
\begin{minipage}{0.96\textwidth}
{\footnotesize ~\\[2mm]
Notes: amount of the EITC in 1987 (solid line) and 1989 (dotted line), for families with two or more children.}
\end{minipage}
\caption{EITC schedules in 1987 and 1989}
\label{fig:eitc_schedules}
\end{figure}

To see this more precisely, consider the left graph in Figure \ref%
{fig:eitc_plus_income}. The graph shows total income, obtained as the sum of
the pre-aid income and the EITC amount, for each of the years 1987 (solid
line) and 1989 (dotted line) plotted against that of year 1989, i.e., the
solid line is the 45-degree line. The right graph in the same figure (Fig. %
\ref{fig:eitc_plus_income}) focuses on this difference. As these figures
suggest, we expect a crossing at 12K USD, computed as the sum of 10.75K USD
(the cut-off for the change in the schedule) and 1.25K USD for the
corresponding EITC amount, provided that total pre-EITC income does not
change substantially. Note that these figures are solely derived from the
known policy schedules, but we will confirm our expectation with real data
below. Before we detail this, however, we first give an overview of the data.

\begin{figure}[]
\centering
\begin{minipage}{0.48\textwidth}
	\includegraphics[scale=0.55]{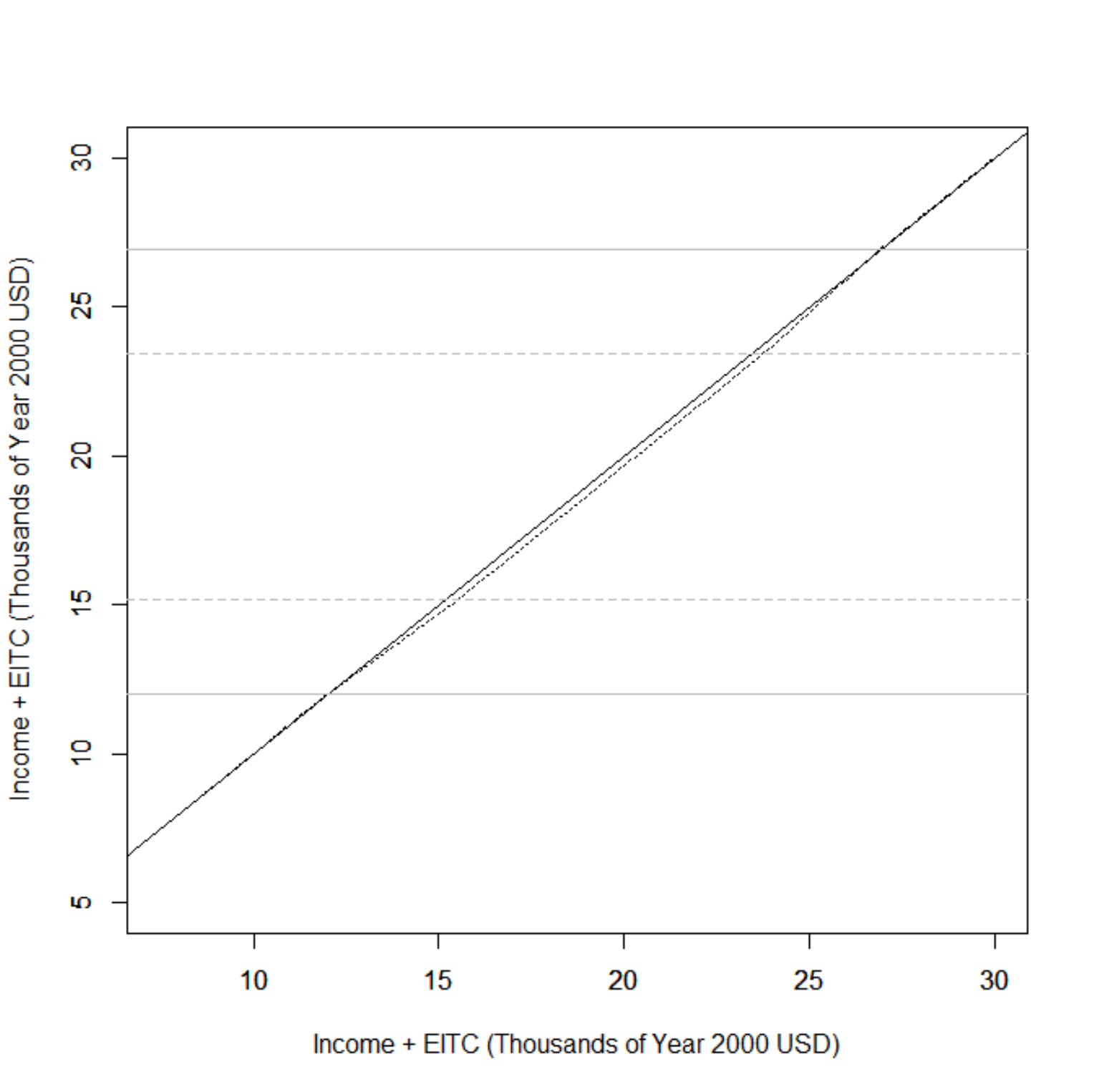}
\end{minipage}
\begin{minipage}{0.48\textwidth}
	\includegraphics[scale=0.55]{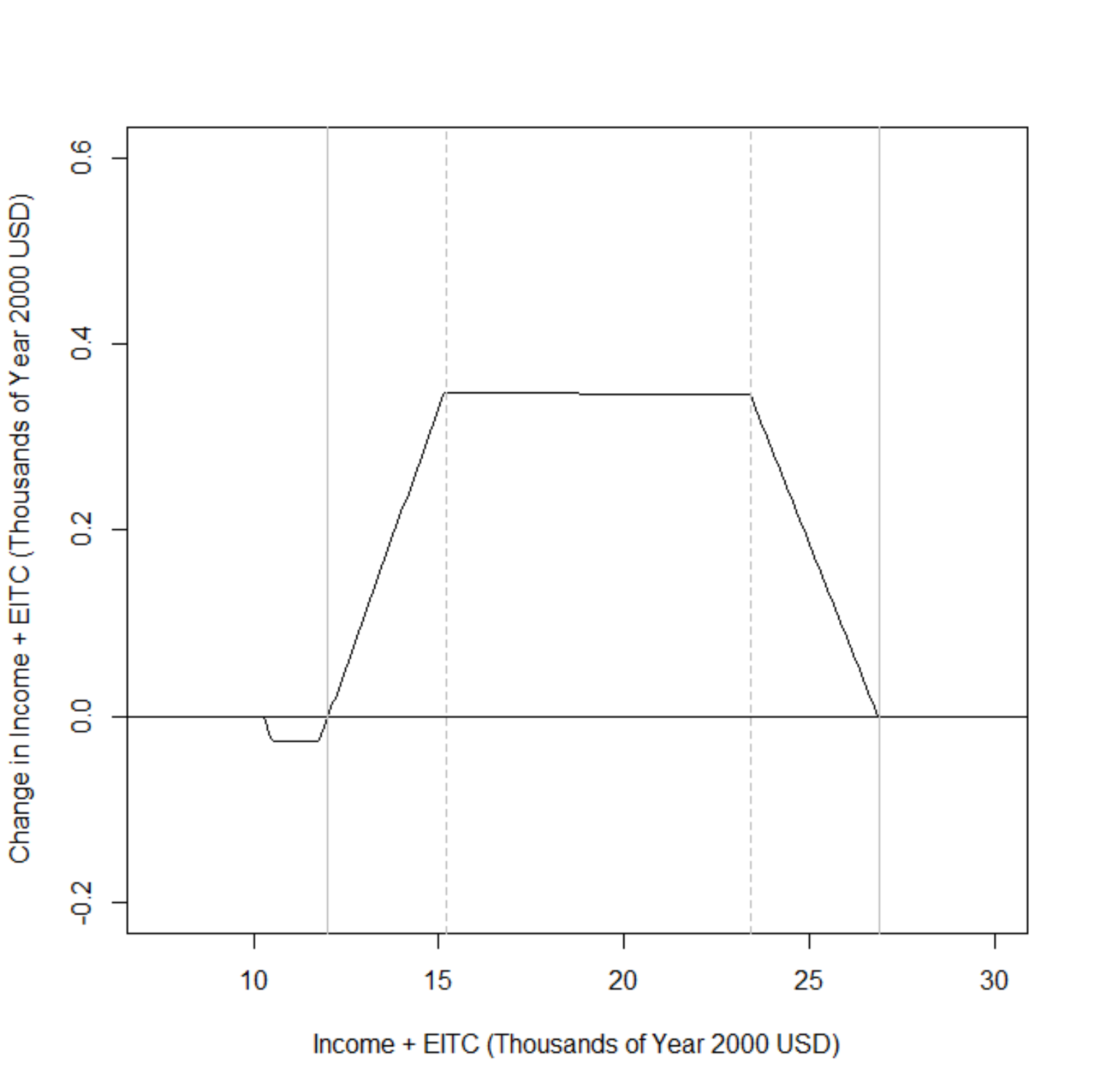}
\end{minipage}	
\begin{minipage}{0.96\textwidth}
{\footnotesize ~\\[2mm]
Notes: left panel: total income, obtained as the sum of the pre-aid income and the EITC amount, for 1987 and 1989 plotted against that of 1989. Right panel: the change in the total income, obtained as the sum of the pre-aid income and the EITC amount, between 1987 and 1989.}
\end{minipage}	
\caption{Theoretical change in total disposable income between 1987 and 1989
due to EITC change.}
\label{fig:eitc_plus_income}
\end{figure}


\subsection{Data: The CEX}

\label{sub:data}

For our analysis, we use repeated cross-sectional data from the Consumer
Expenditure Survey (CEX) for the calendar years 1987 and 1989. The treatment
variable is, more precisely, total disposable family income measured in
thousands of Year 2000 US dollars. The outcome (dependent) variable is
non-durable household consumption, defined as the sum of expenditures for
food at home, apparel, health, entertainment, personal care, and readings,
measured in thousands of Year 2000 US dollars. Since the policy described
above applies only to families with two or more children, we use the
sub-sample of individuals with two or more children. Table \ref{tab:summary_stat} shows summary statistics for our sub-sample.

\medskip Note that after controlling for inflation (i.e., in year 2000
prices), the mean of total family disposable income does not change
substantially between 1987 and 1989 (roughly 2\%). Indeed, this modest
increase from 1987 to 1989 is quite consistent with the EITC policy change
and an otherwise pretty stationary environment, strengthening the case that
we should expect to have the type of variation in cdfs our method requires%
\footnote{%
As a caveat, we remark that not all families take up the aid even if
eligible, and that only a part of the population of families is eligible,
which together accounts for the modest 2\% increase in mean total family
income from 1987 to 1989.}.

\begin{table}[]
\centering%
\scalebox{0.97}{
\begin{tabular}{lccccc}
\hline\hline
& \multicolumn{2}{c}{Thousands of present year USD} &  & \multicolumn{2}{c}{Thousands of Year 2000 USD}
\\ \cline{2-3}\cline{5-6}
Whole Sample
& 1987 & 1989 &  & 1987 & 1989 \\ \hline
Total Family Income
& 27.973 & 31.162 &  & 42.402 & 43.275 \\
&(20.944)&(23.438)&  &(31.747)&(32.548) \\ \hline
CEX Nondurable Consumption
& 9.072 &10.787 &  &13.752 & 14.980 \\
&(6.187)&(8.872)&  &(9.378)&(12.320) \\
\hline
Number of Observations & 4,827 & 4,120 &  & 4,827 & 4,120 \\
\\
Subsample: Total Disposable
& \multicolumn{2}{c}{Thousands of present year USD} &  & \multicolumn{2}{c}{Thousands of Year 2000 USD}
\\ \cline{2-3}\cline{5-6}
Family Income $\in [15.2, 23.4]$
& 1987 & 1989 &  & 1987 & 1989 \\ \hline
CEX Nondurable Consumption
& 6.132  & 7.426 & & 9.296 & 10.312 \\
&(3.606)& (5.655) & &(5.467)& 7.854 \\
\hline
Consumption/Income Ratio
& 0.491 & 0.536 & & 0.491 & 0.536 \\
&(0.289)& (0.386) & &(0.289)& (0.386) \\
\hline
Number of Observations & 559 & 442 & & 559 & 442\\
\hline\hline
\multicolumn{6}{p{470pt}}{\footnotesize Notes: CEX data restricted to
famiies with two or more children for 1987 and 1989. The standard deviations
are indicated in parentheses.}
\end{tabular}
}
\caption{Relevant Summary Statistics of the CEX Data}
\label{tab:summary_stat}
\end{table}

\medskip Turning to our nondurable consumption measure, we first notice that
it only captures a little less than half of disposable income. Within the
subsample that we focus on, the average ratio of nondurable consumption to
total disposable income (which is different than the ratio of averages) is
around 50\%. This may be due to the fact that the large category of rent and
mortgage payments are excluded as are large and durable and nondurable
consumption items (e.g., TVs, cars, phones). However, we also suspect a
certain modest degree of underreporting in the data. Like in the standard
Diff-in-Diff approach, our analysis would be invalidated if the evolution of
this underreporting is systematically different between treatment and
control group. We believe this to be unlikely and certainly have no evidence
of this difference in effects. Moreover, since the overall degree of
underreporting seems to be tolerable as well (e.g., food and clothing
account for a budget share of 50\% in the British FES as well, see Hoderlein
(2011)), we hence proceed with our analysis.

\medskip One thing that stands out is that the nondurable consumption
measure increased more than proportionally to the change in disposable
income in both the sample we focus on (average share increase from 49.1\% to
53.6\%), but also in the population at large. This may be due to changes in
economic outlook and general optimism in 1989 at the end of the cold war.
Because of this observation, we definitely want to include a time trend $%
g_{t}$ in the empirical analysis, as our method warrants. Indeed, our model
identifies an increase in nondurable consumption in particular at higher
levels of the consumption distribution even if our policy experiment would
not have taken place.

\begin{figure}[]
\centering
\begin{minipage}{0.48\textwidth}
	\includegraphics[scale=0.22]{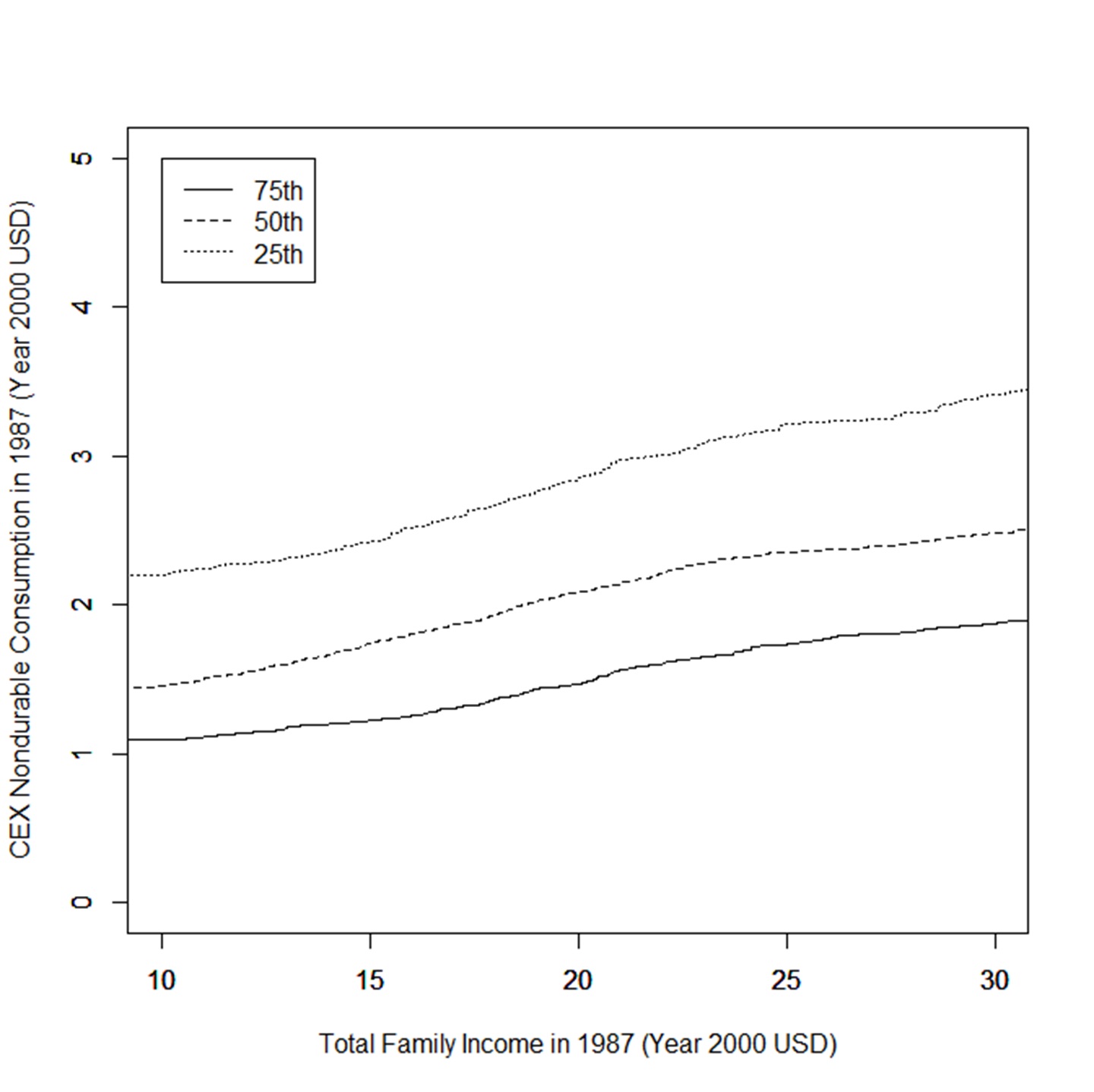}
\end{minipage}
\begin{minipage}{0.48\textwidth}
	\includegraphics[scale=0.22]{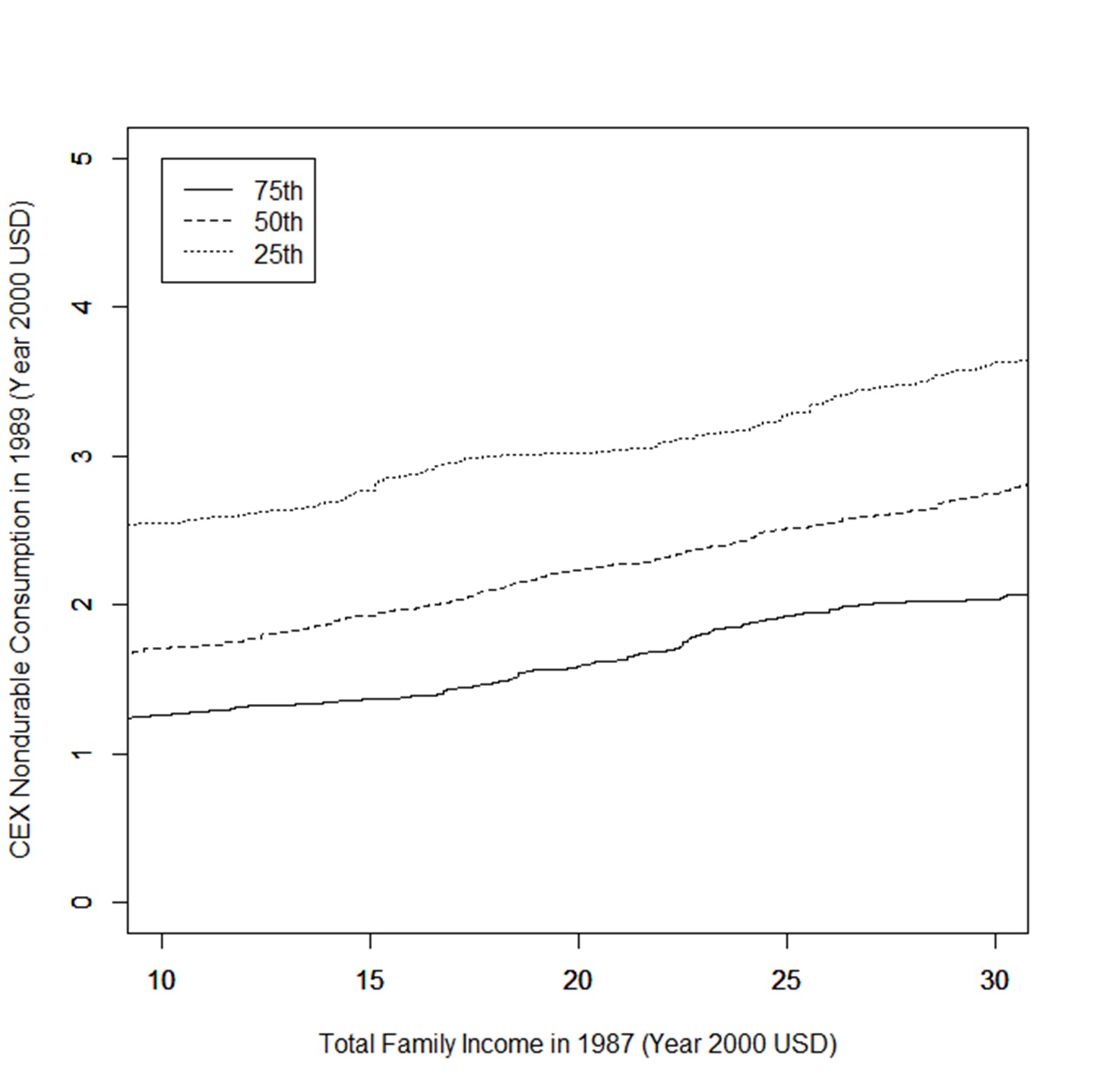}
\end{minipage}	
\begin{minipage}{0.96\textwidth}
{\footnotesize ~\\[2mm]
Notes: CEX data restricted to families with two or more children. Family income is given in thousands of 2000 US dollars.}
\end{minipage}
\caption{Conditional quartiles of $Y_t$ (CEX nondurable consumption) given $%
X_t$ (total family income) in 1987 and 1989}
\label{fig:empirical_conditional quartiles}
\end{figure}


\subsection{Analysis and Results}

\label{sub:analysis}

First, we use our data to confirm that the policy change in the EITC
described above indeed induces a crossing in the cdfs. In particular, we
want to study whether there is a divergence from 12.0 K to 26.8 K USD of
cdfs of total family income between 1987 and 1989. The left panel of Figure %
\ref{fig:empirical_cdfs} displays the two empirical cdfs. The solid vertical
lines indicate the limit points of the range inside which the policy change
matters; these lines correspond to those displayed in Figure \ref%
{fig:eitc_plus_income}. To check that the distributions of income are in
line with this policy change, we made one-sided test of $F_{1}(x)\leq
F_{2}(x)$ for all $x\in \lbrack 12.0K,26.8K]$. We find that at the 5\%
level, $F_{1}(x)>F_{2}(x)$ for at least some $x\in \lbrack 12.0K,26.8K]$. We
take this as strong evidence that the change in the EITC was, at least for
this subpopulation of households, the main driving force in the change of
the empirical cdfs between the two years. Moreover, the direction of the
crossing is what we expect from the design of the policy change: the
families falling within the range where we expect an increase in total
disposable income due to the change in EITC experience a positive change in
total family income between 1987 to 1989.

\medskip Using these two empirical cdfs, we next compute the empirical
quantile-quantile plot of the total family income from 1987 to 1989 in terms
of Year 2000 US dollars. The right panel of Figure \ref{fig:empirical_cdfs}
displays the plot. Observe how well this data-based figure resembles Figure %
\ref{fig:eitc_plus_income}, which is constructed using the policy formulas.
This provides further evidence that the data follows our research design,
and that there are no other major unaccounted sources of change in
disposable income. Recall, moreover, that this quantile-quantile plot, which
is mathematically represented by $q_{1}$ in our framework, is the main
building block for our identification results.

\begin{figure}[]
\centering
\begin{minipage}{0.48\textwidth}
	\includegraphics[scale=0.22]{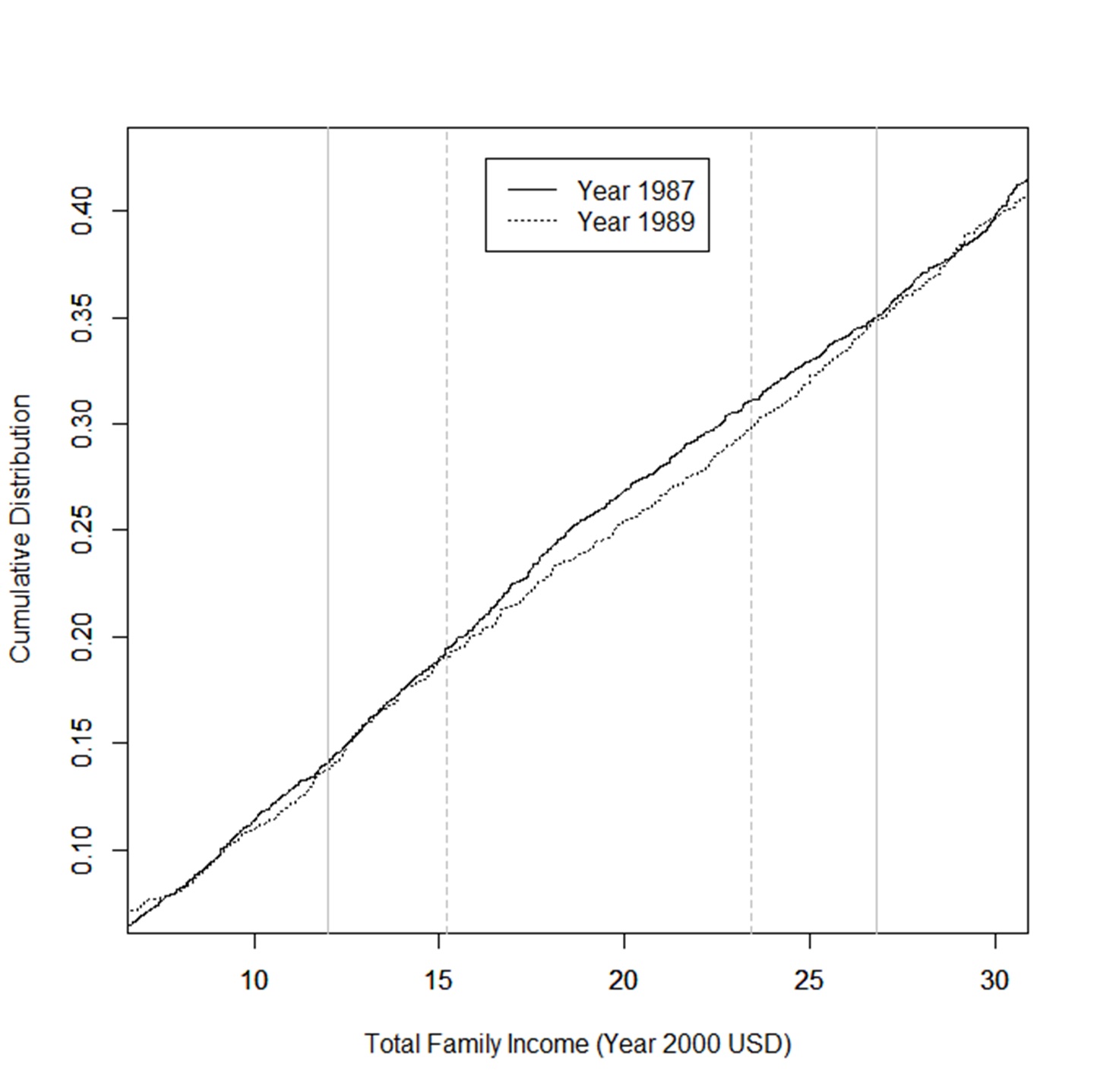} 
\end{minipage}
\begin{minipage}{0.48\textwidth}
	\includegraphics[scale=0.22]{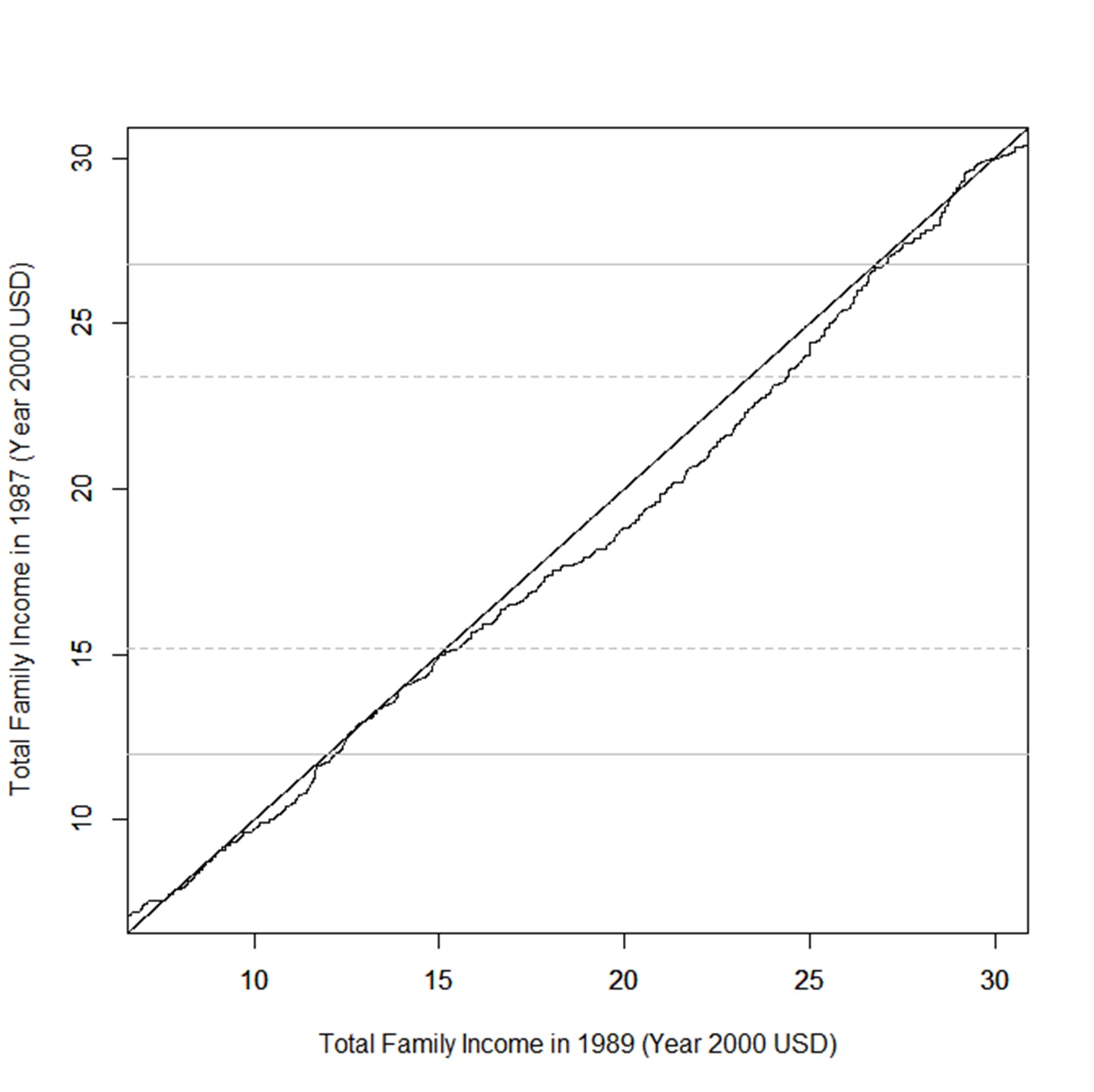} 
\end{minipage}	
\begin{minipage}{0.96\textwidth}
{\footnotesize ~\\[2mm]
Notes: CEX data restricted to families with two or more children. Family income is given in thousands of 2000 US dollars. In the left panel, the black (resp. grey) curve corresponds to family income in 1987 (resp. 1989). In the right panel, we display the Q-Q plot, i.e. $q_1$ against the identity function. The solid lines indicate the theoretical limits inside which we should observe a divergence of the cdfs, given the policy design. The dotted lines are the limits of the interval on which the efffect of the policy is supposed to be maximal (see the right panel of Figure \ref{fig:eitc_plus_income}).}
\end{minipage}
\caption{Cdf's and Q-Q plot of total family income in 1987 and 1989}
\label{fig:empirical_cdfs}
\end{figure}

\medskip After having confirmed that the change in the distribution of the
treatment is in line with our modeling assumption and largely driven by the
policy change, we proceed to use our framework and estimate the time trend $%
g_{1}(.)$. In line with the theoretical design, Figure \ref%
{fig:empirical_cdfs} shows that we have more than a single point $x^{\ast }$
as a control group. We can use the whole set $\mathcal{S}=[10K;12K]\cup
\lbrack 26,8K;50K]$, where the two cdfs overlay. This results in more
precise estimates of $g_{1}(.)$ and marginal effects, because we can use the
whole set $\mathcal{S}$ instead of a single point $x^{\ast }$. Specifically,
we can use $g_{1}(y)=F_{Y_{1}|X_{1}\in \mathcal{S}}^{-1}\left[
F_{Y_{2}|X_{2}\in \mathcal{S}}(y)\right] $ instead of $%
g_{1}(y)=F_{Y_{1}|X_{1}}^{-1}\left[ F_{Y_{2}|X_{2}}(y|x^{\ast })|x^{\ast }%
\right] $. Figure \ref{fig:g_appli} displays the estimate of $g_{1}^{-1}$, which
corresponds to the (heterogeneous) time trend between 1987 and 1989. As
mentioned before, we observe an increase in the upper tail of the
distribution of nondurable consumption, corresponding with an improved
overall economic outlook, in particular for middle and upper class
households.

\begin{figure}[]
\centering\includegraphics[width=0.75\textwidth]{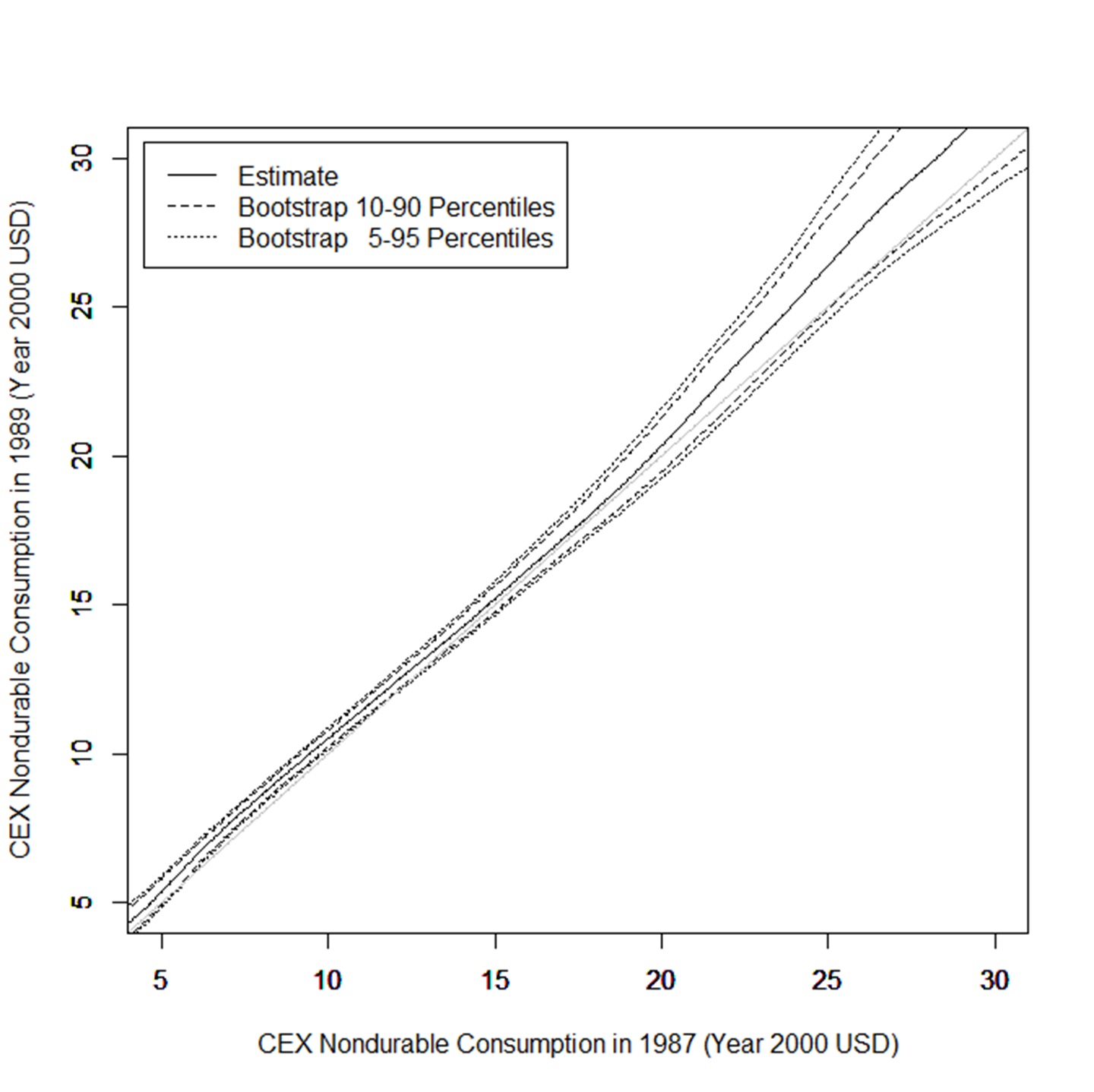}
\begin{minipage}{0.9\textwidth}
{\footnotesize ~\\[2mm]
Notes: CEX data restricted to families with two or more children. Curves are kernel-smoothed. CEX nondurable consumptions are in thousands of year 2000 US dollars.}
\end{minipage}
\caption{Estimate of the time trend function $g_{1}^{-1}$ from 1987 to 1989.}
\label{fig:g_appli}
\end{figure}

To come to the main purpose of this application, we estimate average
marginal effects of the total family income in 1987 in terms of Year 2000 US
dollars on various expenditures in terms of Year 2000 US dollars.
Specifically, we estimate $\Delta _{app}^{AME}(x)=\Delta
^{ATT}(x,q_{1}(x))/(q_{1}(x)-x)$ instead of $\Delta ^{ATT}(x,q_{1}(x))$. The
former quantity has the advantage over the latter of being interpretable as
an average marginal effect. By the mean value theorem (and under mild
regularity conditions), indeed, $\Delta _{app}^{AME}(x)=E\left[ dY_{2}/dx(%
\widetilde{X})|X_{2}=x\right] $, for some random $\widetilde{X}\in \lbrack
q_{1}(x),x]$. To the extent that $q_{1}(x)$ is close to $x$, we then
interpret $\Delta _{app}^{AME}(x)$ as the average marginal effect at $%
X_{2}=x $. Note, on the other hand, that by dividing by $\widehat{q}%
_{1}(x)-x $, the estimator of $\Delta _{app}^{AME}(x)$ is more volatile than
that of $\Delta ^{ATT}(x,q_{1}(x))$, especially when $q_{1}(x)-x$ is close
to zero. To obtain more precise estimates, we rely hereafter on a piecewise
linear estimator of $q_{1}(x)-x$. Such a constrained estimator is consistent
with the policy design and fits well the data. We refer to Appendix \ref%
{sec:details_app} for more details on its construction.

\medskip Figure \ref{fig:ame_estimates} presents the estimated average
marginal effects. The estimates are displayed on the interval $[15.2,23.8]$,
namely the interval on which the EITC policy change is supposed to be
pronounced. We focus on this region because elsewhere the denominator of $%
\Delta _{app}^{AME}(x)$ is either close or equal to zero. The solid line
represents the point estimate of the average marginal effect. Specifically,
the line shows how much out of one dollar increase is spent on our
nondurable consumption bundle. Our results are very much in line with the
literature, with values ranging from 0.5 for disposable income just below
\$16K to virtually zero for incomes above \$22K. Our point estimate also
suggests that the average marginal effect decreases with income. This is in
line with previous findings in the literature, in particular those of PSJM.
Such a pattern is also consistent with rational consumers facing credit
constraints. Indeed, credit constraints are likely to be less severe for
households with higher income, as such consumers are on average more able to
use parts of their wealth as collateral to get new credits more easily.

\begin{figure}[]
\centering
\includegraphics[width=0.6\textwidth]{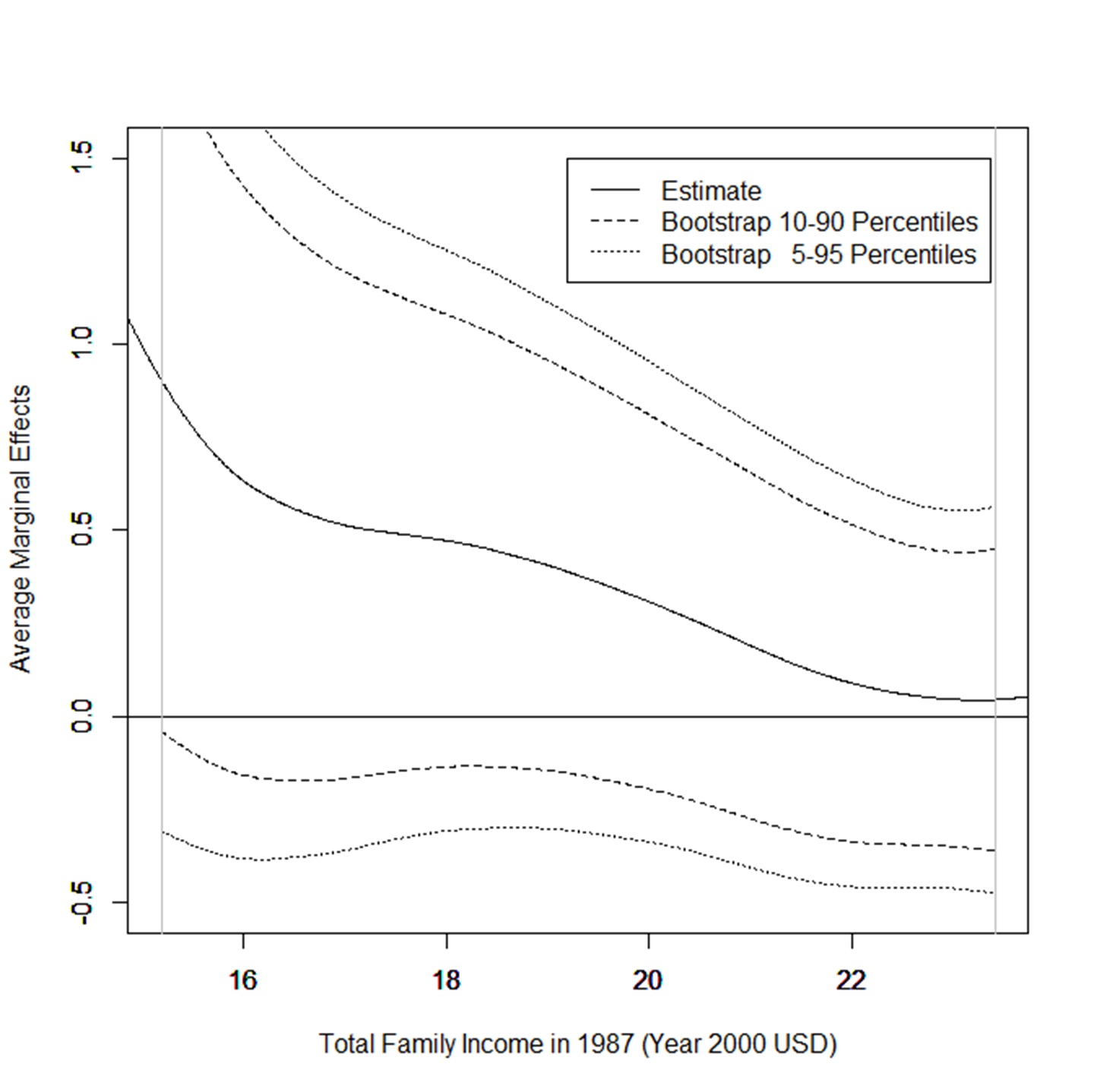}
\begin{minipage}{0.9\textwidth}
{\footnotesize ~\\[2mm]
Notes: CEX data restricted to families with two or more children. Curves are kernel-smoothed.  Income and consumption are in thousands of 2000 USD. The vertical lines indicate the limits of the region with cdf divergence.}
\end{minipage}
\caption{Average marginal effects of total family income on CEX nondurable
consumption}
\label{fig:ame_estimates}
\end{figure}

\medskip Several remarks are in order. The first concerns significance:
While the results for low levels of disposable income are borderline
pointwise significant at the 90\% level, most of the estimated effect is
insignificant (as are results based on 95\% significance). This is in
particular regrettable at income levels around \$ 18K where there is
probably a substantive nonzero effect, but the evidence is slightly too weak
to conclude this with statistical certainty. As already outlined above,
there is significant noise in the data that complicates our analysis and the
instrumental variation used to identify the model is only moderately strong.
Having said that, given the borderline significance at lower income levels,
we are confident that if we were to consider an estimator for the average
marginal effect across the region between \$16K and \$19K we would find a
strongly significant effect, because average derivatives are much more
accurately estimable than pointwise derivatives. Developing such a formal
test is quite involved and thus left for future research. Note that the
monotonically declining shape is very much in line with the literature which
finds the strongest evidence for the failure of intertemporal smoothing at
lower income. While certainly not as precise as we had hoped for, we feel
that our estimates lend support to the recently found evidence of
excessively large effects of an anticipated shock to income.

\medskip The second remark concerns our modeling assumptions. As mentioned
above, the stationarity assumption Assumption \ref{hyp:stationarity}  together with the modeling
assumption limits the degree of unobserved heterogeneity. In particular,
individual households might have heterogeneous preferences both for
consumption and leisure that enter in a complicated fashion resulting in a
multivariate $A_{t}$. While we acknowledge the possibility of these effects
biasing our results, we do not think that they are large in absolute size.
Labor supply of the main breadwinner, especially in families in the 1980s,
has proven to be very inelastic to the degree that wages are frequently used
as an instrument in consumer demand studies, see  \cite{blundell1993we}. This is less true for secondary income (e.g., part time work by the spouse). However, given the relatively small magnitude of the change,
we would be surprised if the effect on labor supply be large (which would be
the main channel for misspecification impacting our estimates). Still, we do
acknowledge that a cautionary remark is in order at this point, also with
respect to our omission of potentially complex dynamics as would arise,
e.g., with habit formation.

\medskip The third remark concerns our omission of observable heterogeneity.While
clearly important, as the paper does not develop the associated theory we
leave this for future research. Having said, note that we work with the
subsample of families with two or more children with at least (and typically
in 1987 also at most) one bread winner of a low income level which is a
fairly homogeneous population. A similar stratification strategy to deal
with observed heterogeneity is very common in the consumer demand literature
\citep[see][for a discussion]{hoderlein2011many}.

\medskip The last remark concerns the magnitude of the effect. Here, it is
instructive to compare the marginal effect with the average expenditure
share of our nondurable consumption measure. This share is roughly equal to
0.5 for the levels of income we consider.\footnote{%
It is also very mildly decreasing with income levels, as one could expect.}
Similarly, our results imply that at a disposable income level of 16.5., the
consumers spend roughly 50 cent out of an additional dollar on nondurable
consumption. This is compatible with a model where low income households,
when receiving an (anticipated) additional dollar of income, consume it
entirely and in roughly equal proportions on our set of nondurable
consumption goods as well as on the remaining (mostly durable) consumption
items. This points clearly to a violation of the hypothesis of rational
consumers. The marginal effect diminishes to near zero for higher income
levels. For incomes lower than 16.5, we find effects that are even larger
than 0.5, meaning that households spend a larger fraction of every
additional dollar on nondurable consumption than its income share. Since
durable consumption is illiquid, we view such an effect as entirely
conceivable, though we want to voice caution given the aforementioned large
level of noise in the data.

\medskip In sum, we interpret our evidence as favoring the recent findings
in the literature that low (disposable) income households spend large parts
of an anticipated and possibly transitory real world shock on consumption.
Conversely, they do not engage in intertemporal smoothing to the degree that
the theory of rational consumer behavior would predict. Again, very much in
parallel to recent findings, we also observe that this effect decreases with
increasing disposable income, meaning that the driver for the higher effects
at low levels is either liquidity constraints or a precautionary savings
motive.

\section{Conclusion}

We consider in this paper an extension of the change-in-change model of \cite%
{Athey06} to continuous treatments. We impose similar restrictions as theirs
on time effect and a crossing condition on the cdfs of the treatment
variable. This crossing condition may be seen as a generalization of the
existence of a control group in both the usual difference-in-difference and
change-in-change settings. Importantly, our framework can allow for
heterogeneous time trends and treatment effects. We show that under these
conditions, some average and quantile treatment effects are point
identified. We propose nonparametric multistep estimators of these treatment
effects and show their asymptotic normality. Finally, we apply our method to
the effect of disposable income on consumption. Our results suggest large
effects for low-income households, in line with recent empirical findings.

\pagebreak
\bibliographystyle{Chicago}
\bibliography{../Revision1/biblio}

\pagebreak \appendix

\begin{center}
{\huge Appendix}
\end{center}

\section{Point Identification of Usual Marginal Effects}

\label{sub:identification_ME}

We have focused in the paper on the effect of changes of the treatment from $%
x$ to $x^{\prime }$. Other popular effects are the following average and
quantile marginal effects:
\begin{align*}
\Delta ^{AME}(x)& \equiv E\left[ \frac{dY_{T}}{dx}(x)|X_{T}=x\right] \qquad%
\text{and} \\
\Delta ^{QME}(p,x)& \equiv \lim_{h\rightarrow 0}\frac{%
F_{Y_{T}(x+h)|X_{T}}^{-1}(p|x)-F_{Y_{T}(x)|X_{T}}^{-1}(p|x)}{h},
\end{align*}%
where we assume that the derivatives exist.

\medskip Intuitively, because the variations induced by time are discrete,
we cannot identify these parameters everywhere unless we impose additional
conditions, as in Section \ref{sub:point_ident} below. On the other hand, if
$x\simeq x^*_t$, $q_t(x)$ is also close to $x$. Then,
\begin{equation*}
\frac{Y_T(q_t(x))-Y_T(x)}{q_t(x)-x}\simeq \frac{\partial Y_T}{\partial x}%
(x_{t}^{\ast }).
\end{equation*}
Moreover, if the conditional distribution of $Y_T(x_t^*)$ is regular,
conditioning on $X_T=x$ becomes the same as conditioning on $X_T=x_{t}^{\ast
}$, so that
\begin{equation*}
\frac{\Delta ^{ATT}(x,q_t(x))}{q_t(x)-x}\simeq \Delta^{AME}(x_{t}^{\ast }).
\end{equation*}
Similarly,
\begin{equation*}
\frac{\Delta^{QTT}(p,x,q_t(x))}{q_t(x)-x}\simeq \Delta^{QME}(p,x_{t}^{\ast
}).
\end{equation*}
Formally, identification of these marginal effects is achieved on the set $%
\mathcal{X}_{0}$ defined by 
\begin{equation*}
\mathcal{X}_{0}=\bigg\{ x\in \mathbb{R}:\exists (t,(x_{n})_{n\in \mathbb{N}%
})\in \{1,...,T-1\}\times \left(\mathbb{R}\right)^{\mathbb{N}%
}:\;q_t(x)=x,\lim_{n\rightarrow \infty }x_{n}=x, \,q_t(x_{n}) \neq x_{n}%
\bigg\}.
\end{equation*}
$\mathcal{X}_{0}$ is the set of points $x$ such that $q_t(x)=x$ for some $%
t=1...T-1$, while $q_t$ is different from the identity function on the
neighborhood of $x$. With $T=2$, $\mathcal{X}_0$ is simply the boundary of
the set of crossing points $\{x: F_{X_1}(x)=F_{X_2}(x)\in(0,1)\}$. We refer
to Figure \ref{fig:Xronde} for an illustration.

\medskip
\begin{figure}[]
\begin{center}
\setlength{\unitlength}{0.5cm}
\begin{picture}(22,10)
  \linethickness{0.1mm}
	\multiput(1,1)(1,0){1}{\vector(0,1){10}}
	\multiput(1,1)(1,0){1}{\vector(1,0){20}}
	\qbezier(1,1.7)(3,2.1)(5,3)
	\qbezier(1,1.5)(3,2.1)(5,3)
	\qbezier(5,3)(7,3.8)(10,6)
	\qbezier(10,6)(14,9.4)(19,9.5)
	\qbezier(10,6)(14,8.9)(19,9.1)
	\qbezier(5,2.4)(5,4)(5,1)
	\qbezier(10,6)(10,6)(10,1)
	\put(19.1,9.5){\footnotesize$F_{X_1}$}
	\put(19.1,8.8){\footnotesize$F_{X_2}$}
	\put(4.8,0.5){\footnotesize$x \in \mathcal{X}_0$}
	\put(9.8,0.5){\footnotesize$x' \in \mathcal{X}_0$}
	\put(7.1,0.5){\footnotesize$x'' \not\in \mathcal{X}_0$}
\end{picture}
\end{center}
\caption{Example of points belonging or not to $\mathcal{X}_{0} =
\{x,x^{\prime }\}$}
\label{fig:Xronde}
\end{figure}
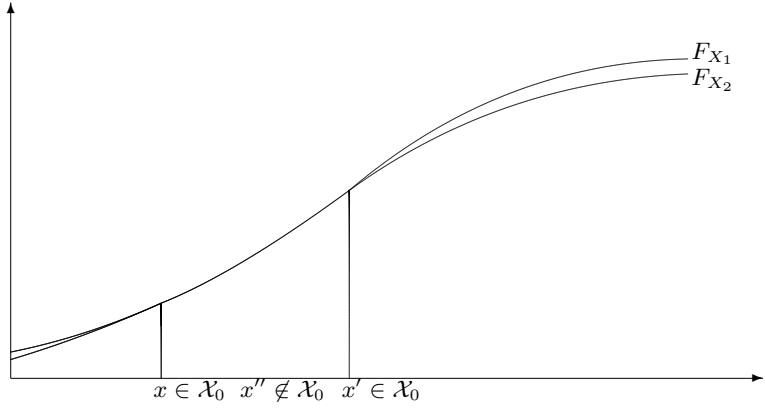

To make the preceding identification argument of marginal effects rigorous,
the following technical conditions are also required.

\begin{assumption}
(Additional regularity conditions) For all $x_{0}\in \mathcal{X}_{0}$, there
exists a neighborhood $\mathcal{V}$ of $x_{0}$ such that: \\[2mm]
(i) The map $x\mapsto U_{T}(x)$ is differentiable on $\mathcal{V}$ (almost
surely), and there exists a random variable $A$ such that $x\mapsto
E[A|X_{T}=x]$ is continuous on $\mathcal{V}$, and for all $(x,x^{\prime
})\in \mathcal{V}^{2}$,
\begin{equation*}
\left\vert \frac{\partial U_{T}}{\partial x}(x^{\prime })-\frac{\partial
U_{T}}{\partial x}(x^{\prime })\right\vert \leq A|x^{\prime }-x|.
\end{equation*}%
(ii) For all $x^{\prime }\in \mathcal{V}$, $x\mapsto E\left[ \partial
U_{T}/\partial x(x^{\prime })|X_{T}=x\right] $ is continuous on $\mathcal{V}$%
. \\[2mm]
(iii) For all $x\in \mathcal{V}$, $x^{\prime }\mapsto F_{U_{T}(x^{\prime
})|X_{T}}^{-1}(p|x)$ is differentiable at $x_{0}$. Moreover,
\begin{equation*}
(x,x^{\prime })\mapsto \lim_{h\rightarrow 0}\frac{F_{U_{T}(x^{\prime
}+h)|X_{T}}^{-1}(p|x)-F_{U_{T}(x^{\prime })|X_{T}}^{-1}(p|x)}{h}
\end{equation*}%
is continuous on $\mathcal{V}^{2}$. \label{hyp:technique}
\end{assumption}

\begin{theorem}
Under Assumptions \ref{hyp:stationarity}- \ref{hyp:technique}, $\Delta
^{AME}(x)$ and $\Delta ^{QME}(p,x)$ are identified, for all $x\in \mathcal{X}%
_{0}$. \label{thm:point_ident_ME}
\end{theorem}


\section{Proofs}

\subsection{Theorem \protect\ref{thm:point_ident_ME}}

Consider a sequence $(x_n)_{n\in\mathbb{N}}$ converging to $x\in\mathcal{X}%
_0 $ and such that $q_{t}(x_{n}) \neq x_{n}$. Let us assume without loss of
generality that $x_n\in \mathcal{V}$ for all $n\in \mathbb{N}$. Given that $%
q_t$ is continuous and $q_t(x)=x$, we can also assume without loss of
generality that $q_t(x_n)\in \mathcal{V}$ for all $n\in \mathbb{N}$.

\medskip Now, by the mean value theorem, there exists a random variable $%
\widetilde{X}_n$ between $x_n$ and $q_{t}(x_{n})$ such that
\begin{equation*}
\frac{U_T(q_t(x_n))-U_T(x_n)}{q_t(x_n)-x_n}=\frac{\partial U_T}{\partial x}(%
\widetilde{X}_n).
\end{equation*}
Hence,
\begin{align}
\frac{\Delta^{ATT}(x_n,q_t(x_n))}{q_{t}(x_n)-x_{n}} = & E\left(\frac{%
\partial U_T}{\partial x}(x)\bigg|X_T=x\right) + \left[E\left(\frac{\partial
U_T}{\partial x}(x)\bigg|X_T=x_n\right) - E\left(\frac{\partial U_T}{%
\partial x}(x)\bigg|X_T=x\right)\right]  \notag \\
& + E\left(\frac{\partial U_T}{\partial x}(\widetilde{X}_n) - \frac{\partial
U_T}{\partial x}(x) \bigg|X_T=x_n\right).  \label{eq:decomp_for_AME}
\end{align}
The term into brackets tends to zero by Assumption \ref{hyp:technique}-(ii).
Moreover, by Assumption \ref{hyp:technique}-(i),
\begin{equation*}
\left|E\left(\frac{\partial U_T}{\partial x}(\widetilde{X}_n) - \frac{%
\partial U_T}{\partial x}(x) \bigg|X_T=x_n\right)\right| \leq
\max\left(|x_n-x|,|q_t(x_n)-x|\right)\sup_{n \in \mathbb{N}} E[A|X_T=x_n].
\end{equation*}
Given that $x \mapsto E[A|X_T=x]$ is continuous, the supremum on the
right-hand side is finite. Therefore, this right-hand side tends to zero.
Hence, in view of \eqref{eq:decomp_for_AME},
\begin{equation*}
\lim_{n\rightarrow\infty} \frac{\Delta^{ATT}(x_n,q_t(x_n))}{q_{t}(x_n)-x_{n}}
= \Delta^{AME}(x),
\end{equation*}
and this latter is identified by Theorem \ref{thm:point_ident}.

\medskip Let us turn to $\Delta^{QME}(p,x)$. By the mean value theorem,
there exists a random variable $\widetilde{X}^{\prime }_n$ between $x_{n}$
and $q_t(x_{n})$ such that
\begin{eqnarray*}
\frac{\Delta^{QTT}(p,x_n,q_t(x_n))}{q_t(x_{n})-x_{n}} &=& \frac{%
F_{U_T(q_t(x_n))|X_T}^{-1}(p|x_n)-F_{U_T(x_n)|X_T}^{-1} (p|x_n)}{%
q_t(x_{n})-x_{n}} \\
& =&\frac{\partial F^{-1}_{U_T(x^{\prime })|X_T}(p|x_n)}{\partial x^{\prime }%
}{}_{|x^{\prime }=\widetilde{X}^{\prime }_n}.
\end{eqnarray*}
By Assumption \ref{hyp:technique}-(iii), the last derivative converges to
\begin{equation*}
\frac{\partial F^{-1}_{U_T(x^{\prime })|X_T}(p|x)}{\partial x^{\prime }}{}_
{|x^{\prime }=x}\; = \Delta^{QME}(p,x).
\end{equation*}
The result follows as above. \qed

\subsection{Theorem \protect\ref{thm:bounds}}

Suppose first that $U_T$ is locally concave on $[\min(x,\underline{x}%
_T(x^{\prime })),\overline{x}_T(x^{\prime })]$. Then, for all $x_1\leq
x^{\prime }\leq x_2$, almost surely,
\begin{equation}
\frac{U_T(x_2) - U_T(x) }{x_2-x} \leq \frac{g(x^{\prime },U_T) - U_T(x)}{%
x^{\prime }-x} \leq \frac{U_T(x_1) - U_T(x) }{x_1-x}.
\label{eq:ineg_concave}
\end{equation}

Taking $x_1 = \underline{x}_T(x^{\prime })$ and $x_2 = \overline{x}%
_T(x^{\prime })$, and integrating conditional on $X_T=x$, we obtain
\begin{equation*}
(x^{\prime }-x)\frac{\Delta^{ATT}(x,\overline{x}_T(x^{\prime }))}{\overline{x%
}_T(x^{\prime }) - x} \leq \Delta^{ATT}(x,x^{\prime }) \leq (x^{\prime }-x)%
\frac{\Delta^{ATT}(x,\underline{x}_T(x^{\prime }))}{\underline{x}%
_T(x^{\prime }) - x}.
\end{equation*}
The inequality is simply reverted if $g$ is locally convex. Hence, in either
case,
\begin{eqnarray*}
& & (x^{\prime }-x)\min\left\{\frac{\Delta^{ATT}(x,\underline{x}_T(x^{\prime
}))}{\underline{x}_T(x^{\prime })-x}, \frac{\Delta^{ATT}(x,\overline{x}%
_T(x^{\prime }))}{\overline{x}_T(x^{\prime })-x}\right\} \leq
\Delta^{ATT}(x,x^{\prime }) \\
& \leq & (x^{\prime }-x)\max\left\{\frac{\Delta^{ATT}(x,\underline{x}%
_T(x^{\prime }))}{\underline{x}_T(x^{\prime })-x}, \frac{\Delta^{ATT}(x,%
\overline{x}_T(x^{\prime }))}{\overline{x}_T(x^{\prime })-x}\right\}.
\end{eqnarray*}

The reasoning is the same for marginal effects using, instead of Equation %
\eqref{eq:ineg_concave},
\begin{equation*}
\frac{U_T(x_2) - U_T(x) }{x_2-x} \leq \frac{\partial U_T}{\partial x}(x)
\leq \frac{U_T(x_1) - U_T(x) }{x_1-x}.
\end{equation*}
\qed

\subsection{Theorem \protect\ref{thm:convergence_rate_att}}

\label{sec:convergence_rate_att}

Before showing the result, we state and prove a series of lemmas.

\begin{lemma}[Consistency of $\widehat{x}^*_1$]
\label{lemma:consistency_xstar} If Assumptions \ref{hyp:stationarity}, \ref%
{a:iid} and \ref{a:asymptotics_xstar}-(i) hold, then $\widehat{x}^*_1 -
x^*_1 = o_p(1)$.
\end{lemma}

\begin{proof}
Let $M_n(x)= - |\Psi_n(x)|$ and
$M(x) = -\left|F_{X_2}(x)-F_{X_1}(x)\right|$ and let $I= [F_{X_2}^{-1}(\underline{p})-\eps,F_{X_2}^{-1}(\overline{p})+\eps]$ for some  $\eps >0$. By Assumption \ref{a:asymptotics_xstar}-(i), $x^*_1$ is the unique maximum of $M$ on $I$. Besides, by Glivenko-Cantelli's theorem,
\begin{align*}
	\norm{ M_n - M }_\infty & \leq \norm{\Psi_n(x) - (F_{X_2}(x)-F_{X_1}(x))}_\infty \\
	& \leq \norm{ \widehat{F}_{X_2} - F_{X_2} }_\infty + \norm{ \widehat{F}_{X_1} - F_{X_1} }_\infty \\
	& \stackrel{p}{\longrightarrow} 0.	
\end{align*}
Fix $\eta >0$ and let $B = \left\{x \in I: |x - x^*_1|\geq  \eta\right\}$. Because $B$ is compact and $M$ is continuous, $\sup_{x\in B} M(x)= \max_{x\in B} M(x) < M(x^*_1)$. We have
\begin{equation}
	\label{eq:conv_sup_Mn}
\sup_{x\in B} M_n(x) \leq \norm{ M_n - M }_\infty + \sup_{x\in B} M(x) \stackrel{p}{\longrightarrow} \sup_{x\in B} M(x) < M(x^*_1).	
\end{equation}
Suppose that $\widehat{x}^*_1 \in B$ and $x^*_1 \in [\widehat{F}_{X_2}^{-1}(\underline{p}),\widehat{F}_{X_2}^{-1}(\overline{p})]$. Then
$\sup_{x\in B} M_n(x) = M_n(\widehat{x}^*_1) \geq M_n(x^*_1)$. Hence,
$$P\left(\widehat{x}^*_1 \in B, x^*_1 \in [\widehat{F}_{X_2}^{-1}(\underline{p}),\widehat{F}_{X_2}^{-1}(\overline{p})]\right) \leq P\left(\sup_{x\in B} M_n(x) - M_n(x^*_1) \geq 0\right),$$
but the latter probability tends to zero in view of \eqref{eq:conv_sup_Mn}. Now, remark that $x^*_1\in (F_{X_2}^{-1}(\underline{p}), F_{X_2}^{-1}(\overline{p}))$, so that with a probability approaching one, $x^*_1 \in [\widehat{F}_{X_2}^{-1}(\underline{p}),\widehat{F}_{X_2}^{-1}(\overline{p})]$. With probability approaching one, we also have $ [\widehat{F}_{X_2}^{-1}(\underline{p}),\widehat{F}_{X_2}^{-1}(\overline{p})] \subset I$, so that $\widehat{x}^*_1\in I$ with probability approaching one. Hence, $P(|\widehat{x}^*_1-x^*_1|<\eta) \stackrel{p}{\longrightarrow} 0$.
\end{proof}

\begin{lemma}[Convergence Rate of $\widehat{x}^*_1$]
\label{lemma:convergence_rate_xstar} If Assumptions \ref{hyp:stationarity}, %
\ref{a:iid} and \ref{a:asymptotics_xstar} hold, then $\sqrt{n}\left(\widehat{%
x}^*_1 - x^*_1\right) = O_p(1)$.
\end{lemma}

\begin{proof}
	Let $\psi_x(u,v)= \mathds{1}\{u\leq x\} - \mathds{1}\{v\leq x\}$ and $\Psi(x)=E(\psi_x(X_2,X_1))$. Because the set of functions $(\mathds{1}\{. \leq x\})_x$ is Donsker and by the conservation properties of Donsker classes, $\mathcal{F}_\delta= \left\{\psi_x: |x-x^*_1|<\delta\right\}$ is Donsker for any $\delta>0$. Moreover, by independence between $X_1$ and $X_2$,
	\begin{align*}
		E\left(\psi_x(X_2,X_1) - \psi_{x^*_1}(X_2,X_1)\right)^2 & =
			F_{X_2}(x) + F_{X_1}(x) - 2 F_{X_2}(x) F_{X_1}(x) + F_{X_2}(x^*_1) + F_{X_1}(x^*_1) \\
	& - 2 F_{X_2}(x^*_1) F_{X_1}(x^*_1) \\
	& - 2 \left(F_{X_2}(x \wedge x^*_1) + F_{X_1}(x \wedge x^*_1) - 2 F_{X_2}(x \wedge x^*_1) F_{X_1}(x \wedge x^*_1)\right).		
	\end{align*}
Therefore, by continuity of $F_{X_1}$ and $F_{X_2}$,
$$E\left[\left(\psi_x(X_2,X_1) - \psi_{x^*_1}(X_2,X_1)\right)^2\right] \rightarrow 0 \; \text {as } x \rightarrow x^*_1$$
This and Lemma \ref{lemma:consistency_xstar} above imply \citep[see, e.g.,][Lemma 19.24]{vandervaart98} that
\begin{equation}
	\sqrt{n}\left[\left(\Psi_n(\widehat{x}^*_1) - \Psi(\widehat{x}^*_1)\right) - \left(\Psi_n(x^*_1) - \Psi(x^*_1)\right)\right] = o_P(1).
	\label{eq:ep}	
\end{equation}
Besides, $\Psi(x^*_1) = 0$ and by the central limit theorem, $\Psi_n(x^*_1)=O_p(1/\sqrt{n})$.
Moreover, with probability approaching one, $|\Psi_n(\widehat{x}^*_1)|\leq |\Psi_n(x^*_1)|$, implying $\Psi_n(\widehat{x}^*_1)=O_p(1/\sqrt{n})$. Combined with \eqref{eq:ep}, this yields
\begin{align}
\sqrt{n}\left[\Psi(\widehat{x}^*_1) - \Psi(x^*_1)\right] & = -\sqrt{n}\left[\Psi_n(\widehat{x}^*_1)-\Psi_n(x^*_1)\right] + o_p(1) \nonumber \\ 	
& = O_p(1). \label{eq:approx_psi}
\end{align}
By Assumption \ref{a:asymptotics_xstar}-(ii) and because $\widehat{x}^*_1$ is consistent by Lemma \ref{lemma:consistency_xstar}, we have, with probability approaching one, $\left|\Psi(\widehat{x}^*_1) - \Psi(x^*_1)\right| \geq
C^R \left|\widehat{x}^*_1 - x^*_1\right|$. This and \eqref{eq:approx_psi} yields the desired result.
\end{proof}

In the following, we let $\mathcal{D}$ denote the sets of c\`adl\`ag
functions on $\mathcal{Y}$. We also let $\mathcal{C}^1$ denote the subset of
$\mathcal{D}$ of continuously differentiable functions, with positive
derivative.

\begin{lemma}[Hadamard differentiability of two useful maps]
\label{lemma:had_diff} The map $Q: (F_1,F_2) \mapsto F_1^{-1} \circ F_2(x)$
is Hadamard differentiable, tangentially to the set of continuous functions,
at any $(F_{10}, F_{20}) \in \mathcal{D}^2$ such that $F_{10}$ is
differentiable at $F_{10}^{-1}\circ F_{20}(x)$, with positive derivative at
this point. The map $R: (F_1,F_2,F_3) \mapsto F_1 \circ F_2^{-1} \circ F_3$
is also Hadamard differentiable at any $(F_{10},F_{20},F_{30}) \in \mathcal{C%
}^1 \times \mathcal{C}^1\times \mathcal{D}$ continuously differentiable
functions tangentially to the set of continuous functions.
\end{lemma}

\begin{proof}
Let $Q_1:(F_1,F_2) \mapsto (F_1, F_2(x))$ and $Q_2:(F,p)\mapsto F^{-1}(p)$, so that $Q = Q_2 \circ Q_1$. The map $Q_1$ is linear and continuous, and therefore Hadamard differentiable at any $(F_{10}, F_{20}) \in \mathcal{D}^2$. Let us prove that $Q_2$ is Hadamard differentiable at any $(F_0,p) \in \mathcal{D} \times (0,1)$ such that $F_0$ is differentiable at $F_0^{-1}(p)$, with a corresponding positive derivative. We have to show that for any $h_u$ converging uniformly to $h$ continuous and $p_u \rightarrow p$, $\lim_{u\rightarrow 0} [(F_0+uh_u)^{-1}(p_u) - F_0^{-1}(p)]$ exists. By differentiability of $F_0^{-1}$ at $p$, this is the case if $\lim_{u\rightarrow 0} [(F_0+uh_u)^{-1}(p_u) - F_0^{-1}(p_u)]$ exists. Now, an inspection of the proof of Lemma 21.3 of \cite{vandervaart98} reveals that it still applies if we replace $p$ by $p_u$, with $p_u\rightarrow p$. Hence, $Q_2$ is Hadamard differentiable tangentially to the set of continuous functions at $(F_0, p)$. By applying the chain rule \citep[see][Theorem 20.9]{vandervaart98}, $Q$ is Hadamard differentiable at any $(F_{10}, F_{20}) \in \mathcal{D}^2$ such that $F_{10}$ is differentiable at $F_{10}^{-1}\circ F_{20}(x)$, with positive derivative at this point. The result for $R$ is proved in \citeauthor{Chaisemartin14}  (2018, see the proof of Lemma S5).
\end{proof}

\begin{lemma}[Convergence rate of $\widehat{q}_1(x)$]
\label{lemma:convergence_rate_q_1_x} Suppose that Assumption \ref{a:iid}
holds and $F_{X_1}$ is differentiable at $q_1(x)$ with $F^{\prime
}_{X_1}(q_1(x))>0$. Then, $\widehat{q}_1(x)-q_1(x) = O_P(1/\sqrt{n})$.
\end{lemma}

\begin{proof}
We have $q_1(x) = F_{X_1}^{-1} \circ F_{X_2}(x)$ and $\widehat{q}_1 (x)= \widehat{F}_{X_1}^{-1} \circ \widehat{F}_{X_2}(x)$. By the standard Donsker's theorem (see, e.g., \citep[see, e.g.,][Theorem 19.3]{vandervaart98},
$$\sqrt{n}\left( \widehat{F}_{X_1} - F_{X_1}, \widehat{F}_{X_2} - F_{X_2}\right) \convind (G_1 \circ F_{X_1}, G_2 \circ F_{X_2}),$$
where $G_1$ and $G_2$ are two independent standard Brownian bridges. Because $F'_{X_1}(q_1(x))>0$, Lemma \ref{lemma:had_diff} and the functional delta method
\citep[see, e.g.][Lemma 3.9.4]{vandervaart96} ensure that $\sqrt{n}\left(\widehat{q}_1 (x)-q_1(x)\right)$ is asymptotically normal. The result follows.
\end{proof}

In the following, we let $w_t(y,x) =F_{Y_t|X_t}(y|x)f_{X_t}(x)$ for $t \in
\{1,2\}$. Let us also denote by $\widehat{f}_{X_t}$ the kernel density
estimator of $f_{X_t}$ and $\widehat{w}_t(y,x) = \widehat{F}_{Y_t|X_t}(y|x)%
\widehat{f}_{X_t}(x)$.

\begin{lemma}[Behavior of some nonparametric estimators]
\label{lemma:behavior_num} Suppose that Assumptions \ref{a:iid} and \ref{a:Y}%
-\ref{a:kernel_bandwidth} hold. Then, for any closed and bounded interval $%
V\subset \mathcal{X}$ and $t \in \{1,2\}$,
\begin{equation*}
\sqrt{nh_n}\left \|E\left[\widehat{w}_t(.,x)\right]/E\left[\widehat{f}%
_{X_t}(x)\right] - F_{Y_t|X_t}(.,x)\right \|_\infty \longrightarrow 0,
\end{equation*}
\begin{equation*}
\sup_{x \in V} \left \|\partial_x \widehat{w}_t(.,x)\right \|_\infty =
O_P(1).
\end{equation*}
\end{lemma}

\begin{proof}
First, because $K(y)\geq 0$, $E\left[\widehat{w}_t(y,x)\right]/E\left[\widehat{f}_{X_t}(x)\right]\leq 1$ for all $y$. Thus,
\begin{align}
& \norm{E\left[\widehat{w}_t(.,x)\right]/E\left[\widehat{f}_{X_t}(x)\right] - F_{Y_t|X_t}(.,x)}_\infty \nonumber \\
\leq & \frac{1}{f_{X_t}(x)} \left[\norm{E\left[\widehat{w}_t(.,x)\right] - w(.,x)}_\infty + \left|E\left[\widehat{f}_{X_t}(x)\right] - f_{X_t}(x)\right|\right].
\label{eq:decomp_bias}
\end{align}
We have
$$E\left[\widehat{f}_{X_t}(x)\right] - f_{X_t}(x) = \int K(u) \left[f_{X_t}(x+h_n u)-f_{X_t}(x)\right]du.$$
Thus, because $|f'_{X_t}|$ is bounded,
$$\sqrt{nh_n}\left|E\left[\widehat{f}_{X_t}(x)\right] - f_{X_t}(x)\right| \leq C \sqrt{nh_n^5}\int |t|K(u)du,$$
for some $C >0$. Hence, the left-hand side tends to zero by Assumption \ref{a:kernel_bandwidth}-(i). Now consider the first term of \eqref{eq:decomp_bias}. A change of variable yields
$$E\left[\widehat{w}_t(y,x)\right] - w(y,x) = \int K(u)\left[w(y,x-h_n u) - w(y,x)\right]du.$$
By a  second-order Taylor expansion, we obtain
$$E\left[\widehat{w}_t(y,x)\right] - w(y,x) = \int K(u)\left[-h_n u \partial_x w(y,x) + \frac{1}{2}(h_n u)^2 \partial_{xx} w(y,\widetilde{x}_1) \right]du,$$
where $\widetilde{x}_1 \in (x, x+h_n u)$.  As a result, by Assumption \ref{a:Y}-(ii) and \ref{a:kernel_bandwidth}-(ii),
$$\norm{E\left[\widehat{w}_t(.,x)\right] - w(.,x)}_\infty \leq C' h_n ^2,$$
for some $C'>0$. By Assumption \ref{a:kernel_bandwidth}-(i) once more, the first term of \eqref{eq:decomp_bias} tends to zero, which yields the first result of the lemma.

\medskip
To obtain the second result, first observe that by the triangular inequality,
\begin{align}
	\sup_{x\in V} \norm{\partial_x \widehat{w}_t(.,x)}_\infty \leq &
	\sup_{x\in V} \norm{\partial_x \widehat{w}_t(.,x) -  E\left[\widehat{w}_t(.,x)\right]}_\infty \nonumber \\
	& + \sup_{x\in V} \norm{E\left[\widehat{w}_t(.,x)\right]- \partial_x w(.,x)}_\infty + \sup_{x\in V} \norm{\partial_x w(.,x)}_\infty.	\label{ineq:wprime}
\end{align}
By Assumption \ref{a:Y}-(iii) $\sup_{x\in V} \norm{\partial_x w(.,x)}_\infty<\infty$. Therefore, to show the result, it suffices to show that the two first terms of the right-hand side of \eqref{ineq:wprime} tend to zero in probability.

\medskip
To analyse the first term, let us remark that
\begin{align*}
& nh^2_n\left(\partial_x \widehat{w}_t(y,x) -E[\partial_x \widehat{w}_t(y,x)]\right) \\
= & \sum_{i=1}^n \mathds{1}\{Y_{it}\leq y\}K'\left(\frac{x-X_{it}}{h_n}\right) - n E\left[
\mathds{1}\{Y_{it}\leq y\}K'\left(\frac{x-X_{it}}{h_n}\right)\right].	
\end{align*}
Thus, the left-hand side corresponds to $W(x,f)$ in \cite{einmahl2000}, with $f(u)=\mathds{1}\{y\leq u\}$ and $K'$ in place of $K$. Moreover, $f_{X_t,Y_t}$ is continuous, $f_{X_t}$ is continuous and $\inf_{x\in V} f_{X_t}(x)>0$ and $K'$ satisfies their (K)-(i) and (K)-(ii). Finally, remark that Proposition 1 of \cite{einmahl2000} does not rely on their condition (K)-(iii). Hence, with probability one,
$$\lim\sup_{n\rightarrow \infty}\sqrt{\frac{nh_n^3}{2|\log(h_n)|}}\sup_{x \in V}\norm{\partial_x \widehat{w}_t(.,x) -E[\partial_x \widehat{w}_t(.,x)]}_\infty <\infty.$$
Because $nh_n^3/|\log(h_n)| \rightarrow \infty$ by Assumption \ref{a:kernel_bandwidth}, $\sup_{x \in V}\norm{\partial_x \widehat{w}_t(.,x) -E[\partial_x \widehat{w}_t(.,x)]}_\infty \rightarrow 0$.

\medskip
Now let us turn to the second term of \eqref{ineq:wprime}. First, remark that
$$E\left[\partial_x \widehat{w}_t(y,x)\right]= \frac{1}{h_n^2} \int w(y,x) K'\left(\frac{x-u}{h_n}\right)du.$$
Integrating by part and using the facts that $f_{X_1}$ is bounded above and $K(u)\rightarrow 0$ as $|u|\rightarrow \infty$, we obtain
$$E\left[\partial_x \widehat{w}_t(y,x)\right]=  \int K(u) \partial_x w(y,x+ h_n u)du.$$
By Assumption \ref{a:Y}-(iii), there exists a constant $C'>0$ such that for all $y$ and $x\in V$, $|\partial_x w(y,x+h_n u) - \partial_x w(y,x+h_n u)| \leq C' h_n |u|$. Hence,
$$\sup_{x\in V} \norm{E\left[\partial_x \widehat{w}_t(.,x)\right] - \partial_x w(.,x)}_\infty \leq  C' h_n\int |u| K(u) du,$$
and the left-hand side tends to zero.\end{proof}

\begin{lemma}[Negligible effect of estimating covariates]
\label{lemma:no_effect_x} Suppose that $x \in \mathcal{X}$ and $\widehat{x}$
satisfies $\widehat{x}-x=O_P(1/\sqrt{n})$. If Assumptions \ref{a:iid} and %
\ref{a:Y}-\ref{a:kernel_bandwidth} hold, then, for $t \in \{1,2\}$,
\begin{equation*}
\sqrt{nh_n} \left \|\widehat{F}_{Y_t|X_t}(.|\widehat{x}) - \widehat{F}%
_{Y_t|X_t}(.|x)\right \|_\infty \overset{P}{\longrightarrow} 0.
\end{equation*}
\end{lemma}

\begin{proof}
Let us denote by $\widehat{f}_{X_t}$ the kernel density estimator of $f_{X_t}$ and $\widehat{w}_t(y|x) = \widehat{F}_{Y_t|X_t}(y|x)\widehat{f}_{X_t}(x)$. With a large probability, $\widehat{x}\in V$. Then, using the fact that $\widehat{F}_{Y_t|X_t}\leq 1$,
\begin{align*}
	& \norm{\widehat{F}_{Y_t|X_t}(.|\widehat{x}) - \widehat{F}_{Y_t|X_t}(.|x)}_\infty \\
	\leq & \frac{1}{\inf_{x'\in V} \widehat{f}_{X_t}(x')}\left[\norm{\widehat{w}_t(.|\widehat{x})-\widehat{w}_t(.|x)}_\infty + \left|\widehat{f}_{X_t}(\widehat{x}) - \widehat{f}_{X_t}(x)\right| \right] \\
	\leq & \frac{1}{\inf_{x'\in V} \widehat{f}_{X_t}(x')}\left[\sup_{x' \in V} \norm{\partial_x \widehat{w}_t(.|x')}_{\infty}| + \sup_{x'\in V} \left|\widehat{f}'_{X_t}(x')\right| \right] \left|\widehat{x}-x\right|.
\end{align*}
Now, $f_{X_t}$ and $f'_{X_t}$ are uniformly continuous on $V$. By Assumption \ref{a:kernel_bandwidth}, $h_n\rightarrow 0$ and
\linebreak $nh_n^3/|\log(h_n)| \rightarrow\infty$. Moreover, $K$ satisfies the conditions of Theorem A and C of \cite{silverman1978}. $K'$ may not satisfy condition (C2) of \cite{silverman1978}, but this condition is not needed for the necessity part of his Theorem 3 that we use here.  Therefore, $\widehat{f}_{X_t}$ and $\widehat{f}'_{X_t}$ are uniformly consistent on $V$. The result follows by $\widehat{x}-x=O_P(1/\sqrt{n})$, $h_n\rightarrow 0$ and Lemma \ref{lemma:behavior_num}.
 \end{proof}

\begin{lemma}[Asymptotic distribution of $\widehat{F}%
_{Y_{2}|X_{2}}(.|x_{1}^{*})$]
\label{lemma:conv_F} If Assumptions \ref{a:iid} and \ref{a:Y}-\ref%
{a:kernel_bandwidth} hold, then, for $t\in \{1,2\}$,
\begin{align*}
\sqrt{nh_{n}}& \left( \widehat{F}_{Y_{2}|X_{2}}(.|x)-F_{Y_{2}|X_{2}}(.|x),%
\widehat{F}_{Y_{1}|X_{1}}(.|q_{1}(x)-F_{Y_{1}|X_{1}}(.|q_{1}(x)),\widehat{F}%
_{Y_{2}|X_{2}}(.|x_{2}^{\ast })-F_{Y_{2}|X_{2}}(.|x_{1}^{\ast }),\right. \\
& \left. \widehat{F}_{Y_{1}|X_{1}}(.|x_{1}^{\ast
})-F_{Y_{1}|X_{1}}(.|x_{1}^{\ast })\right) \overset{d}{\longrightarrow }%
\mathbb{G},
\end{align*}%
where $\mathbb{G}$ is a continuous Gaussian processes.
\end{lemma}

\begin{proof}
First, by Lemma \ref{lemma:behavior_num}, we have, for any $x\in \mathcal{X}$, $$\norm{E\left[\widehat{F}_{Y_2|X_2}(.|x)\right] - \widehat{F}_{Y_2|X_2}(.|x)}_\infty \leq C |h_n|,$$
for some $C>0$. Hence, we may focus on the process $\mathbb{G}_n = \sqrt{nh_n} \left(\widehat{F}_{Y_2|X_2}(.|x)-E\left[\widehat{F}_{Y_2|X_2}(.|x)\right]\right)$. The proof readily extends to the multivariate process by the Cram\'er-Wold device. Note that convergence of the process follows if (i) for any $k\in \mathbb{N}$ and $(y_1,..,y_k)\in \mathcal{Y}^k$,  $(\mathbb{G}_n(y_1),...,\mathbb{G}_n(y_k))$ is asymptotically normal and (ii) $\mathbb{G}_n$ is asymptotically tight \citep[see, e.g.,][Theorem 18.14]{vandervaart98}, Theorem 18.14). (i) follows by the Cram\'er-Wold device, asymptotic normality of the Nadaraya-Watson estimator and Assumptions \ref{a:Y}-\ref{a:kernel_bandwidth} \citep[see, e.g.,][]{bierens1987}.

\medskip
Now, let us prove (ii). By Theorem 1.1 of \cite{Einmahl97}, the process \newline
$\sqrt{nh_n}\left(\widehat{w}_2(.,x) - E[\widehat{w}_2(.,x) ]\right)$ is asymptotically tight. Now, remark that
\begin{align*}
\mathbb{G}_n  = \frac{1}{f_{X_2(x)}} \bigg[ &\sqrt{nh_n}\left(\widehat{w}_2(.,x) - w_2(.,x)\right) + F_{Y_2|X_2}(.|x)\sqrt{nh_n}\left(\widehat{f}_{X_2}(x) - f_{X_2}(x)\right) \\
& \left. + \left(\widehat{F}_{Y_2|X_2}(.|x) - F_{Y_2|X_2}(.|x)\right)\sqrt{nh_n}\left(\widehat{f}_{X_2}(x) - f_{X_2}(x)\right)\right].
\end{align*}
By Assumption \ref{a:kernel_bandwidth}, $K$ is defined on a compact set and has bounded variation. Theorem 1 of \citeauthor{stute1986} (1986, see also his remark p.893) then ensures that $\widehat{F}_{Y_2|X_2}(.|x)$ is a uniformly consistent estimator of $F_{Y_2|X_2}(.|x)$. Hence, the supremum norm of the third term in the brackets converges to zero in probability. The second term is asymptotically tight since $\sqrt{nh_n}\left(\widehat{f}_{X_2}(x) - f_{X_2}(x)\right)=O_P(1)$ and $F_{Y_2|X_2}(.|x)$ is uniformly continuous on $\mathcal{Y}$. Hence, $\mathbb{G}_n $ is asymptotically tight, and the result follows.
 \end{proof}

We now prove the theorem. Let $H(y)=F_{Y_{1}|X_{1}}\left(
F_{Y_{2}|X_{2}}^{-1}(F_{Y_{1}|X_{1}}(y|x_{1}^{\ast })|x_{1}^{\ast
})|q_{1}(x)\right) $ and
\begin{equation*}
\widehat{H}(y)=\widehat{F}_{Y_{1}|X_{1}}\left( \widehat{F}%
_{Y_{2}|X_{2}}^{-1}(\widehat{F}_{Y_{1}|X_{1}}(y|\widehat{x}_{1}^{\ast })|%
\widehat{x}_{1}^{\ast })|\widehat{q}_{1}(x)\right).
\end{equation*}
It is easy to see that $H$ is the cumulative distribution function of $%
g_{1}(Y_{1}) $ conditional on $X_{1}=q_{1}(x)$. Lemmas \ref%
{lemma:no_effect_x} and \ref{lemma:conv_F} imply that
\begin{equation*}
\left( \widehat{F}_{Y_{2}|X_{2}}(.|x),\widehat{F}_{Y_{1}|X_{1}}(.|\widehat{q}%
_{1}(x)),\widehat{F}_{Y_{2}|X_{2}}(.|\widehat{x}_{1}^{\ast }),\widehat{F}%
_{Y_{1}|X_{1}}(.|\widehat{x}_{1}^{\ast })\right)
\end{equation*}%
converges to a continuous Gaussian process. By Lemma \ref{lemma:had_diff}
and the functional delta method, $\left( \widehat{F}_{Y_{2}|X_{2}}(.|x),%
\widehat{H}\right) $ also converges to a continuous Gaussian process at the
rate $\sqrt{nh_{n}}$.

\medskip Now, by integration by parts for Lebesgue-Stieljes integrals,
\begin{equation*}
\Delta^{ATT}(x,q_1(x)) = \int_{\underline{y}}^{\overline{y}}
F_{Y_2|X_2}(y|x) - H(y) dy.
\end{equation*}
The map $\varphi: (F_1,F_2) \mapsto \int_{\underline{y}}^{\overline{y}}
[F_1(y) - F_2(y)]dy$, defined on the set of bounded c\`adl\`ag functions, is
linear and also continuous with respect to the supremum norm. It is
therefore Hadamard differentiable. Because $\widehat{\Delta}^{ATT}(x,q_1(x))
= \varphi\left(\widehat{F}_{Y_2|X_2}(.|x), \widehat{H}\right)$, it is
asymptotically normal at the rate $\sqrt{nh_n}$. 

\medskip Finally, we have $\Delta^{QTT}(p, x,q_1(x)) = H^{-1}(p) -
F^{-1}_{Y_2|X_2}(p|x)$ and $\widehat{\Delta}^{QTT}(p, x,q_1(x)) = \widehat{H}%
^{-1}(p) - \widehat{F}^{-1}_{Y_2|X_2}(p|x)$. Because the quantile function
is Hadamard differentiable \citep[see, e.g.,][Lemma 21.3]{vandervaart98},
the map $(F_1, F_2) \mapsto F^{-1}_1(p) - F^{-1}_2(p)$ is Hadamard
differentiable at any $(F_{10},F_{20}) $ such that $F_{10}$ and $F_{20}$ are
differentiable at $F^{-1}_{10}(p)$ and $F^{-1}_{20}(p)$ respectively, with
positive corresponding derivatives. The result follows by applying the
functional delta method once more. \qed

\section{Additional Details on the Application}

\label{sec:details_app}

We first present the piecewise linear estimator of $q_1(x)-x$. Equivalently,
we impose such a parametric restriction on $q_1^{-1}(x)$. In line with the
theoretical design of the policy, we consider the specification
\begin{equation}
q_1^{-1}(x)= x+ \zeta_0 (x-12.0)^+ + \zeta_1 (x-15.2)^+ + \zeta_2 (x-23.4)^+
- (\zeta_0+\zeta_1+\zeta_2)(x-26.8)^+,  \label{eq:param_q1}
\end{equation}
where $x^+=\max(0,x)$. The values 12 and 26.8 correspond to the theoretical
limits outside which we should not observe any difference between the 1987
and 1989 income. The values 15.2 and 23.4 are the theoretical limits inside
which the difference between the two incomes should be maximal -- see the
right panel of Figure \ref{fig:eitc_plus_income}. The last term in %
\eqref{eq:param_q1} ensures that $q_1^{-1}(x)=q_1(x)=x$ when $x\geq 26.8$.
We estimate $(\zeta_0,\zeta_1,\zeta_2)$ by minimizing $\int_{12.0}^{26.8} (
q_1^{-1}(x) - \hat q_1^{-1}(x))^2 dx$, where $\hat q_1^{-1} = \hat
F_{X_2}^{-1} \circ \hat F_{X_1}$.

\medskip The estimate of $q_1$ appears in Figure \ref{fig:piecewise_q1}. The
estimator is close to the nonparametric estimator, but also to the
theoretical function implied by the policy design, displayed in the left
panel of Figure \ref{fig:eitc_plus_income}.

\begin{figure}[]
\centering
\begin{minipage}{0.48\textwidth}
	\includegraphics[scale=0.22]{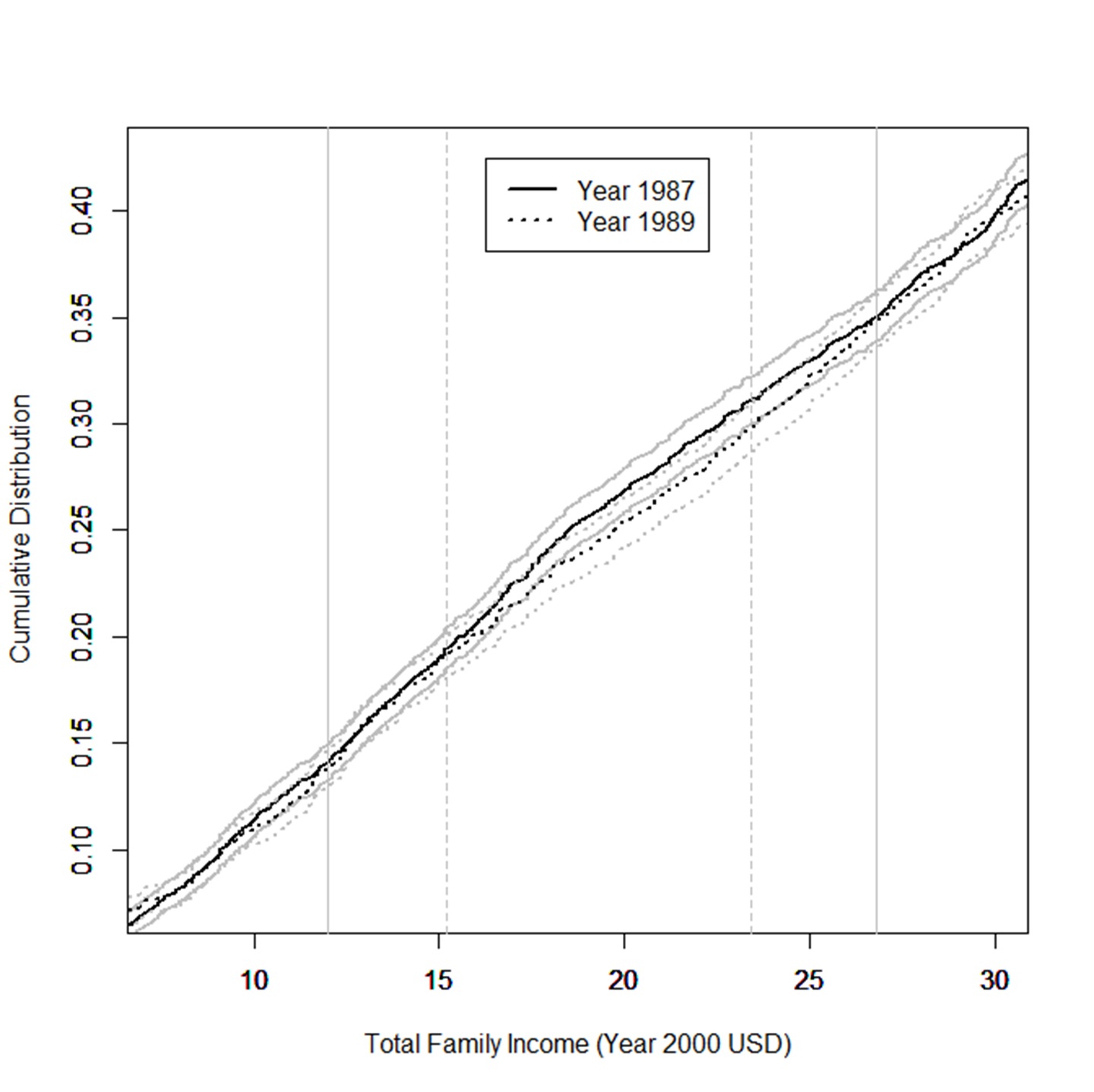}
\end{minipage}
\begin{minipage}{0.48\textwidth}
	\includegraphics[scale=0.22]{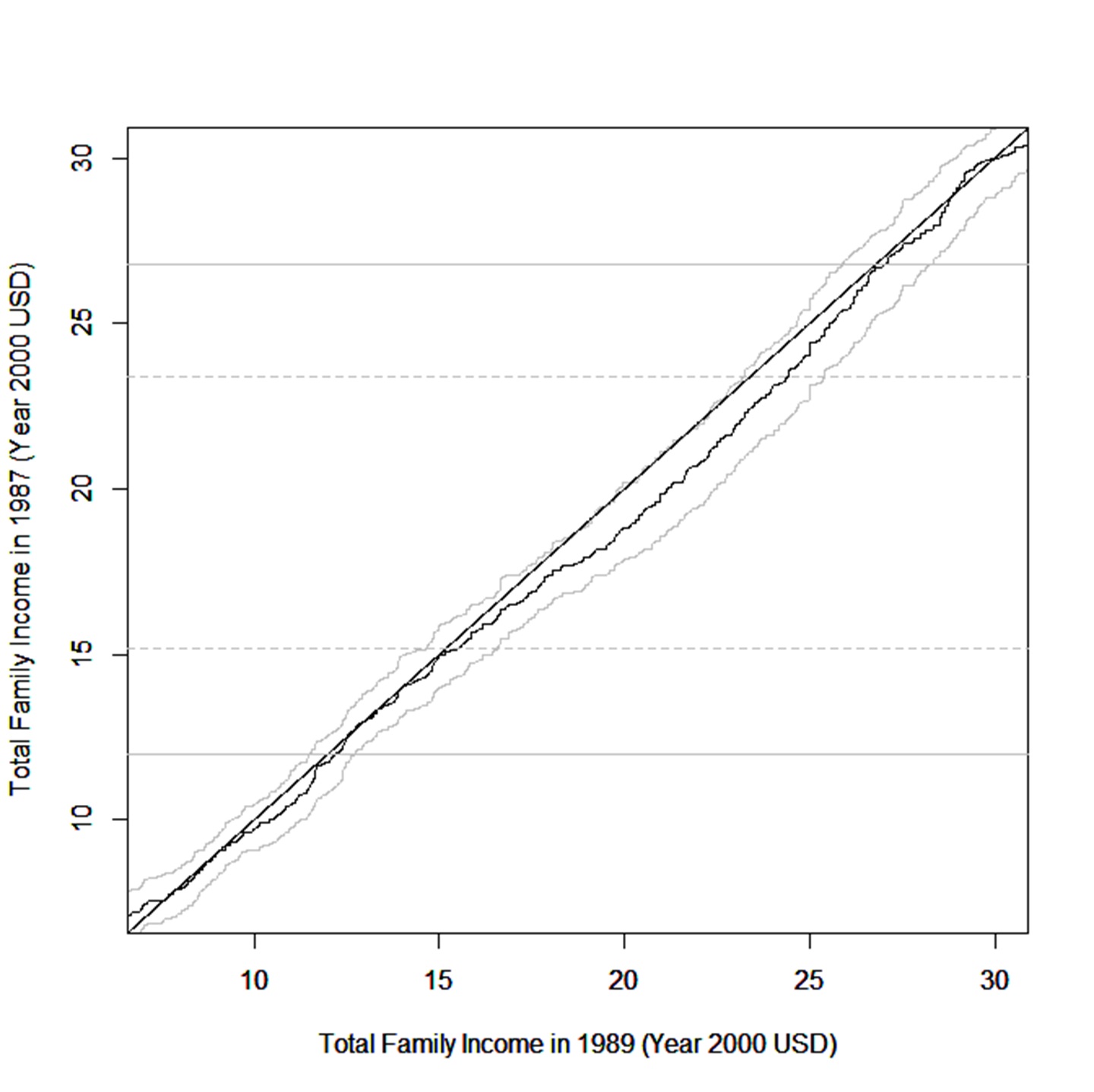} 
\end{minipage}	
\begin{minipage}{0.96\textwidth}
{\footnotesize ~\\[2mm]
Notes: CEX data restricted to families with two or more children. Family income is given in thousands of 2000 US dollars. In the left panel, the black (resp. grey) curve corresponds to family income in 1987 (resp. 1989). In the right panel, we display the Q-Q plot, i.e. $q_1$ against the identity function. The solid lines indicate the theoretical limits inside which we should observe a divergence of the cdfs, given the policy design. The dotted lines are the limits of the interval on which the efffect of the policy is supposed to be maximal (see the right panel of Figure \ref{fig:eitc_plus_income}).}
\end{minipage}
\caption{Cdf's and Q-Q plot of total family income in 1987 and 1989 with
bootstrap 5--95 percentiles.}
\label{fig:empirical_cdfs_with_ci}
\end{figure}

\begin{figure}[]
\begin{center}
\includegraphics[scale=0.22]{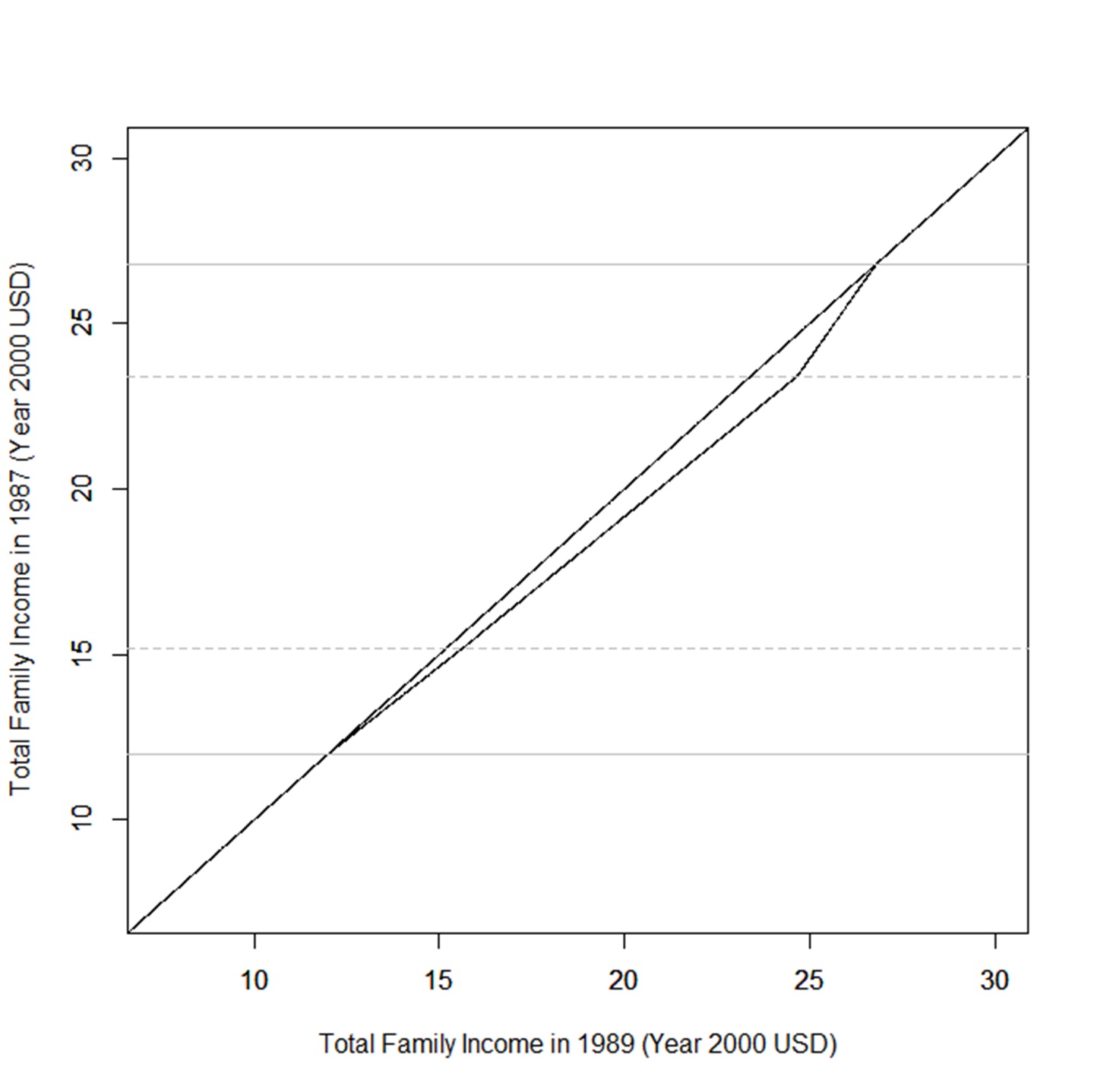}
\includegraphics[scale=0.22]{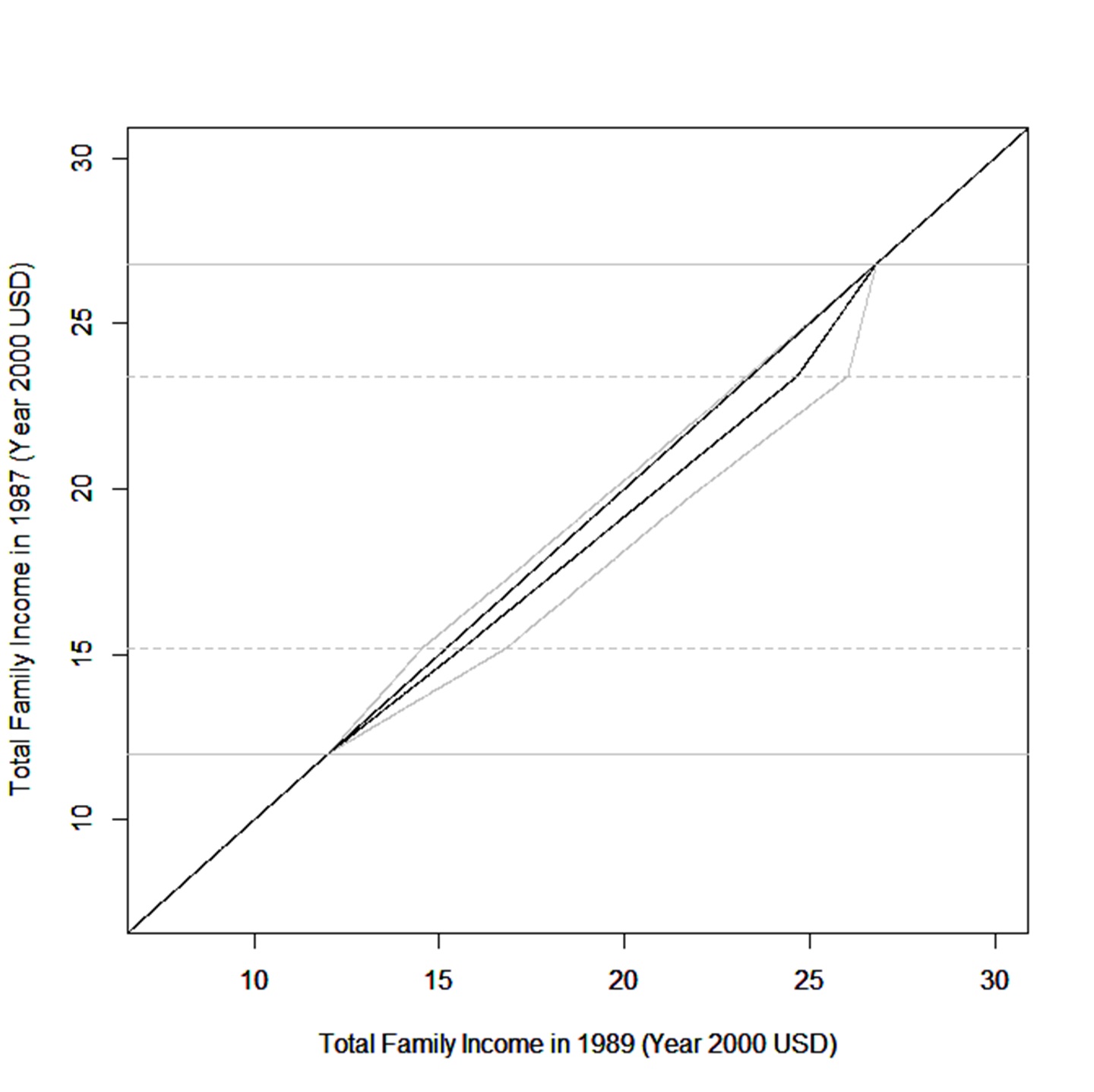} ~\vspace{-1cm}
\end{center}
\caption{Piecewise linear estimator of $q_1$. The right panel shows
bootstrap 5--95 percentiles.}
\label{fig:piecewise_q1}
\end{figure}

\begin{figure}[]
\begin{center}
\includegraphics[scale=0.22]{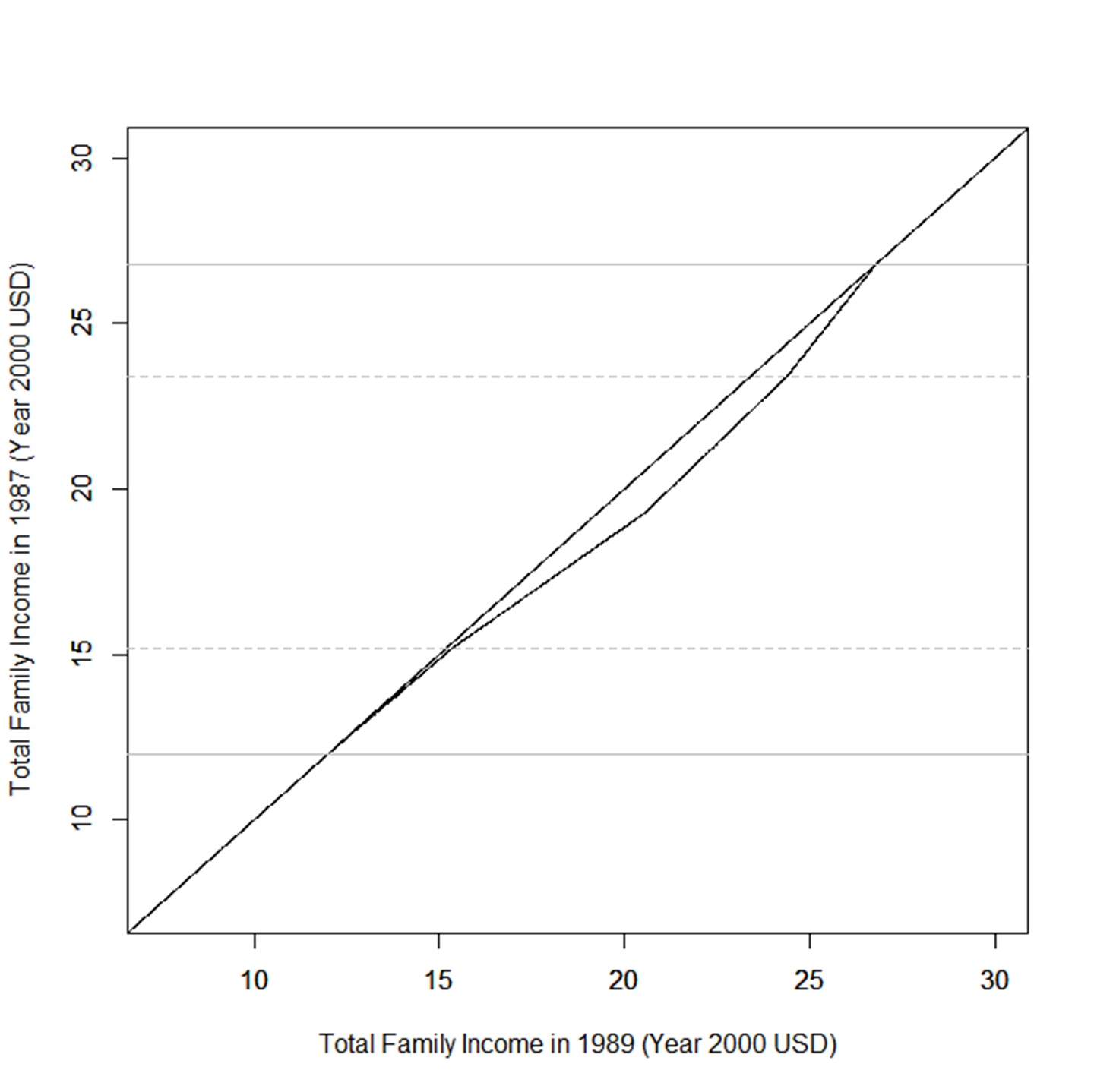} %
\includegraphics[scale=0.22]{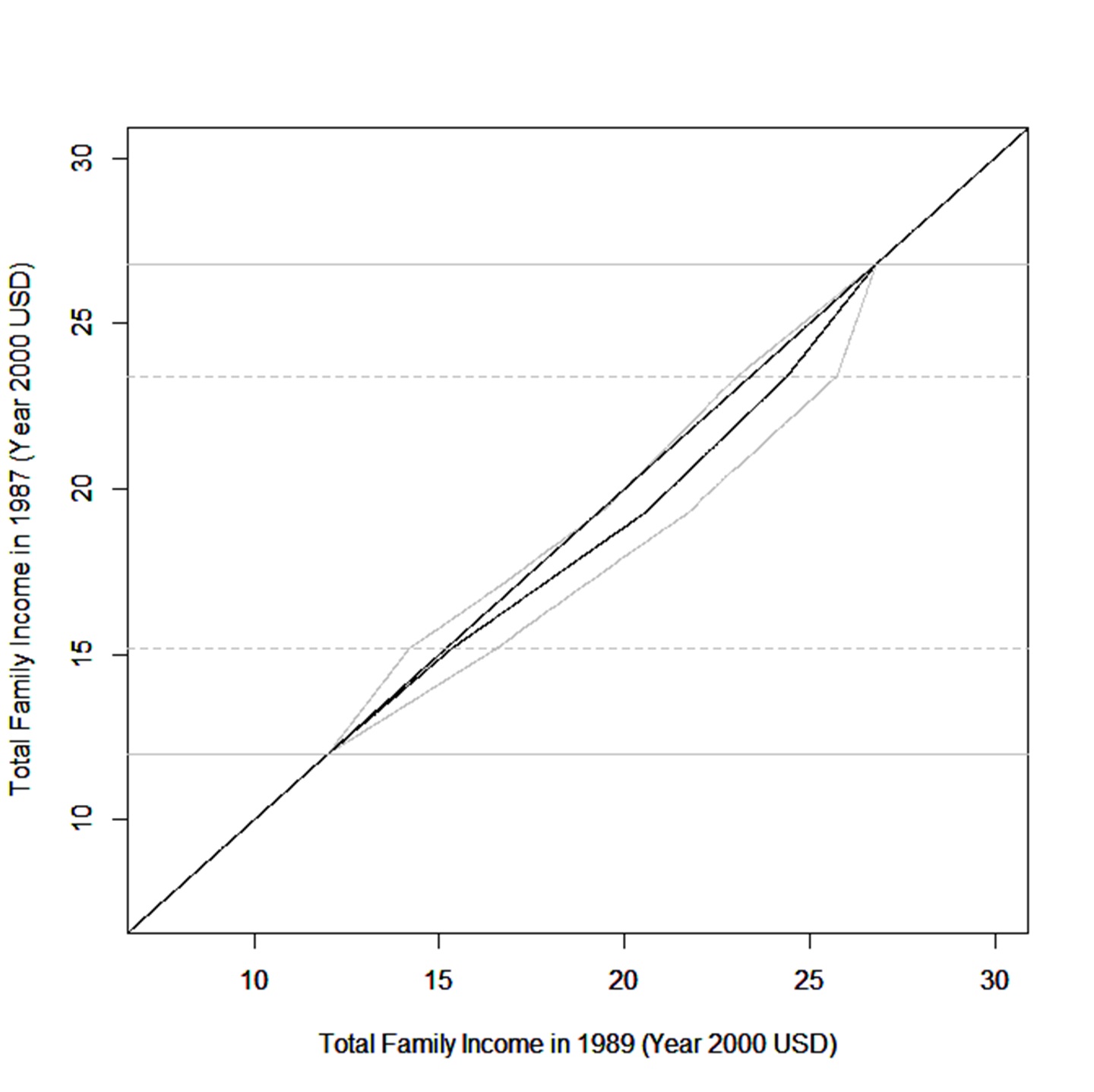} ~\vspace{-1cm}
\end{center}
\caption{Piecewise linear estimator of $q_1$ with an additional knot. The
right panel shows bootstrap 5--95 percentiles.}
\label{fig:piecewise_q1_additional_knot}
\end{figure}

\end{document}